\documentclass[notitlepage,aps,pra,superscriptaddress]{revtex4-1}
\newif\ifarxiv\arxivtrue
\arxivfalse

\usepackage[margin=2.5cm]{geometry}

\usepackage{amsmath}
\usepackage{amssymb}
\usepackage{graphicx}
\usepackage{hyperref}
\usepackage{color}

\usepackage{xcolor}

\definecolor{ForestGreen}{RGB}{34,139,34}

\hypersetup{
	colorlinks,
	linkcolor={black!30!blue},
	citecolor={red},
	urlcolor={blue}
}

\usepackage[T1]{fontenc}
\usepackage[latin9]{inputenc}

\usepackage{amsthm}
\usepackage{bbm}
\usepackage{mathtools}
\mathtoolsset{centercolon}
\ifarxiv\fi

\usepackage{xspace}

\usepackage{pgfplots}
\usepgfplotslibrary{patchplots}
\usepackage{tikz}
\usetikzlibrary{arrows}
\usetikzlibrary{decorations.pathreplacing}
\usetikzlibrary{decorations.markings}
\usetikzlibrary{shapes,arrows,positioning,patterns}
\usetikzlibrary{matrix}
\usetikzlibrary{fit}
\usetikzlibrary{fadings}

\def\eat#1{}

\newtheorem{thm}{Theorem}[section]
\newtheorem{defi}[thm]{Definition}
\newtheorem{lem}[thm]{Lemma}
\newtheorem{prop}[thm]{Proposition}
\newtheorem{cor}[thm]{Corollary}

\def\Qref#1{\splitref#1@}
\def\splitref#1:#2@{\thing#1@\;\ref{#1:#2}}
\def\thing#1#2@{\ifx#1sSect.\else\ifx#1fFig.\else\ifx#1TTable\else\ifx#1tThm.\else\ifx#1dDef.\else\ifx#1lLem.\else\ifx#1cCor.\else\ifx#1pProp.\fi\fi\fi\fi\fi\fi\fi\fi}

\def\tr{\operatorname{tr}}
\def\idty{{\leavevmode\rm 1\mkern -5.4mu I}} 
\def\id{{\rm id}}
\def\Rl{{\mathbb R}}\def\Cx{{\mathbb C}}
\def\Ir{{\mathbb Z}}\def\Nl{{\mathbb N}}

\let\eps\varepsilon

\def\norm #1{\Vert #1\Vert}

\def\mod{{\mathop{\rm mod}\nolimits}}
\def\ind{{\mathop{\rm ind}\nolimits}\,}

\def\braket#1#2{\langle #1,#2\rangle}
\def\brAket#1#2{\langle #1\vert#2\rangle}

\def\ket #1{\vert#1\rangle}
\def\ketbra #1#2{{\vert#1\rangle\langle#2\vert}}
\def\kettbra#1{\ketbra{#1}{#1}}

\def\tr{\mathop{\rm tr}\nolimits}
\def\abs#1{\vert#1\vert}
\def\rank{{\mathop{\rm rank}\nolimits}\,}

\let\veps\varepsilon

\def\inv{^{-1}}

\def\re{\Re e}
\def\im{\Im m}
\def\IM{\im}

\def\calk#1{\lceil#1\rceil}

\def\dff{\bf} 

\def\AA{{\mathcal A}}{{\def\HH{{\mathcal H}}\def\KK{{\mathcal K}}

\makeatletter
\newcommand{\raisemath}[1]{\mathpalette{\raisem@th{#1}}}
\newcommand{\raisem@th}[3]{\raisebox{#1}{$#2#3$}}
\makeatother

\def\ig{\mathbf I} 
\def\six{\mathop{\mathrm s\mkern-1mu\mathrm i}\nolimits}
\DeclareRobustCommand{\sixnp}{\mathop{\mathrm s\mkern-1mu \textup \i}\nolimits}
\DeclareRobustCommand{\sixR}{\mathop{\hbox{\raisebox{.45em}{$\scriptstyle\rightharpoonup$}}\hspace{-.9em}\sixnp}}
\def\sixL{\mathop{\hbox{\raisebox{.45em}{$\scriptstyle\leftharpoonup$}}\hspace{-.9em}\sixnp}}
\def\tsix{{\widetilde\six}}
\def\sixrel#1:#2{\six(#1{:}#2)}
\def\sixmat#1{\left(\begin{array}{cc} \six_+(#1_L)&\six_+(#1_R)\\ \six_-(#1_L)&\six_-(#1_R)\end{array}\right)}%

\def\romkern{\kern-1pt}
\def\rom#1{\textnormal{I\if#11\else\romkern I\if#12\else\romkern I\fi\fi}}
\def\symA{\textnormal{\textsf A}\xspace}
\def\symD{\textnormal{\textsf D}\xspace}
\def\symC{\textnormal{\textsf C}\xspace}
\def\symAI{\textnormal{\textsf A}\rom1\xspace}
\def\symAII{\textnormal{\textsf A}\rom2\xspace}
\def\symAIII{\textnormal{\textsf A}\rom3\xspace}
\def\symBDI{\textnormal{\textsf{BD}}\rom1\xspace}
\def\symCI{\textnormal{\textsf C}\rom1\xspace}
\def\symCII{\textnormal{\textsf C}\rom2\xspace}
\def\symDIII{\textnormal{\textsf D}\rom3\xspace}
\def\symS{\textnormal{\textsf S}\xspace}
\def\igS{\ig(\symS)}   
\def\igtiS{\ig_{\rm ti}(\symS)}

\def\indF{\ind_{\mathrm F}}
\def\indFR{\hbox{$\ind\hspace{-1.2em}\raisebox{.45em}{$\scriptstyle\rightharpoonup$}\hspace{.5em}$}}
\def\indFL{\hbox{$\ind\hspace{-1.2em}\raisebox{.45em}{$\scriptstyle\leftharpoonup$}\hspace{.5em}$}}

\def\Pl#1{P_{<#1}}
\def\Pg#1{P_{\geq#1}}

\ifarxiv\else\fi

\def\ph{\eta}\def\rv{\tau}\def\ch{\gamma}
\def\tph{{\widetilde\ph}}\def\trv{{\widetilde\rv}}\def\tch{{\widetilde\ch}}
\def\trh{{\widetilde\rho}}

\def\Wh{\widehat W} 

\def\wind{{\mathop{\rm wind}}}

\def\Vcan{V_{\rm can}}
\def\pfaff{{\mathop{\rm pf}\nolimits}}

\def\s{\strut\quad\strut}

\newcommand{\pf}{\operatorname{pf}}

\usepackage{hyperref}
\begin{document}

\title{The topological classification of one-dimensional symmetric quantum walks
}
\nopagebreak

\author{C. Cedzich}
\affiliation{Institut f\"ur Theoretische Physik, Leibniz Universit\"at Hannover, Appelstr. 2, 30167 Hannover, Germany}
\author{T. Geib}
\affiliation{Institut f\"ur Theoretische Physik, Leibniz Universit\"at Hannover, Appelstr. 2, 30167 Hannover, Germany}
\author{F.~A. Gr\"unbaum}
\affiliation{Department of Mathematics, University of California, Berkeley CA 94720}
\author{C. Stahl}
\affiliation{Institut f\"ur Theoretische Physik, Leibniz Universit\"at Hannover, Appelstr. 2, 30167 Hannover, Germany}
\author{L. Vel\'azquez}
\affiliation{Departamento de Matem\'{a}tica Aplicada \& IUMA,  Universidad de Zaragoza,  Mar\'{\i}a de Luna 3, 50018 Zaragoza, Spain}
\author{A.~H. Werner}
\affiliation{{QMATH}, Department of Mathematical Sciences, University of Copenhagen, Universitetsparken 5, 2100 Copenhagen, Denmark,}
\affiliation{{NBIA}, Niels Bohr Institute, University of Copenhagen, Denmark}
\author{R.~F. Werner}
\affiliation{Institut f\"ur Theoretische Physik, Leibniz Universit\"at Hannover, Appelstr. 2, 30167 Hannover, Germany}

\begin{abstract}
We give a topological classification of quantum walks on an infinite 1D lattice, which obey one of the discrete symmetry groups of the tenfold way, have a gap around some eigenvalues at symmetry protected points, and satisfy a mild locality condition. No translation invariance is assumed. The classification is parameterized by three indices, taking values in a group, which is either trivial, the group of integers, or the group of integers modulo 2, depending on the type of symmetry. The classification is complete in the sense that two walks have the same indices if and only if they can be connected by a norm continuous path along which all the mentioned properties remain valid. Of the three indices, two are related to the asymptotic behaviour far to the right and far to the left, respectively. These are also stable under compact perturbations. The third index is sensitive to those compact perturbations which cannot be contracted to a trivial one.
The results apply to the Hamiltonian case as well. In this case all compact perturbations can be contracted, so the third index is not defined. Our classification extends the one known in the translation invariant case, where the asymptotic right and left indices add up to zero, and the third one vanishes, leaving effectively only one independent index. When two translation invariant bulks with distinct indices are joined, the left and right asymptotic indices of the joined walk are thereby fixed, and there must be eigenvalues at $1$ or $-1$ (bulk-boundary correspondence). Their location is governed by the third index.  We also discuss how the theory applies to finite lattices, with suitable homogeneity assumptions.
\end{abstract}


\maketitle


\section{Introduction}
The classification of quantum lattice systems according to ``topological phases'' is currently an area of intensive research \cite{kitaevPeriodic,KitaevLectureNotes,MPSphaseI,MPSphaseII,Schnyder1,Schnyder2,HasanKaneReview,KaneMeleQSH,KaneMeleTopOrder,ZhangTopologicalReview}. A basic observation, called bulk-boundary correspondence,  is that this classification becomes experimentally tangible when two regions are joined along an interface: when dynamical laws of the two regions belong to different classes, new ``topologically protected'' modes appear along the interface, which may be absent if the laws are different, but in the same class. The typical setting is that of free Fermions, characterized by their effective one-particle description,  with additional discrete symmetries, and a spectral gap condition. The groups whose elements are expected to label the classes are known in a large variety of dimensions and symmetry classes \cite{kitaevPeriodic}. In the translation invariant case this amounts to a classification task for vector bundles over the quasi-momentum space. However,  no clear picture of the general non-translation invariant setting is emerging from the heuristic literature. There is, however, a growing interest in the Mathematical Physics literature \cite{SchulzBaldesBook,SchulzZ2,Schulz2016index,SchulzNonTechnicalOverview,GawedzkIIndex,Thiang,Graf}, so the situation is improving.

Unfortunately, this literature does not cover a closely related kind of system known as quantum walks \cite{Grimmet,electric,QDapproach,TRcoin}. These are simply the discrete time analog of Hamiltonian systems, with the dynamics given by a unitary ``one-step'' evolution operator $W$. Physically, walks are realized by periodically driven systems \cite{Karski:2009,Gensketal,Schreiber:2010cl,TopSilberhorn,GawedzkIndex}, for which $W$ is the evolution operator after one driving period.  Hence these systems are also known as Floquet systems. The classification problem for these systems can be posed in analogy to the Hamiltonian case.  An indication that this might require more than a simple translation comes from the case without symmetries and gap conditions. For Hamiltonian systems this classification is trivial. For walks, however, there is an integer valued homotopy invariant \cite{OldIndex}. For a translation invariant system this is the total winding number of the energy bands on a torus whose coordinates are the quasi-energy and the quasi-momentum. Another observation is that in the walk case we have to be more careful in specifying what kind of perturbations we consider. The implicit claim of the phrase ``topologically protected mode'' is that such modes will appear independently of how the crossover between the two bulk phases is designed. This would suggest a stability of these modes under arbitrary local perturbations in the interface region. In the Hamiltonian case any two such crossover designs are continuously connected, so stability against local perturbations is implied by stability against continuous deformations. This is not true in the walk case, and indeed a main theme of our work is to explore the consequences of the existence of local perturbations which cannot be achieved by a sequence of norm-small modifications.

Throughout this work we will consider walks on a one-dimensional doubly infinite lattice. Within this confined setting we have attempted to go for the maximum generality in which the basic questions make sense and can be answered naturally. No translation invariance whatsoever is assumed. Indeed such an assumption would make it impossible to discuss the joining of two different bulk phases. In contrast to much of the literature, where only walks with strictly finite maximal jump length are considered, we allow matrix elements of the walk operator to decrease rather slowly with distance. In the translation invariant case, where decay of matrix elements translates to smoothness in momentum space, our condition turns out to be equivalent to mere continuity \cite{UsOnTI}. That is, continuous band structures with non-differentiable kinks are allowed. Throughout, we assume a spectral gap, but in order to discuss the protected eigenvalues appearing in the gap, we relax this condition to a gap only in the essential spectrum. For the discrete symmetries we followed the literature to restrict consideration to the so-called tenfold way \cite{Altland-Zirnbauer}. In many works in the literature this leads to a proliferation of case distinctions. We have tried to find the concepts which allow a uniform treatment of the ten symmetry types with as few case distinctions as possible. The main ingredient for this is an apparently new way to arrive at the classifying index groups. This is a completely elementary group theoretical construction, which does not require any  K-theory as used in \cite{kitaevPeriodic,Thiang,SchulzBaldesBook,SchulzZ2,Schulz2016index,SchulzNonTechnicalOverview}. This makes the paper self-contained and, hopefully, accessible to a wider audience. An announcement of some basic results was given in a letter \cite{letter}.

The main results of our paper are the following: for any walk $W$ in our setting we define three indices denoted \hbox{$\sixL(W)$}, \hbox{$\sixR(W)$}, and \hbox{$\six_-(W)$}, which are elements of the index group belonging to the symmetry type. These characterize the walk up to homotopy. That is, these indices coincide for two walks if and only if they can be connected by a norm continuous path of walks, each of which satisfies the assumptions of our setting. $\sixL(W)$ and $\sixR(W)$ are also stable under arbitrary compact (in particular, local) perturbations. $\sixL(W)$ can be computed if the behaviour of $W$ is only known far to the left (resp.\ $\sixR(W)$ far to the right). These indices can also be defined in the Hamiltonian case, which is covered by our theory as well. Indeed, we confirm and generalize the heuristic claims made for the Hamiltonian case. However, the third index, which is a special feature of the walk case,  seems to have been largely missed so far. There were some indications of additional invariants, because in the walk case one has two gaps rather than one \cite{Asbo1,Asbo2}. However, the homotopy-stable combination furnished by $\six_-(W)$ was not identified. This index gives a complete classification of general compact perturbations modulo contractible ones. That is, a compact perturbation can be contracted to the identity in the set of compact perturbations iff it leaves this index unchanged.  One of our tools, needed especially for establishing the completeness of the index invariants, is the statement that (apart from certain trivial cases) all walks can be decoupled gently, i.e., continuously deformed to a walk in which the left half and the right half of the system do not interact.

Our paper is organized as follows: in the remainder of this introductory chapter we provide a detailed description of the setting we choose to work in and give an overview of the results obtained. In \Qref{sec:groups}, after introducing the notion of ``symmetry'' and  ``symmetry type'' in our setting, we define the crucial assumption of ``essential gap'' around spectral points which are invariant under the symmetries. This assumption allows us to define the symmetry indices. The index groups are computed by elementary group theoretical methods and formulas for calculating the invariants are provided. In \Qref{sec:twosettings} we explain the important distinction between the symmetry indices ($\sixR$, $\sixL$) stable under both compact and contractible perturbations, treated in \Qref{sec:indices}, and the index which is invariant only under homotopy ($\six_-$), which is treated in \Qref{sec:homotopy-indices}. These two kinds of indices require rather different methods, and even different natural settings. The translation invariant examples fit into the first category and are treated in \Qref{sec:examples}. This is needed to discuss the bulk-boundary phenomenon (\Qref{sec:bulkboundary}).
In \Qref{sec:locpert} we provide the classification of compact non-contractible perturbations. The existence of decouplings, i.e., deformations of a walk into another one in which right and left half do not interact, is studied in \Qref{sec:decoup}. By an explicit construction we show that the indices given are complete, i.e., that walks with equal indices can be deformed into each other (\Qref{sec:complete}). Finally, in \Qref{sec:finite} we explain how, even though our theory assigns zero indices to finite systems, one may extract non-trivial results for this case of prime physical interest.

\subsection{Setting}\label{sec:firstsetting}
We state here the complete set of assumptions  of our theory.
Where the notions involved are defined only later in the paper we refer to the appropriate definition or section.

The  Hilbert space for the quantum systems under consideration has the form
\begin{equation}\label{Hcells}
\HH=\bigoplus_{x=-\infty}^\infty\HH_x,
\end{equation}
where each $\HH_x$ is of finite, non-vanishing dimension. The labels $x$ represent the position of a particle, and $\HH_x$ some internal degrees of freedom. We refer to $\HH_x$ as the {\dff cell} at $x$.
For every $a\in\Ir$, we denote by $\Pg a$ the projection onto the subspace
$\bigoplus_{x\geq a}\HH_x$ and, analogously, the complement $\Pl a=\idty-\Pg a$. We often abbreviate $P=\Pg0$.

We assume that certain (anti-)unitary operators on $\HH$ are given, which
represent a discrete set of {\dff symmetries}. The symmetry types considered
here are described in detail in
\Qref{sec:abstypes} \& \ref{sec:typeclassification}, and are combinations
of the so-called particle-hole, time-reversal, and chiral symmetries. We assume
that the symmetries act locally, i.e., each symmetry operator is the direct sum
of operators acting in each cell. The action in each cell is assumed to be
``balanced'' in the sense of  \Qref{def:balancedrep}.

We call a unitary operator $W$ {\bf admissible}, if it satisfies the following conditions:
\begin{itemize}
	\item[(1)]  $W$ satisfies a certain commutation relation with each of the
	symmetries, which are specified with the symmetry types in
	\Qref{sec:abstypes}.
	\item[(2)]  $W$ is {\dff essentially gapped}, i.e., in a small neighborhood of the points $+1$ and $-1$, $W$ has only discrete eigenvalues with finite total multiplicity (see \Qref{sec:gaps})
\end{itemize}

Note that so far we did not use the cell structure. This is brought in to formulate the locality condition which makes a unitary operator $W$ a {\bf walk}: the standard assumption in many earlier papers was that in each time step the system can only jump a finite distance $L$, in which case we call $W$ {\bf strictly local} for emphasis. In that case, for every $a$, the operator $\Pg a-W^*\Pg aW$ has non-zero matrix elements only between finitely many cells around $a$. However, our theory also works if this is only approximately true, and $\Pg a-W^*\Pg aW$ is merely assumed to be a compact operator for some $a$, in which case we call $W$ {\bf essentially local}. This will be the  standing assumption in the current paper. We note that $\Pg a-W^*\Pg aW$ is compact for {\it all} $a$ iff that is true for some $a$: for every $b$ the difference between $\Pg a-W^*\Pg aW$ and $\Pg b-W^*\Pg bW$ is the finite rank operator $\Pg a-\Pg b-W^*(\Pg a-\Pg b)W$.

We stress that in our general setting \emph{no translation invariance} is assumed. For translation invariance to make sense there must be a unitary operator $T$ with $T\HH_x=\HH_{x+1}$, which hence also serves to identify all the cells $\HH_x$, making the Hilbert space isomorphic to $\ell^2(\Ir)\otimes\HH_0$ with $T$ acting as the standard shift in the first factor. When discussing translation invariant systems we will assume that all symmetry operations and, of course, $W$ commute with $T$. In that case, condition (2) is equivalent to requiring that $W$ is strictly gapped, i.e., has $\pm1$ lie in the resolvent set. The translation invariant case is an important reference case, and has a well known classification, which will be described in \cite{UsOnTI}.

Our task is the classification of admissible walks, in such a way that the classes are closed under certain perturbations. We consider the following:

\begin{defi}\label{def:pertsorts}
	Let $W_1,W_2$ be admissible walks. Then we say that
	\begin{itemize}
		\item[(1)] $W_2$ is a {\dff gentle} perturbation of $W_1$ (or that $W_1$ and $W_2$ are {\dff homotopic}), if there is a norm-continuous function $t\mapsto W(t)$ on the unit interval with $W(0)=W_1$, and $W(1)=W_2$, such that each $W(t)$ is admissible and essentially local.
		\item[(2)] $W_2$ is a {\dff local} perturbation of $W_1$ if $W_2-W_1$ is non-zero only on finitely many of the spaces $\HH_x\subset\HH$.
		\item[(3)] $W_2$ is a {\dff finite rank} perturbation of $W_1$ if $W_2-W_1$ is an operator of finite rank.
		\item[(4)] $W_2$ is a {\dff compact} perturbation of $W_1$ if $W_2-W_1$ is a compact operator on $\HH$.
	\end{itemize}
\end{defi}

Clearly, a local perturbation is of finite rank, because the cells are finite dimensional. The converse is not true because a rank one perturbation might have components in infinitely many cells. Furthermore, finite rank perturbations are compact. It will be crucial that the implication (2)$\Rightarrow$(1) fails: there are non-gentle local perturbations. It is interesting to note that this distinction is not needed at all in the Hamiltonian case. All definitions in \Qref{def:pertsorts} directly make sense in the Hamiltonian case, too. But suppose that $H_2$ is a compact perturbation of $H_1$. Then since the admissibility conditions are $\Rl$-linear in $H$, all the Hermitian operators
$H_t=(1-t) H_1 +t H_2$ with $t\in\Rl$ are also admissible for the symmetry, and since they are all compact perturbations of $H_1$, the essential spectrum is the same for all $t$, so the essential gap remains open. Hence we have a continuous admissible connection from $H_1$ to $H_2$, and all compact perturbations are gentle. Since stability under gentle perturbations is easily achieved, the Hamiltonian case is much more straightforward than the case of walks.

\subsection{Overview of results}

In \Qref{sec:groups} we construct, for every symmetry type $\symS$, an abelian group denoted by
$\igS$ which is called the index group of the type. To each finite dimensional
representation $\rho$ of the symmetry type we associate an element
$\six(\rho)\in\igS$. The construction is completely elementary, yet the
groups match those obtained from the K-theoretical classification of vector
bundles in the translation invariant case \cite{kitaevPeriodic}.

For admissible $W$, the symmetry operators leave the eigenspaces at $\pm1$
invariant. Since these are finite dimensional by virtue of the essential gap condition, the symmetry indices of the symmetry representations in these subspaces, denoted by $\six_+(W)$ and $\six_-(W)$, are well defined.
They are invariant under gentle perturbations (\Qref{pro:homoto}), but not under general local ones. However, their sum $\six(W):=\six_-(W)+\six_+(W)$ is invariant even under all compact perturbations.
Therefore, under a non-gentle compact perturbation eigenvalues may be swapped between the eigenspaces at $+1$ and $-1$. This effect completely characterizes compact perturbations up to gentle ones, i.e., a compact perturbation is gentle if and only if it has the same $\six_-(W)$ (and consequently $\six_+(W)$). This theory works independently of the cell structure, and is described in \Qref{sec:homotopy-indices}.

Further invariants, which do depend on the cell-structure are described in \Qref{sec:indices}. It is here that the restriction to a one-dimensional lattice system enters. These invariants are obtained by splitting the system into two halves. It turns out that the most efficient way to do this is by temporarily suspending the unitarity condition, and admitting essentially unitary operators, which are defined by the property that $WW^*-\idty$ and $W^*W-\idty$ are both compact. Then there is a simple way to split the system, namely to consider the block-diagonal operator  $W'=PWP\oplus(\idty-P)W(\idty-P)=:W'_L\oplus W'_R$, which is essentially unitary by our essential locality condition. The index $\six(W)$ is easily seen to extend to essentially unitary operators, so we get two quantities
$\sixL(W):=\six(W'_L)$ and $\sixR(W):=\six(W'_R)$. Since $W'_R$ depends continuously on $W$, these are homotopy invariants (\Qref{thm:sixProps}). They are also stable under compact perturbations, and are tail properties in the sense that $\sixR(W)=\six(\Pg a W\Pg a)$ does not depend on $a$, and can be computed as far to the right as desired. The same statements hold for the continuous time, i.e., Hamiltonian case, with $\six(H)$, $\sixR(H)$ and $\sixL(H)$ defined via the $0$-eigenspaces of the respective Hermitian operators. In particular, the bulk-boundary principle holds. Since all compact perturbations are contractible, no further subtleties arise, and no further invariants need to be considered.

How the system is split does not matter in this definition. In an earlier version of the theory \cite{letter} we used instead a {\dff decoupling} of $W$, i.e., a compact strictly unitary perturbation  $W''$ of $W$ commuting with $P$. The existence of gentle decouplings is an interesting question in its own right, which is established in \Qref{sec:decoup}. It turns out that the decoupling process is {\it not} homotopy stable, i.e., it may happen that there is a continuous path of walks connecting $W_1$ and $W_2$, and gentle decouplings $W_1''$ and $W_2''$, which are not homotopic in the set of decoupled walks (see the discussion in \Qref{sec:twosettings}). This phenomenon is closely related to the instability of $\six_-(W)$ with respect to compact perturbations. The existence of gentle decouplings is also needed to establish the completeness of the index triple
$\bigl(\sixL(W),\sixR(W),\six_-(W)\bigr)$ for homotopy equivalence in \Qref{sec:complete}.

Since the indices $\sixL(W)$ and $\sixR(W)$ are tail properties, they have, strictly speaking, nothing to say  about finite systems. On the other hand, physical systems are finite, and so are the numerical simulations supporting the bulk-boundary correspondence. In the final \Qref{sec:finite} we describe a simple principle by which our theory nevertheless gives non-trivial results also in the finite case. Roughly speaking, this requires a notion of homogeneity for the ``bulk'', which comes with a typical length. Then when the bulk systems are large compared to this length the predicted eigenvalues, with eigenfunctions near the boundary, do occur, albeit only close to $\pm1$ and not exactly at these values.

\begin{table}\centering
	\begin{tabular}{|l|c|c|}
		\multicolumn{1}{|c|}{index}	& \s definition\s & \multicolumn{1}{|c|}{ \s page\s}  \\ \hline
		$\six(\rho)$	& \Qref{prop:six} & \pageref{prop:six}\\
		$\six_\pm(W)$, $\six(W)$\hspace{0.2cm} & \Qref{def:sixOp} & \pageref{def:sixOp}\\
		$\sixL(W)$, $\sixR(W)$ & \Qref{sec:indiceshamwalk} & \pageref{sec:indiceshamwalk}\\
		$\sixrel W':W$ & \Qref{def:relInd} & \pageref{def:relInd}\\
		$\ind(W)$ & \Qref{sec:decoup} & \pageref{sec:decoup}\\
	\end{tabular}
	\caption{\label{fig:sixtab}Index notations used in this paper, where $\rho$ stands for a finite dimensional representation of the symmetry type, and $W$ for a walk operator. }
\end{table}

\section{Group theoretical definition of the symmetry index}\label{sec:groups}

\subsection{Symmetry types}\label{sec:abstypes}
In this section we provide the basic analysis of the symmetries, and their impact on the structure of unitary or
Hermitian operators satisfying such symmetries. Since we later want to avoid boring case distinctions we describe the structure perhaps a bit more abstractly than absolutely needed, thus providing a language to treat all symmetry types under consideration, and perhaps a few more, in a uniform way. We begin with a compact description of the basic structure.

In every instance we investigate, the symmetries will be given by unitary or antiunitary operators and we are given an (essentially) unitary operator $W$ or a Hermitian operator $H$, which ``satisfies the symmetry''. The symmetries, their commutation relations between each other, their (anti-)unitary character and the commutation relations between the symmetries and the operator under investigation constitute a {\dff symmetry type}.

Rather than building the most general abstract structure of this description, let us be more specific. Every symmetry $\sigma$ under consideration will be an {\bf involution}, i.e., its action on operators ($X\mapsto \sigma X\sigma^*$)
squares to the identity. Thus, by Wigner's theorem \cite{Wignerbook}, $\sigma^2$ is a phase factor times the identity. The abstract group implemented by all the symmetries will either consist of  just the identity,  the identity and a single involution, or the Klein four-group, so that two involutions multiply to a phase factor times the third. Moreover, for each symmetry $\sigma$ we specify what it means that an operator $W$ or $H$ ``satisfies it'' or is {\dff admissible} for $\sigma$, namely either $\sigma W\sigma^*=W$ or $\sigma W\sigma^*=W^*$, resp.\
$\sigma H\sigma^*=H$ or $\sigma H\sigma^*=-H$.
Specifically, we consider one or all three of the following:
\begin{itemize}
	\item[] {\dff particle-hole symmetry} $\ph$, which is antiunitary  satisfying
                  $\ph W\ph^*=W$, resp.\ $\ph H\ph^*=-H$,
	\item[] {\dff time reversal symmetry} $\rv$, which is antiunitary satisfying
                  $\rv W\rv^*=W^*$, resp.\ $\rv H\rv^*=H$,
	\item[] {\dff chiral symmetry} $\ch$, which is unitary satisfying
                  $\ch W\ch^*=W^*$ resp.\ $\ch H\ch^*=-H$.
\end{itemize}
It is clear that if any two of these are part of the symmetry type, their product will be a symmetry of the third kind.

We will call a {\dff representation} of a symmetry type any collection of Hilbert space operators satisfying the
specified multiplication table, and (anti-)unitarity conditions.
Writing down a representation is also supposed to retain the information about how each symmetry is to act on unitary resp.\ Hermitian operators. Hence it makes sense to call an operator $W$ or $H$ {\dff admissible for the representation}, if it is admissible for all symmetry operators of the given representation. We note that,  as introduced  in \Qref{sec:firstsetting}, we generally use ``admissibility'' as including the condition of an essential gap (see also \Qref{sec:gaps}).

\subsection{Classification of symmetry types}\label{sec:typeclassification}
While we try to cover many symmetry types with as few case distinctions as possible, i.e., to allow ``general'' symmetry types, we do follow the literature in its typical restrictions in this regard. That is, we consider only symmetry groups formed out of the three kinds of symmetries described in the previous section, with up to four elements (counting the identity). Moreover, in the cases with four-element group, we take the three non-identical involutions to be one of each of the three different kinds described.

Since all symmetries are defined by their action on observables, we consider $\sigma$ and $\zeta\sigma$ with a phase $\zeta\in\Cx$, $\abs\zeta=1$ to represent the same symmetry. Therefore, we are free to adjust such phases in order to simplify the relations between the (anti-)unitary symmetry operators.
For a single involution we must have that $\sigma^2$ is a multiple of the identity. Hence if $\sigma$ is unitary we can adjust the phase so that $\sigma^2=\idty$; in the antiunitary case, from equating $\sigma^2\sigma=\sigma\sigma^2$, we must have $\sigma^2=\pm\idty$. This leaves three kinds of involutive symmetries: unitary, antiunitary with square $+\idty$, and antiunitary with square $-\idty$. These are clearly
distinguished geometrically by their action on operators and on the state space. In the simplest case (a qubit with the Bloch sphere as the state space) these correspond, respectively, to reflections around an axis (equivalent to a rotation around the axis by $\pi$), reflections along a plane through the origin, and the reflection at the origin. Including the case of no symmetry, this accounts for the first six entry lines in \Qref{Tab:sym}.

Going on to the cases with all three symmetries $(\ph,\rv,\ch)$ present, we get four more cases, distinguished by the signs of the squares of the antiunitary elements. Altogether we get the so called \textit{tenfold way} \cite{Altland-Zirnbauer} shown in \Qref{Tab:sym} together with their customary identifiers, and further information to be explained below. For this classification to be complete within its scope, we need to verify that the signs of the squares of antiunitary elements determine the symmetry type up to a phase convention, i.e., that we do not have to distinguish further subcases. This verification is done in the following Lemma.

\begin{lem}\label{lem:phases}
	For symmetry types that contain all three operations $\ph,\rv,\ch$ there is a distinguished phase convention which makes the
	three operators $\ph,\rv,\ch$ commute, and satisfy the relation $\ph\rv=\ch$.  With this convention the signs of $\rv^2$ and $\ph^2$  determine the entire multiplication table.
\end{lem}

\begin{proof}
	Observe that the operators $\ph\rv$ and $\rv\ph$ both implement the same symmetry and therefore can only differ by a phase. By choosing a phase for $\ph$ we can therefore achieve $\ph\rv=\rv\ph$, and we adjust the phase of $\ch$ so that $\ph\rv=\ch$. Then 	$\ph\ch=\ph^2\rv=\rv\ph^2=(\ph\rv)\ph=\ch\ph$, and similarly $\rv\ch=\ch\rv$.
\end{proof}

We note that with this convention we can have that $\ch^2=-\idty$ although, with a different, perhaps more widespread convention this could be made to be $+\idty$. As the Lemma shows, one then needs to memorize the signs in the (anti-)commutation table of the symmetries. In the sequel we will stick to the convention described in the Lemma.

\begin{table}\begin{centering}
		\begin{tabular}{|c||c|c|c||c||c|c|}
			$\symS$  &$\ph^2$  &$\rv^2$  &$\ch^2$  & irreps &$\igS$ &$\six$
			\\
			\hline
			\symA    &         &         &         &$1$   & $0$ & \\[3pt]
			\symD    &$\idty$  &         &         &$1$   &$\Ir_2$  &$d\,\mod2$  \\[3pt]
			\symC    &$-\idty$ &         &         &$2$   &$0$  &
			\\[3pt]
			\symAI   &         &$\idty$  &         &$1$   &$0$  &
			\\[3pt]
			\symAII  &         &$-\idty$ &         &$2$   &$0$  &
			\\[3pt]
			\symAIII &         &         &$\idty$  &$1^+,1^-$   &$\Ir$    &$\tr\ch$
			\\[3pt]
			\symBDI  &$\idty$  &$\idty$  &$\idty$  &$1^+,1^-$   &$\Ir$    &$\tr\ch$
			\\[3pt]
			\symCI   &$-\idty$ &$\idty$  &$-\idty$ &$2$   &$0$  &
			\\[3pt]
			\symCII  &$-\idty$ &$-\idty$ &$\idty$  &$2^+,2^-$   &$2\Ir$   &$\tr\ch$
			\\[3pt]
			\symDIII &$\idty$  &$-\idty$ &$-\idty$ &$2$   &$2\Ir_2$ &$d\,\mod4$
		\end{tabular}\hspace{10pt}
		\caption{\label{Tab:sym}Symmetry types considered in this paper, with their generators and relations for their squares.
			Absence of an entry in the respective column means that the generator is not part of the type.
			The first column gives the Cartan classification \cite{Altland-Zirnbauer}.
			Irreducible representations (see \Qref{sec:irreps}) are labelled by their dimension and the sign of $\ch$ as a superscript, where applicable.
			$\igS$ is the range of the symmetry index $\six$ defined in \Qref{prop:six}. Here ``$0$'' is a shorthand for the abelian group $\{0\}$.
			The last column gives an explicit expression for $\six(\rho)$, where $d$ denotes the dimension of the representation.
		}
\end{centering}\end{table}

\subsection{Gaps and essential gaps}\label{sec:gaps}
When a chiral symmetry $\ch$ or a particle-hole symmetry $\ph$ is present, with every eigenvector $\psi$ of an admissible unitary $W$, say $W\psi=\omega\psi$, $\ch\psi$ (resp.\ $\ph\psi$) is also an eigenvector, but for the complex conjugate eigenvalue. More generally, the spectrum is invariant under complex conjugation. The two real points $\pm1$ therefore play a special role in that the respective eigenspaces are invariant under the symmetry operators. Much of our analysis rests on the analysis of these symmetry representations, and we will make crucial use of the property that they are finite dimensional and that the eigenvalues $\pm1$ are isolated. This will be our standing assumption, and conveniently covers both the case of gapped translation invariant systems,  and combinations of two bulk phases, which do develop eigenvalues at $\pm1$. We express it by saying that $W$ has an {\dff essential gap} at each of the points $\pm1$, and include this condition when calling a walk ``admissible''.

This is equivalent to saying that $\pm1$ are not in the essential spectrum of $W$, which is defined as the set of complex numbers such that the spectral projection of any neighbourhood is infinite dimensional. Yet another useful way to state the essential gap condition is to consider the {\dff Calkin algebra}, i.e., the quotient of the algebra of bounded operators by the two-sided, closed ideal of compact operators. Denoting by $\calk W$ the image of
$W$ in the Calkin algebra, we can say that $W$ is essentially gapped iff $\pm1$
are not in the spectrum of $\calk W$. This provides the quickest way to see that compact perturbations do not destroy the essential gap condition: when $W'-W$ is compact,  $\calk{W'}=\calk W$, so clearly the two have the same essential spectrum.

Throughout this paper we will use the qualification ``essential'' in the sense of ``up to compact operators'', in particular for ``essentially local'' (\Qref{sec:firstsetting}) and ``essentially unitary''(\Qref{sec:compactsetting}), although we also follow common usage to call ``essentially invertible'' operators ``Fredholm''. In all these cases it is useful to view the corresponding property as defined in terms of $\calk W$.

Of course, the same ideas apply to Hamiltonians. In that case the spectrum has to be symmetric with respect to the operation $E\mapsto-E$, and the distinguished point in the spectrum is $0$ (see \Qref{fig:gaps}).
Admissibility for Hamiltonians is taken to imply an essential gap at $0$.

\begin{figure}
	\begin{tikzpicture}[
	scale=1.5,
	font=\footnotesize,
	cont/.style={line width=3,green!80!black}
	]
	\def\hamx{2.8}
	\tikzset{
		rot/.style={rotate around={0:(\hamx,0)}}
	}
	\def\thsq{.08}
	\def\thp{.03}
	\colorlet{colorp}{red}		
	\colorlet{colorgap}{blue!80!black}
	\draw ({-1.0-5*\thsq},1) node(a) {a)};
	\draw ({\hamx-1.0},1) node(b) {b)};

	\draw[gray,rot] (\hamx,-1)  -- (\hamx,1);
	\draw[cont,rot] (\hamx,-1)  -- (\hamx,-.5);
	\draw[cont,rot] (\hamx,1)  -- (\hamx,.5);
	\foreach \y in {.4,.25,-.25,-.4}
	\draw[colorp,fill,thick,rot] (\hamx,\y) circle (\thp);

	\draw[colorgap,line width=1.5,rot] ({\hamx-5*\thsq},0) to +(10*\thsq,0);
	\draw [<->,line width=2,rot] ({\hamx+.2},.8) .. controls ({\hamx+.8},.5) and ({\hamx+.8},-.5) .. ({\hamx+.2},-.8);
	\draw[gray] (0,0) circle (1);		
	\draw[cont,domain=40:160] plot ({cos(\x)}, {sin(\x)});
	\draw[cont,domain=200:320] plot ({cos(\x)}, {sin(\x)});
	
	\draw[colorgap,line width=1.5] ({-1.0-5*\thsq},0) to (1+5*\thsq,0);
	
	\foreach \ang in {12,30,168,-12,-30,-168,-180}
	\draw[colorp,fill,thick] ({cos(\ang)}, {sin(\ang)}) circle (\thp);

	\draw [<->,line width=2] ({cos(80)}, {sin(80)-.03}) -- ({cos(80)}, {-sin(80)+.03});
	\end{tikzpicture}
	\caption{\label{fig:gaps}
		(a) schematic spectrum of unitary quantum walk with band spectrum (green) and discrete spectrum (red dots) and
		essential gaps at $\pm1$. The arrow symbolizes the action of the symmetry operators $\ph,\ch$ on quasi-energies.
		We are interested in the eigenvalues on the symmetry axis indicated, around which we assume an essential gap.
		(b) spectrum of self-adjoint Hamiltonian marked analogously.}
\end{figure}
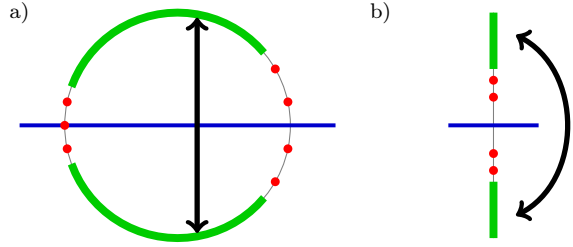

\subsection{Balanced representations and index group}\label{sec:balanced}

Under a norm-small perturbation of a unitary $W$ the eigenspace of $+1$ may change,
e.g., by splitting into branches. However, in a scenario where the deformations
must be admissible for a  symmetry type which forces spectra to be invariant
under complex conjugation, such branches must come in complex conjugate pairs.
Hence the parity of the dimension of the eigenspace will be constant.
We need to sharpen this argument, by making full use of the given symmetry type, and also by considering how the symmetry is represented on the eigenspace. We begin by defining the ``trivial'' representations, which will later be identified as those by which the eigenspace can change.

\begin{defi}\label{def:balancedrep}
	A symmetry representation $\rho$ is called {\dff balanced} if there exists a unitary operator $U$ or, equivalently, a Hamiltonian $H$, which is admissible for the symmetry representation and gapped.
\end{defi}

The ``non-triviality'' of a representation is then described by the following proposition in which we identify the group equivalence classes of representations modulo balanced ones, with the direct sum as addition. A direct way to express this is by a map $\six$ associating with every representation an element of an abelian group $\igS$.

\begin{prop}\label{prop:six} For every symmetry type $\symS$ there is an abelian group $\igS$ and a map ``$\six$'' taking any finite dimensional representation $\rho$ of the type
	to an element $\six(\rho)\in\igS$ such that
	\begin{itemize}
		\item[(0)] For every $j\in\igS$ there is a representation $\rho$ such that $j=\six(\rho)$.
		\item[(1)] $\six(\rho_1)=\six(\rho_2)$, if $\rho_1=U\rho_2U^*$ for a unitary operator $U$.
		\item[(2)] $\six(\rho_1\oplus\rho_2)=\six(\rho_1)+\six(\rho_2)$.
		\item[(3)] $\six(\rho)=0$ if and only if $\rho$ is balanced.
	\end{itemize}
\end{prop}

\begin{proof}
	We begin with an approach that is slightly more abstract than needed, in which (3) is replaced by the weaker condition (3'): ``$\rho$ balanced $\Rightarrow\six(\rho)=0$'', and some maximality condition to make $\six$ as non-zero as possible. For this approach we only need the obvious statement:
\begin{equation}\label{bplus}\tag{b$\oplus$}
  \text{``The direct sum of balanced representations is balanced.''}
\end{equation}
We define $\igS$ as the set of equivalence classes of finite dimensional representations with respect to the relation ``$\rho_1\sim\rho_2$'' which is defined by the existence of balanced representations $\beta_1,\beta_2$, such that $\rho_1\oplus\beta_1$ is unitarily equivalent to $\rho_2\oplus\beta_2$, i.e., 	
	\begin{equation}\label{eq:equivrel}
	\rho_1\sim\rho_2\quad\Leftrightarrow\quad\exists\beta_1,\beta_2\text{ balanced, }U\text{ unitary s.t. }
	\rho_1\oplus\beta_1=U(\rho_2\oplus\beta_2)U^*.
	\end{equation}
	This relation is clearly reflexive and symmetric. Transitivity is implied by \eqref{bplus}.  The map $\six$ is just the assignment of each representation to its equivalence class. Then (0), (1), and (3') are obvious. From \eqref{bplus} it is easily verified that the direct sum of equivalent representations gives equivalent sums. Therefore, (2) can be taken as the definition of the sum in $\igS$. Since $\rho_1\oplus\rho_2$ and $\rho_2\oplus\rho_1$ are unitarily equivalent, this sum makes $\igS$ an abelian semigroup.
	
	From here we could go on to define the Grothendieck group of the monoid $\igS$, as the largest {\it group} with the above properties. However, for the symmetry types of the tenfold way we are in the lucky situation that we already have inverses:
	
	For any representation $\rho$ we consider $\rho'$, a representation formed by the same symmetry operators, but with additional signs, namely for $\rv'=-\rv$ and $\ch'=-\ch$  whereas $\ph'=\ph$. Then, as one easily shows, the operator $V(\phi_1\oplus\phi_2)=\phi_2\oplus-\phi_1$ is an admissible unitary for $\rho\oplus\rho'$ and has eigenvalues $\pm i$. Hence $\rho\oplus\rho'$ is balanced, and $\six(\rho)+\six(\rho')=0$. We conclude that $\igS$ as defined above is indeed an abelian group.
\end{proof}

Now as will be made clear in \Qref{sec:indexgroups} there is only a small list of irreducible representations (irreps) of each symmetry type: either each irrep is already balanced, which means $\igS=\{0\}$, or there is only one irrep, which is its own inverse ($\igS=\Ir_2$), or there are two, which are inverses of each other ($\igS=\Ir$). In either case, property (3) is obvious.

\subsection{Computing the index groups}\label{sec:indexgroups}
We will now compute the index groups of the ten symmetry types. This will be
done by computing the irreducible representations (``irreps'') of the respective symmetry
type and identifying them with the generators of the corresponding index group. This section makes use of some well-known results from the theory of (anti-) unitarily represented symmetries. See for example \cite[Sect.~26]{Wignerbook} for a detailed discussion of the topic. Finally we will comment on \textit{forget homomorphisms}, that is, the maps on the index groups that correspond to changing symmetry types by ``forgetting'' two out of four symmetries.

\subsubsection{Computing the irreps}\label{sec:irreps}
For any group $G$ consisting of unitary and antiunitary operators, the subset of unitary operators is a normal subgroup $G_U$ and $G=G_U\cup \theta G_U$, where $\theta\in G$ is any antiunitary element.
For the symmetry types considered here, $G_U$ either contains only the identity or, in addition, the
chiral symmetry. It therefore has either one or two one-dimensional irreps, corresponding to the different possible signs for $\ch$ ($\pm1$ for $\ch^2=\idty$ and $\pm i$ for
$\ch^2=-\idty$). To determine the irreps of $G$, assume the representation of $G_U$ to be completely reduced and pick an irrep. Now choose a basis element $\psi$ of the underlying Hilbert space. Regarding the action of the antiunitary operator $\theta$ on $\psi$ there are then three cases to distinguish:
\begin{itemize}
	\item[(1)] $\psi$ can be chosen invariant.
	\item[(2)] $\psi$ is mapped to a linearly
	independent vector, which is a basis for an equivalent irrep of $G_U$.
	\item[(3)] $\psi$ is mapped to a linearly independent vector, which is a basis for an inequivalent irrep of $G_U$.
\end{itemize}

In case (1) the irreps of $G_U$ already determine the irreps of the whole
group. This is the case for the symmetry types without an antiunitary operator
which squares to $-\idty$: \symA{}, \symD{}, \symAI{}, \symAIII{} and
\symBDI{}, with two inequivalent irreps for \symAIII{} and \symBDI{}
respectively.

The remaining symmetry types contain at least one anti-unitary $\sigma$ with $\sigma^2=-\idty$. Thus, since $\braket\phi{\sigma\phi}=\braket{\sigma^2\phi}{\sigma\phi}=-\braket{\phi}{\sigma\phi}=0\ \forall\phi\in\HH$, they belong to either case (2) or (3).

In case (2) we get irreps of $G$, which contain two copies of equivalent irreps of
$G_U$ and are therefore of dimension two. This is the case for the symmetry types
\symC{}, \symAII{} and \symCII{}: assume $\ch\psi=\pm\psi$. Then, with \Qref{lem:phases}, we have $\ch\ph\psi=\ph\ch\psi=\pm\ph\psi$, so that $\ph\psi$ and $\psi$ are eigenvectors for the same eigenvalue of $\ch$ and therefore correspond to equivalent irreps. Since for \symC{} and \symAII{} all irreps of
$G_U$ are equivalent, the procedure yields one overall irrep of
dimension two. For \symCII{} there are two inequivalent irreps of $G_U$ and
therefore also two for $G$.

The two remaining symmetry types \symCI{} and \symDIII{} are examples for
case (3). Then the $\pm i$-eigenspaces of $\ch$ are mapped to each other by $\rv$ and $\ph$, which can be deduced in a way analogous to the one before. But now we get a sign-flip when passing $\pm i$ by $\ph$ or $\rv$. An overall
irrep then contains a copy of each of the two inequivalent irreps of $G_U$
leaving only one possible two-dimensional irrep for $G$.

\subsubsection{Computing the index groups}
The computation of the index groups is now relatively simple, since, by \Qref{prop:six}, the generators of the index group are given by the images of the irreps under $\six$. When there is only one irrep $\rho$, $\six(\rho)$ must be its own inverse, so the index group is $\igS\cong\{0\}$ or $\igS\cong\Ir_2$. When there are two, which are inverses of each other, the index group is $\igS\cong\Ir$. Let us go through the computations in detail:

\noindent {\it Types \symA{}, \symAI{}, \symAII{}:}

In these cases $U=i\idty$ ($H=\idty$) is always admissible and gapped, so every
representation is balanced and $\igS=\{0\}$.

\noindent {\it Types  \symC{} and \symCI{}:}

Both symmetry types have only one irrep of dimension two. We find a basis $\{\phi,\ph\phi\}$ of it by choosing any vector $\phi$ in the case of \symC{} and an eigenvector of $\ch$ in the case of \symCI, respectively. By defining $U\phi=\ph\phi$ and $U\ph\phi=-\phi$ ($H\phi=i\ph\phi$ and $H\ph\phi=-i\phi$) we then get an admissible and gapped unitary (Hamiltonian). Therefore every representation is a
direct sum of balanced ones and we thus have $\igS=\{0\}$.

\noindent {\it Types \symD{} and \symDIII{}:}

Both symmetry types have only one irrep which, as will be shown below,
is not balanced. Since there is only one, the irrep must be its own inverse in the sense of the construction given in the proof of \Qref{prop:six}. Therefore, we identify the index groups of both symmetry types with $\Ir_2$.

In case \symD{} one can always choose the basis-element of the one-dimensional irrep space to be $\ph$-invariant. $\ph$ then acts as the complex conjugation with respect to this basis vector. Hence the only admissible unitaries (Hamiltonians) are $U=\pm \idty$ ($H=0$), which are not gapped. In two dimensions a real rotation matrix $U$ by an angle $\notin\pi\Ir$ ($H=\sigma_y$) is gapped and admissible, so all even dimensional representations are balanced. Therefore,  $\ig(\symD)=\Ir_2$, with
\begin{equation}\label{sixmod2}
  \six(\rho)\equiv d\ \mod2.
\end{equation}

In case \symDIII{} the symmetry conditions for $\ph$ and $\ch$ force the eigenvalues of an admissible unitary $U$ (Hamiltonian $H$) to come in complex conjugate pairs (pairs with opposite signs). Since $\rv$ leaves the eigenspaces invariant, $\rv^2=-\idty$ forces them to be at least of dimension two and therefore the eigenvalues of a gapped operator come in groups of four. On the other hand, a four-dimensional representation is always balanced. To see this choose a vector $\phi=\phi_1+\ph\phi_2$, with $\ch\phi_k=i\phi_k$. Then $\{\phi,\rv\phi,\ph\phi,\ch\phi\}$ is an orthogonal basis on which the action of the symmetries is fixed by \Qref{lem:phases}. In this basis, the diagonal operator $U=\operatorname{diag}(i,i,-i,-i)$ ($H=\operatorname{diag}(1,1,-1,-1)$) is admissible and gapped.  Similar to $\ig(\symD)$, the index group is isomorphic to $\Ir_2$. To take into account the two-dimensional irreps of \symDIII{} we set $\ig(\symDIII)=2\Ir_2$ considered as a subgroup of $\Ir_4$. The symmetry index of a $d$-dimensional representation $\rho$ is then given by
\begin{equation}\label{sixmod4}
  \six(\rho)\equiv d\ \mod4.
\end{equation}

\noindent{\it Types \symAIII{}, \symBDI{} and \symCII{}:}

In these cases, we know that $\ch^2=\idty$, therefore $\ch$ has eigenvalues
$\pm1$ and is given up to unitary equivalence by the pair $(n_-,n_+)$ of
multiplicities of these eigenvalues.
Let us now look at one of the irreducible representations of \symAIII{}, e.g.\ $\rho_+$, with
$(n_-,n_+)=(0,1)$. It lives on a one dimensional space in
which $\ch$ is the identity. Admissibility of a unitary operator $U$ means that
$U=U^*$ ($H=-H$), allowing only $U=\pm\idty$ ($H=0$), which are not gapped. Therefore, $\rho_+$ is not balanced.
Its inverse, as given in the proof of \Qref{prop:six}, is the
representation $\rho_-$ with $(n_-,n_+)=(1,0)$. By applying the direct sum rule
we get, for an arbitrary representation $\rho$ with multiplicities $(n_-,n_+)$
\begin{equation}\label{sixchi}
\six(\rho)=n_+\six(\rho_+)-n_-\six(\rho_+)=(\tr\gamma)\,\six(\rho_+),
\end{equation}
where the multiplication of group elements by integers is understood as
iterated sum. Hence $\ig(\symAIII)$ is generated by the single element
$\six(\rho_+)$. Moreover, no multiple of this element vanishes, since the
representations with multiplicities $(0,n)$ are all unbalanced. This identifies
the index group with $\Ir$, and we can just set $\six(\rho)=\tr\ch$.
Note that any such identification can only be up to an automorphism of
$\ig(\symAIII)$, so $\six(\rho)=-\tr\ch$ would also be a valid choice, differing
only by taking the other irreducible representation $\rho_-$ as the generator of
$\ig(\symAIII)$ to be identified with $1\in\Ir$.  All chiral symmetry indices,
when written as natural numbers, depend on such a convention. It is, in fact, a
phase convention in disguise, because we can always change the sign of $\ch$
(and one of the other symmetries) getting an equivalent representation of the
symmetry.

In case \symBDI\ $\ph$ plays the role of complex conjugation, and the arguments
of \symAIII\ apply with minimal changes:
The representations are still labelled by pairs of multiplicities  $(n_-,n_+)$,
only that now we fix an $\ph$-real basis in each eigenspace, so that the action
of all symmetry operators is determined. If we make the same convention for
identifying the generator of $\ig(\symBDI)$, the groups are just related by the
``forget homomorphism'' which considers a \symBDI-representation as an
\symAIII-representation by not considering $\ph$ and $\rv$.

The irreducible representations of \symCII\ are
two-dimensional. Thus the possible generators of the group $\ig(\symCII)$ are the representations with ($\ch$-) multiplicities $(0,2)$ and $(2,0)$. Writing
the group as $2\Ir$ rather than $\Ir$ makes the forget homomorphism (see next paragraph)
$\ig(\symCII)\to\ig(\symAIII)$ especially simple.

\subsubsection{Forget homomorphisms}
Whenever we have a symmetry type that contains all three symmetries, it can also be considered as another type, by ``forgetting'' two symmetries. How this is reflected in the index group can be described by \textit{forget homomorphisms}. The nontrivial homomorphisms are collected in \Qref{Tab:forget}.

When we leave out $\ph$ and $\ch$, we are left with a symmetry group containing only $\rv$ (\symAI{}, \symAII{}) with trivial index group and therefore the forget homomorphism must be the zero map. If we instead ignore $\rv$ and either $\ph$ or $\ch$, we are left with a nontrivial index group. The forget homomorphisms can be deduced in a straight forward way (see e.g.\ $\ig(\symBDI)\to\ig(\symD)$ in \Qref{Tab:forget}). The only case which is not obvious, is the homomorphism $\ig(\symDIII)\to\ig(\symAIII)$: consider an irrep of \symDIII{}, where $\ch$ can be chosen as $\operatorname{diag}(-i,i)$, where the eigenvalues are determined by \Qref{lem:phases}. If we ignore $\ph$ and $\rv$, we are free to change the phase of $\ch$, to get a representation with $\ch=\operatorname{diag}(-1,1)$. Therefore an irrep of \symDIII{} becomes a representation of \symAIII{} of the form $\rho_-\oplus\rho_+$, which has index zero.

\begin{table}\begin{centering}
		\begin{tabular}{|c|ccc|}
			& \s\symBDI\s  &\s\symCII\s  &\s\symDIII\s   \\\hline
			\symAIII  & $\id$       & $\id$     & 0\\
			\symD     & $\mod2$     &         & 0
		\end{tabular}\hspace{10pt}
		\caption{\label{Tab:forget} Forget homomorphisms between non-trivial index groups.  Sample entry: symmetry type \symBDI\ contains an involution $\ph$ with $\ph^2=\idty$,
			so every representation of \symBDI\ is also one of \symD. This induces a homomorphism $\Ir=\ig(\symBDI)\to\ig(\symD)=\Ir_2$, in this case the quotient map $\mod2$. }
\end{centering}\end{table}

\subsection{Symmetry indices of unitaries and Hamiltonians}
\label{sec:sym:indicDef}

So far, the symmetry index is only defined for finite dimensional representations. This may seem insufficient for classifying operators $W$ on an infinite dimensional Hilbert space.
However, the symmetries leave the $\pm1$-eigenspaces of admissible unitaries (and the $0$-eigenspaces of admissible Hamiltonians) invariant and the essential gap condition ensures their finite dimensionality. Hence, the symmetry index of the restriction of the symmetries to these eigenspaces is well-defined.

\begin{defi}\label{def:sixOp}
The {\dff symmetry index} of an admissible Hamiltonian $H$, denoted by $\six(H)$,  is the symmetry index of the representation of the symmetries on its $0$-eigenspace. For an admissible unitary $W$, we define
$\six_+(W)$ as the symmetry index of the representation on its $+1$-eigenspace, and $\six_-(W)$ using the $-1$-eigenspace. For the index in the combined eigenspaces we write $\six(W)=\six_+(W) + \six_-(W)$.
\end{defi}

When the overall Hilbert space is finite dimensional and carries a symmetry representation $\rho$, we have $\six(W)=\six(\rho)$ for \emph{all} admissible $W$.  Indeed, the representation in the eigenspaces differs from $\rho$ only by the representation belonging to the non-real part of the spectrum, which is balanced by definition. Since we assume each cell to carry a balanced representation, $\six(W)=0$ for $W$ acting on any finite number of cells. However, this is no longer true for infinitely many cells. In a sense the whole theory is about mismatches allowing for $\six(W)\neq0$.

It will be useful to have a concrete formula for $\six_\pm(W)$ in the finite dimensional case.
We collect them here for later use. Symmetry type \symDIII{} is omitted, because we have no equally simple formula.

\begin{lem}
	Given a $d<\infty$ dimensional symmetry representation and a symmetry admissible unitary $W$. Then
	\begin{eqnarray}
	(-1)^{\six_-(W)}&=&\det(W)\quad\mbox{for symmetry type}\ \symD   \label{siDet}\\
	\six_\pm(W)&=& \frac12\tr\gamma(\idty\pm W)\quad\mbox{for symmetry types}\ \symAIII,\symBDI,\symCII. \label{sitrch}
	\end{eqnarray}
\end{lem}
\begin{proof}
	The non-real eigenvalues of admissible unitaries $W$ come in complex conjugate pairs for the symmetry types under consideration. Regarding the first equation, we find that $\det(W)=\pm 1$ and its value is determined by the parity of the $-1$-eigenspace, hence by $\six_-(W)$. Considering $\det(-W)$, a similar formula is true for $\six_+$. For the second equation, note that $\ch$ acts as $\sigma_x$ on eigenvectors corresponding to complex conjugate pairs of eigenvalues of $W$ and is therefore traceless on this subspace. Hence, the formula yields $\tr\ch$ on the $\pm1$-subspace of $W$ respectively.
	
\end{proof}

One cornerstone of this theory is the homotopy invariance of the symmetry indices of admissible operators. That is, whenever $W'$ is a gentle perturbation of $W$, we must have $\six_-(W')=\six_-(W)$ and $\six_+(W')=\six_+(W)$, and similarly for Hamiltonians. This is the upshot of the following Proposition. It actually makes a slightly stronger statement, namely that the indices are locally constant in the norm topology.

\begin{prop}\label{pro:homoto}
	Let $W_0$ be an admissible walk. Then there is a constant $\veps>0$ such that $\six_\pm(W_1)=\six_\pm(W_0)$ for both signs and all admissible $W_1$ with $\norm{W_1-W_0}<\veps$.\\
	The same statement holds for admissible Hamiltonians and the index $\six(H)$.
\end{prop}

\begin{proof}
	We show this only in the Hamiltonian case, because the cases $\six_\pm(W)$ are completely analogous. So let $H_0$ be an admissible Hamiltonian, and $P_0$ the projection onto its $0$-eigenspace, which we will take to be $d$-dimensional. Since $0$ is not in the essential spectrum, this eigenvalue is isolated, i.e., there is a distance $\delta>0$ to all other eigenvalues and other parts of the spectrum. The radius $\veps$ in the Proposition will depend only on $d$ and $\delta$.
	
	The proof will be based on the perturbation theory in terms of resolvents as described in \cite{Kato} or \cite[Ch.~XII]{ReedSi4}.
	We denote by $\Gamma$ the circular path in the complex plane around $0$ of radius $\delta/2$.  For $z\in\Gamma$ the resolvent $R_0(z)=(z\idty-H_0)\inv$ is bounded by $2/\delta$. By $R_1(z)$ we denote the resolvent of $H_1$.  By the resolvent equation
	\begin{equation} \label{eq:R1-R0}
	R_1(z)-R_0(z) = R_1(z) (H_1-H_0) R_0(z)
	\end{equation}
	we have
	\begin{equation} \label{eq:R1}
	R_1(z) = R_0(z) [\idty-(H_1-H_0)R_0(z)]^{-1}=\sum_{n=0}^\infty R_0(z)\Bigl((H_1-H_0)R_0(z)\Bigr)^n,
	\end{equation}
	which is a convergent series everywhere on $\Gamma$, provided $\norm{H_1-H_0}<\delta/2$. Assuming this inequality from now on, we have that the resolvent $R_1$ is defined and uniformly norm bounded by $(2/\delta)\bigl(1-2\norm{H_1-H_0}/\delta\bigr)\inv$. By the Cauchy integral formula the spectral projection of $H_1$ for the interior of the circle $\Gamma$ is
	\begin{equation}\label{Q-integral}
	Q=\int_\Gamma\mkern-3mu\frac{dz}{2\pi i}\,R_1(z).
	\end{equation}
	The corresponding integral with the resolvent $R_0$ just gives the projection $P_0$. Therefore,
	\begin{equation}\label{homidiff}
	Q-P_0=\int_\Gamma\mkern-3mu\frac{dz}{2\pi i}\,(R_1(z)-R_0(z)).
	\end{equation}
The difference of the resolvents is just the sum \eqref{eq:R1} with the term $n=0$ omitted. We insert this into \eqref{homidiff}, estimate the integrand term by term, and multiply by the length of $\Gamma$, which gives
	\begin{equation}\label{homidiffnorm}
	\norm{Q-P_0}\leq \frac{2\norm{H_1-H_0}/\delta}{1-2\norm{H_1-H_0}/\delta} .
	\end{equation}
	Hence if we assume the tighter bound $\norm{H_1-H_0}<\delta/4$, by \cite[I.\S 4.6]{Kato} this implies that $Q$ and $P_0$ have the same dimension $d$.
	
	Let us now bring in the symmetry conditions. Each of the symmetries $(\ph,\rv,\ch)$ has the property that eigenvectors of $H$ are mapped to eigenvectors of $H$, possibly with a sign change of the eigenvalue. Therefore, the subspace $Q\HH$, i.e.\ the eigenspace of $H_1$ for the interval $[-\delta/2,\delta/2]$ is invariant under each of the symmetries. Let us denote by $Q\rho Q$ the representation of the symmetry in this invariant subspace. When $P_1$ is the projection onto the $0$-eigenspace of $H_1$, the representations $Q\rho Q$ and $P_1\rho P_1$ differ by a balanced representation, since $H_1$ is gapped on $(Q-P_1)\HH$. This gives \begin{equation}\label{sixesperturbed}
	\six(H_1):=\six(P_1\rho P_1)=\six(Q\rho Q).
	\end{equation}
	Hence we only have to prove that for sufficiently small $\norm{H_1-H_0}$ we have $\six(Q\rho Q)=\six(P_0\rho P_0)=:\six(H_0)$. This is obvious for the symmetry types \symD{} and \symDIII{}, since in these cases the symmetry index is determined by the dimension of the representation, and we have already established that $\dim Q\HH=\dim P_0\HH$, as soon as $\norm{H_1-H_0}<\delta/4$. That is, the proposition holds with $\veps=\delta/4$.
	
	For symmetry types with trivial index group the Proposition is trivially true, which leaves the cases \symAIII, \symBDI, and \symCII, for which the symmetry index of finite dimensional symmetry representations is given by $\tr(\ch)$.  Hence
	\begin{align}\label{sixQ-P0}
	\abs{\six(Q\rho Q)-\six(P_0\rho P_0)}&=\abs{\tr(Q\ch Q)-\tr(P_0\ch P_0)}\\ &\leq\abs{\tr(Q\ch Q)-\tr(P_0\ch Q)}+\abs{\tr(P_0\ch Q)-\tr(P_0\ch P_0)}\\
	&\leq \norm{Q-P_0}\,(\norm{\ch Q}_1+\norm{P_0\ch}_1)\\
	&\leq 2d \norm{Q-P_0},
	\end{align}
	where $\norm\cdot_1$ denotes the trace norm, and we have used that $\tr P_0=\tr Q=d$. Combining this with \eqref{homidiffnorm} we find that for sufficiently small $\norm{H_1-H_0}\leq\veps$ the left hand side of \eqref{sixQ-P0} is $<1$, so the indices, being integers have to coincide. Collecting the constants we find that the proposition holds with
	\begin{equation}\label{homveps}
	\veps=\frac\delta2\,\frac1{2d+1}.
	\end{equation}
\end{proof}

The converse of this proposition would be the {\it completeness} of these invariants: when the indices of two operators coincide, they can be deformed into each other. This issue will be addressed in \Qref{sec:complete}.

\section{Decouplings and two settings}\label{sec:twosettings}
In the previous section we established two homotopy invariants, $\six_\pm(W)$ for walks in our setting. However, this cannot be the whole story, since these are trivial in the translation invariant case, and reflect nothing of the cell structure and the locality condition. This section is a heuristic introduction and preview of how locality enters and allows us to define new invariants. It turns out, that these follow a somewhat different set of rules, i.e., are best formulated in a setting of their own. This is set out in detail in \Qref{sec:indices}, and will be brought together with the invariants described so far when we get to \Qref{sec:homotopy-indices} and \Qref{sec:decoup}. In the current section we only sketch these ideas, so we frequently refer to concepts to be properly defined only in these later sections.

The main idea for bringing in the one-dimensional lattice structure is to split the system into a right and a left subsystem. Let us assume for the sake of discussion that we have a \emph{gentle decoupling} of $W$, i.e., a local and gentle perturbation of the form
\begin{equation}\label{decoupSum}
W'=W_L\oplus W_R,
\end{equation}
where the direct sum is with respect to $\HH=(\Pl a\HH)\oplus(\Pg a\HH)$, with $a\in\Ir$ some arbitrarily chosen cut point. Thus under $W'$ a walker starting somewhere at $x\geq a$ will have zero probability to reach $y<a$, or conversely, after any number of steps. Whether such a decoupling always exists is far from obvious, but this will be established in \Qref{thm:gentdecoup}.
Then for each walk $W_L,W_R$ we can consider the eigenspaces at $\pm1$, and their symmetry indices leading to the four index quantities in the upper left of the following table
\begin{equation}\label{si22cross}
\begin{array}{cc|c}
     \six_+(W_L)&\six_+(W_R)&\six_+(W')\\
     \six_-(W_L)&\six_-(W_R)&\six_-(W')\\\hline
     \sixL  (W')&\sixR  (W')&\six  (W')\vrule height 5pt depth 0pt width 0pt
\end{array}
\end{equation}
On the right and at the bottom we have collected the respective marginal sums, where $\sixL(W')$ is just defined as $\six(W_L)=\six_-(W_L)+\six_+(W_L)$. So which of these numbers are independent of where and how the cut is made, and could hence play the role of invariants? It turns out that there are two complementary aspects, which in the end require quite different tools, and are, loosely speaking, associated with the row sums and the column sums of the above table, respectively.

\begin{figure}
	\def\sep#1{\draw (#1,-1)--(#1,0);}
\def\frame{
\draw[fill=black!10](-6,-1) rectangle (6,0);
\sep4\sep2\sep{-4}\sep{-2}
\draw[thick] (0,-1.5) -- (0,7);
\draw (-3,1) rectangle +(2,4); \draw (-3,3)--+(2,0);
\draw (1,1) rectangle +(2,4);  \draw (1,3)--+(2,0);
\node at (-4,2) {$\KK_-$};\node at (4,2) {$\KK_-$};
\node at (-4,4) {$\KK_+$};\node at (4,4) {$\KK_+$};
\node at (-3,-2) {$W_L$};\node at (3,-2) {$W_R$};
}

\begin{tikzpicture}[scale=.4]
\frame
\draw[fill=white] (-2,3) ellipse (.8 and 1.2); 
\draw[fill=white] (2,3) ellipse (.8 and 1.2); 
\node at (-2,3) {$W_\KK$};\node at (2,3) {$W_\KK$};
\node at (-5,6) {\Large$W^\flat$};
\end{tikzpicture}
\hspace{1cm}
\begin{tikzpicture}[scale=.4]
\frame
\node at (-2,2) {$-\idty$};
\node at (2,2) {$\idty$};
\node at (-2,4) {$\idty$};
\node at (2,4) {$-\idty$};
\node at (-5,6) {\Large$W^\sharp$};
\end{tikzpicture}
	\caption{\label{fig:crosspert}The walks involved in the example of a crossover perturbation. (See text) }
\end{figure}
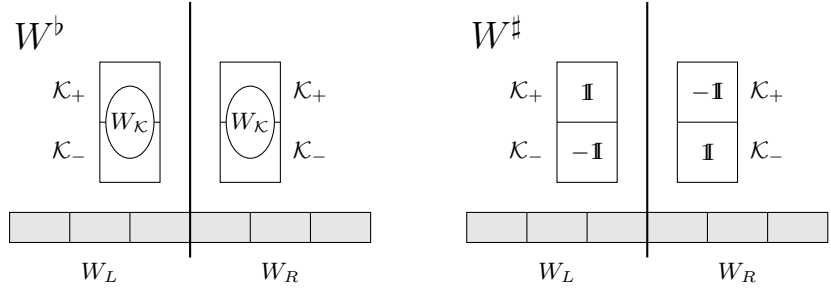

\noindent{{\bf Example: a crossover perturbation:} In order to get a prototype of the dependencies in \eqref{si22cross}, consider an additional cell on the left side, given by a Hilbert space $\KK$ on which the symmetry acts in a balanced way. More concretely, suppose there are  symmetry-invariant subspaces $\KK_\pm$ so that $\KK=\KK_-\oplus\KK_+$ and the restriction $\rho_\pm$ of the symmetry operators to $\KK_\pm$ has symmetry index $\six(\rho_\pm)=\pm n\in\igS$. The overall representation on $\KK$ is then balanced, and hence is a legitimate additional cell for an extended cell structure.  Now we extend the walk to the additional cell by choosing on $\KK$ a gapped unitary $W_\KK$. We make the same extensions with the same representations on the right hand side, and call the resulting walk $W^\flat$. Clearly, we can use the gentle decoupling of $W$ to decouple also $W^\flat$, by just leaving alone all the additional cells. Then since we do not get any additional eigenvalues at $\pm1$ on either side, the matrix \eqref{si22cross} will be exactly the same for $W^\flat$ as for $W$.

Let us now do the same thing but replace the gapped $W_\KK$ on the left hand side by the unitary operator, which is $+\idty$ on $\KK_+$ and $-\idty$ on $\KK_-$, and choose the opposite signs on the right.
Call the resulting unitary $W^\sharp$ (see \Qref{fig:crosspert}). Again we can do the decoupling by simply decoupling $W$. However, the $\pm1$ eigenspaces now contain full copies of $\KK_\pm$ on each side. We have to add the respective indices to get the indices of $W^{\sharp}$, namely
\begin{equation}\label{six22add}%
\sixmat{W^\sharp}=\sixmat{W^\flat}+ \left(\begin{array}{cc} n&-n\\ -n&n\end{array}\right).
\end{equation}
On the right hand side we have used the observation made earlier that the index data for $W^\flat$ and $W$ are equal. In this way we get on the left two walks $W_L^\sharp$ and $W^\flat_L$ on the same cell structure, which are clearly local perturbations of each other. However, by \Qref{pro:homoto}, this perturbation cannot be gentle, because the $\six_\pm$ indices are different. On the other hand, the overall perturbation from $W^\flat$ to $W^\sharp$ is gentle: just exchange the two copies of $\KK_+$ by a real rotation. This leaves $W=\idty$ on the right copy of $\KK$ which has vanishing indices, and is continuously connected to $W_\KK$, by fixing the eigenvectors of $W_\KK$ and moving the eigenvalues to $1$ along the unit circle. On the left we proceed similarly to connect $W_\KK$ to $-\idty$.

Let us collect some conclusions from this construction.
\begin{itemize}
\item[(a)] None of the individual terms $\six_\pm(W_{L,R})$ in \eqref{si22cross} is stable under gentle and local perturbations of $W$ \\(Compare $W^\flat$ and $W^\sharp$).
\item[(b)] None of these terms is independent of the cut position \\(Just make the cut so all additional cells end up on the same side).
\item[(c)] The indices $\six_\pm(W)$ are not stable under non-gentle, though local perturbations. \\(Compare $W^\flat_R$ and $W^\sharp_R$).
\item[(d)] The sum of the two indices $\six(W)=\six_-(W)+\six_+(W)$, corresponding to the row sums $\sixL(W')$ and $\sixR(W')$ in \eqref{si22cross}, might have such desired stability.
\end{itemize}

\noindent{{\bf Column sums:} (See \Qref{sec:indices})
Using gentle decouplings one can indeed argue that the column sums are good invariants, namely independent of the cut position and cut details. To see this one can consider two decouplings sufficiently far apart \cite{letter}. This leaves a middle piece $W_M$, i.e., an admissible unitary with respect to a finite dimensional representation on a direct sum of cells $\HH_x$, which we assumed to be balanced. Hence $\six(W_M)=0$, and this is just the difference of the indices obtained from the different cuts. This also shows that $\sixL(W)$ and $\sixR(W)$ are independent of other details in the choice of $W'$.

The tricky part is now to show that $\sixL(W)$ and $\sixR(W)$ are also stable under global continuous perturbations. This is by no means obvious. Although the decoupling construction described in \Qref{sec:decoup} seems nearly canonical it is {\it not} ``homotopy continuous'' as the above example shows. The arrows in the following diagram indicate which walks in the crossover perturbation example can be continuously and admissibly deformed into each other.
$$\begin{tikzpicture}
  \node (a) {$W^\flat$};
  \node (b1) [right=of a] {$W^\flat_L$};
  \node (b2) [right=-1mm of b1] {$\oplus$};
  \node (b3) [right=-1mm of b2] {$W^\flat_R$};
  \node (c) [below=of a] {$W^\sharp$};
  \node (d1) [right=of c] {$W^\sharp_L$};
  \node (d2) [right=-1mm of d1] {$\oplus$};
  \node (d3) [right=-1mm of d2] {$W^\sharp_R$};
  \draw[<->] (b1) -- (d1) node [left,midway,sloped] {\(\not\)};
  \draw[<->] (b3) -- (d3) node [left,midway,sloped] {\(\not\)};
  \draw[<->] (a) to node {} (b1);
  \draw[<->] (a) to node {} (c);
  \draw[<->] (c) to node {} (d1);
\end{tikzpicture} $$
So there is no way to do the decoupling construction in a way that homotopic walks lead to homotopic left half-walks. Nevertheless, the indices $\sixL(W)$ and $\sixR(W)$ are also homotopy stable. We show this in \Qref{sec:indices} by choosing a decoupling method which is not only not ``gentle'' but even destroys unitarity. Indeed the homotopy stability is easily shown for deformations in the extended class of ``essentially unitary operators''.

\noindent{{\bf Non-gentle perturbations:} (See \Qref{sec:homotopy-indices})
It is clear from the example, conclusion (c), that there are local perturbations, which cannot be contracted to the identity. It turns out that this effect is precisely classified by another index quantity, which we call the relative index $\sixrel W':W$ of a local  (more generally, a compact) perturbation $W'$ of $W$. This vanishes if and only if $W'$ can be contracted locally to $W$ (\Qref{lem:contract}). But could there be a perturbation which cannot be contracted locally, but by some large scale deformation affecting the entire walk? This is excluded (\Qref{thm:locpert}) by showing that the relative index is just the difference of the ``absolute'' indices $\six_+(W')$ and $\six_+(W)$. So the row sums in \eqref{si22cross} provide the complete classification of non-gentle perturbations.

\noindent{\bf Summary:}
The row sums and the column sums in the matrix \eqref{si22cross} are homotopy invariants, and independent of where the cut is made. However, the individual entries have neither stability. This leaves us with three independent homotopy invariants.
On the other hand, if we want stability also with respect to possibly non-gentle local perturbations, only the column sums provide invariants. These can be defined also independently of the existence of a gentle decoupling, and are also stable with respect to compact perturbations.
We will show in \Qref{sec:complete} that in either category the invariants described are complete.

\section{Indices stable under compact perturbations}\label{sec:indices}

As noted after \Qref{def:pertsorts}, in the Hamiltonian case every compact perturbation is gentle, since the convex combination $H_t=(1-t)H_1+tH_2$ provides a continuous connection.
The difficulty in the unitary case lies in keeping unitarity along the connecting path. The approach we take in this section is to waive unitarity for the connecting
path, keeping only a weakened condition: by an {\dff essentially unitary operator} $W$ we mean one such that
$W^*W-\idty$ and $WW^*-\idty$ are both compact operators. Then, for any compact
operator $K$, $W+K$ is also essentially unitary, and so is any convex combination $(1-t)W_1 +t W_2$ if $W_2-W_1$ is compact. So in the enlarged class compact perturbations become gentle.

At the same time all difficulties in getting decouplings vanish: if we project away the off-diagonal matrix blocks with respect to some $P_{\geq a}$, we normally destroy unitarity. But since these blocks are compact, we are still left with an essentially unitary operator.

With all the subtleties of the unitarity conditions gone we can thus expect to get a theory which is as straightforward as the Hamiltonian case. It is clear that in this theory the homotopy classes become much larger, so some features (like walks differing by a non-gentle compact perturbation) become wiped out. But enough remains to get a theory of $\sixL,\sixR$, i.e.\ the column sums in \eqref{si22cross}. We will come back to the questions of strictly unitary homotopy, and thus a theory appropriate for the row sums in \Qref{sec:homotopy-indices} and \ref{sec:decoup}.

\subsection{Setting}\label{sec:compactsetting}
We can now go through the assumptions in \Qref{sec:firstsetting}, making appropriate relaxations.
The cell structure of the Hilbert space and the symmetry types and operators will be unchanged.

The essential gap condition is also unchanged, but we can no longer phrase it in terms of the eigenspaces of $W$. Indeed, $W$ might not be diagonalizable. Instead we use the formulation that the image $\calk W$ of $W$ in the Calkin algebra, which is still a bona fide unitary element, does not have $\pm1$ in its spectrum.
Expressed directly in terms of $W$ this means that there are bounded operators
$R_1,R_{-1}$, namely some preimages of the resolvent of $\calk W$ in the Calkin algebra at $\pm1$, such that $R_1(W-\idty)-\idty$ and $R_{-1}(W+\idty)-\idty$ are
compact operators.

The symmetry conditions for $W$ required by the symmetry type will be assumed to be satisfied exactly.
When all this is satisfied, we call $W$ an {\dff admissible essentially unitary} operator. The connection with the Hamiltonian case is made by introducing the imaginary part of $W$
\begin{equation}\label{IM}
  \IM W=\frac1{2i}\,{(W-W^*)}.
\end{equation}
Note that $\IM W$ is exactly (not merely ``essentially'') Hermitian. It also satisfies the admissibility conditions for Hamiltonians with respect to symmetries as described in \Qref{sec:abstypes}. Finally, it has an essential gap at $0$, which is seen most readily by looking at the image of $\im W$ in the Calkin algebra, $\calk{\IM W}=(\calk W-\calk W^*)/(2i)$, and using the spectral mapping theorem. This allows the following definition.

\begin{defi} \label{def:six4essU}
	For an admissible essentially unitary operator $W$ we define
	\begin{equation}\label{six4essU}
	\six(W)=\six(\IM W).
	\end{equation}
\end{defi}

When $W$ is exactly unitary, this coincides with the earlier definition $\six(W)=\six_-(W)+\six_+(W)$.
Indeed, for unitary $W$ the $0$-eigenspace of $\im W$ is just the direct sum of the eigenspaces at $-1$ and $+1$. But an essentially unitary operator might not even be normal or diagonalizable, so an appropriate direct definition of $\six_+(W)$ would need additional considerations.

In the literature one also finds another reduction of the unitary case to the Hamiltonian case, namely by using the ``effective Hamiltonian'' $H=i\log W$ \cite{Asbo2,Kita}. This is not so useful to us: it not only destroys locality properties of the walk, but in order to preserve the symmetry properties, one needs to put the branch cut of the logarithm on the negative axis, directly through one of the points where we want to study additional eigenvalues.

\subsection{Indices of admissible essentially unitary walks}\label{sec:indiceshamwalk}
For an admissible essentially unitary operator to be an \emph{essentially unitary walk} we demand the same essential locality condition as before: $\Pg aW\Pl a$ and $\Pl aW\Pg a$ are compact operators for some (and hence all) $a\in\Ir$.
As noted in the previous subsection, for an admissible essentially unitary operator $W$ we define $\six(W)$ as the symmetry index of the representation in the $0$-eigenspace of $\IM W$. Now for any cut-point $a\in\Ir$ we define
\begin{equation}\label{sixRL}
\sixL(W)=\six(\Pl a  W\Pl a)  \quad\mbox{and}\  \sixR(W)=\six(\Pg a  W\Pg a),
\end{equation}
where $\Pl a W\Pl a$ and $\Pg a W\Pg a$ are considered as operators on the respective half-spaces.\\
The following Theorem collects a few basic properties.

\begin{thm}\label{thm:sixProps}\
	\begin{itemize}
		\item[(1)] The indices $\six(W)$, $\sixL(W)$, and $\sixR(W)$ are invariant under gentle as well as compact perturbations of $W$.
		\item[(2)] $\six(W)=\sixL(W)+\sixR(W)$.
		\item[(3)] The definitions \eqref{sixRL} do not depend on the cut-point $a\in\Ir$.
		\item[(4)] Let $W'=W_L\oplus W_R$ be a decoupled local perturbation of $W$. Then $\sixL(W)=\six(W_L)$ and  $\sixR(W)=\six(W_R)$.
		\item[(5)] For a translation invariant $W$: $\six(W)=0$.
	\end{itemize}
\end{thm}

\begin{proof}
	\noindent(1)\ Let $t\mapsto W(t)$ be a norm continuous path of admissible essentially unitary walks. Then $\IM W_t$, $\Pl a \IM W_t\Pl a$, and $\Pg a \IM W_t\Pg a$ likewise depend continuously on $t$. They are also essentially gapped, when considered as operators on $\HH$, $\Pl a\HH$, and $\Pg a\HH$, respectively. Therefore by \Qref{pro:homoto}, the respective indices are constant along the path. When $W'$ is a compact perturbation of $W$,
	$(1-t)W+tW'=W+t(W'-W)$ is such a continuous path.
	
	\noindent(2)\ The right hand side is the symmetry index of $\Pl a\IM W\Pl a\oplus\Pg a \IM W\Pg a$, which differs from $\IM W$ by $\Pg aW\Pl a+\Pl aW\Pg a$. By the weakened locality condition in \Qref{sec:compactsetting} these are both compact operators. Hence the equality follows from (1).
	
	\noindent(3)\ Let $a<b$, and let $P_M=\Pg a\Pl b$ be the projection onto the cells $\HH_x$ with $a\leq x<b$. Then as in (2) we have that $P_MWP_M\oplus \Pg bW\Pg b$ is a compact perturbation of  $\Pg aW\Pg a$, so
	$\sixR_a(W)=\six(P_MWP_M)+\sixR_b(W)$, where the subscripts indicate the cut point used for the definition. The difference is the symmetry index of a symmetry-admissible Hamiltonian $\IM P_MWP_M$ on a finite dimensional space. Since the complement of the $0$-eigenspace is always balanced for such an operator, the difference term is just the index of the symmetry representation on the cells $[a,b)$, which we have assumed to be balanced from the outset.
	
	\noindent(4)\ Since local perturbations are compact, this is trivial, and is added here only to show the consistency of the indices defined in this section with the decoupling approach in \Qref{sec:twosettings}.
	
	\noindent(5)\ The Hamiltonian $\IM W$ can be diagonalized jointly with the translations, resulting in a set of eigenvalue functions $\veps_\alpha(k)$, $k\in[-\pi,\pi]$. When one of these functions has a zero (or is constantly equal to zero) $0$ must be in the essential spectrum of $\IM W$, contradicting admissibility. Hence the spectrum of $\IM W$ is strictly gapped.
\end{proof}

With every admissible essentially unitary walk or Hamiltonian we associate the pair of indices $\bigl(\sixL(W),\sixR(W)\bigr)$, which is clearly invariant under both homotopy and compact admissible perturbations. Because of the independence of the cut point we can compute $\sixR(W)$ as far to the right as we please. This is just about the opposite of a ``locally computable invariant'' \cite{OldIndex}: knowing a walk on just a finite piece allows no conclusion about the invariants whatsoever.

\subsection{Bulk-boundary correspondence}\label{sec:bulkboundary}

The bulk-boundary correspondence is the prediction of protected $\pm1$-eigenvalues at the interface between two bulks with distinct indices. It is characteristic for our theory that we can allow a very broad definition of ``bulk'' in this statement. Of course, translation invariance is one possibility. But it is also sufficient to take a disordered system in which only the statistical law for the local coins is translation invariant \cite{dynloc,dynlocalain}. Similarly, an almost periodic system may play the role of a bulk. What we need is only that a bulk system has a proper gap.
Then $\six_\pm(W)=0$ and $\sixL(W)+\sixR(W)=0$, so just one of these, say $\sixR(W)$,  is enough to
determine the invariants.

With this we are now able to formulate the
bulk-boundary correspondence:

\begin{cor}[Bulk-boundary correspondence]\label{cor:bulkedgecor}
	Suppose we have two bulk walks $W_R$ and $W_L$ and some $W$
	which is a unitary crossover of the two in the sense that $W$ coincides with $W_R$ far to the right, i.e.,
    $\lim_{a\to+\infty}\norm{\Pg a (W-W_R)\Pg a}=0$ and,
	similarly $W$ coincides with $W_L$ far to the left. Then
	\begin{equation}\label{bulkedge}
	\six(W)=\sixL(W)+\sixR(W)=-\sixR(W_L)+\sixR(W_R).
	\end{equation}
	Whenever $W_L$ and $W_R$ belong to different classes this is nonzero, and $W$ must have at least one eigenvector with eigenvalue $+1$ or $-1$. Moreover, the absolute value of this number is a lower bound on the
	dimension of the combined $\pm1$-eigenspaces.
\end{cor}

\begin{proof}
    We will first show that $\sixR(W)=\sixR(W_R)$. The argument for $\sixL(W)=\sixL(W_L)$ is completely analogous. By \Qref{thm:sixProps} (1) it is sufficient to show that $\Pg a(W-W_R)\Pg a$ is compact. Pick $b\in\Ir$ with $a<b$. Then
    \begin{equation*}
      \begin{split}
        \Pg a(W-W_R)\Pg a=\Pl b\Pg a(W-W_R)\Pg a\Pl b+\Pg b(W-W_R)\Pg a\Pl b\\
                            +\Pl b\Pg a(W-W_R)\Pg b+\Pg b(W-W_R)\Pg b.
      \end{split}
    \end{equation*}
    The first three terms on the right hand side each have finite rank, because they contain the factor $\Pl b\Pg a$, which is a finite rank projector. Hence $\Pg a(W-W_R)\Pg a-\Pg b(W-W_R)\Pg b=:K_b$ has finite rank. Since $\norm{\Pg b(W-W_R)\Pg b}\to0$ this gives an explicit norm approximation of $\Pg a(W-W_R)\Pg a$ by finite rank operators. This proves \eqref{bulkedge}.
	
Now $\six(W)$ is the index of the symmetry representation in the combined eigenspaces at $\pm1$, and for all symmetry types the absolute value of the index is a lower bound to the dimension of this space.  In particular, when $\sixR(W_R)\neq\sixR(W_L)$, we get a non-zero lower bound.
\end{proof}

Note that the theory does not predict whether the eigenvalues will be $+1$ or
$-1$. This depends on the index $\six_-(W)$, but from the given asymptotic data this value cannot be inferred: any compact non-gentle perturbation will produce a variant of engineering the crossover, for which the location of the protected eigenvalues is changed. Examples are given in \cite{letter,explorerA,UsOnTI}.

\section{Examples}\label{sec:examples}

\subsection{Translation invariant models: Generalities}
As we argued in the previous section, translation invariant systems naturally appear in applications as bulk systems (see also \Qref{sec:bulkboundary} and \Qref{sec:finite}). The simplification brought about by this assumption is considerable, allowing effectively explicit formulas for the index. Moreover, as we will describe in this section the classification reduces precisely \cite{UsOnTI} to the well-known one in terms of vector bundles over the quasi-momentum space \cite{Kita,kitaevPeriodic}. The relevant explicit formulas are then mostly known \cite{Kita,Asbo2} and will be described below.
An explicit derivation covering also additional aspects of the translation invariant case will be given in \cite{UsOnTI}.

The main simplification due to translation invariance is the possibility to partly diagonalize the walk by Fourier transform. After Fourier transform the Hilbert space becomes
$\mathcal L^2([-\pi,\pi),dk)\otimes\Cx^d$, where the interval $[-\pi,\pi)$, also called the Brillouin zone or the quasi-momentum space, should be considered as a parametrization of the circle.
We usually consider the Hilbert space as the space of $\Cx^d$-valued square integrable functions on $[-\pi,\pi)$.
The walk then acts by multiplying $\psi(k)$ with a unitary operator $\Wh(k)$ on $\Cx^d$. For a strictly local walk each entry of $\Wh(k)$ is a polynomial in $\exp(\pm ik)$. Our standard assumption of essential locality translates exactly \cite{UsOnTI} into the continuity of $k\mapsto\Wh(k)$ on the circle, i.e., with periodic boundary condition. The essential gap condition means that $\pm1$ are not in the spectrum of any $\Wh(k)$. For the action of symmetries one has to take into account that antiunitary symmetries also reverse the sign of quasi-momentum $k$. Thus, if $\sigma$ is an antiunitary symmetry, which acts in each cell as the finite dimensional antiunitary operator $\sigma_1$, we have $\widehat{\sigma W\sigma^*}(k)=\sigma_1 \Wh(-k)\sigma_1^*$. Thus the symmetry conditions become
\begin{equation}\label{tisymm}
  \ph_1\Wh(k)\ph_1^*=\Wh(-k),\quad \rv_1\Wh(k)\rv_1^*=\Wh(-k)^*,\ \text{and}\ \ch_1\Wh(k)\ch_1^*=\Wh(k)^*.
\end{equation}
The index 1 for ``single cell'' will be omitted in the sequel.

Let $B(k)$ denote the upper band projection, i.e., the eigenprojection of  $\Wh(k)$ for the eigenvalues with positive imaginary part. Because of the gap condition $k\mapsto B(k)$ is also continuous. Since the continuous functional calculus preserves essential locality, we can deform all eigenvalues of $\Wh(k)$ to $\pm i$, and get the {\it flat-band walk} with $\Wh_\flat(k)=iB(k)-i(\idty-B(k))=2iB(k)-i\idty$. Being homotopic, $W$ and $W_\flat$ have all the same invariants. This underlines that the band projection $B(k)$ as a function of $k$, i.e., the Hermitian vector bundle over the circle (parametrized by $k$) with fiber $B(k)\Cx^d$ is the key geometric object to look at. Indeed this is the starting point for the classification of symmetric Hamiltonians in terms of K-Theory.

It is a remarkable fact that in spite of the much larger flexibility of general (not translation invariant) walks the homotopy classes are the same. Given our theory, this is actually easy to show. For the sake of discussion denote by $\igS$ the classifying group of symmetry type from our \Qref{Tab:sym}, and by $\igtiS$ the classifying group for translation invariant systems, based on the classification of vector bundles \cite{kitaevPeriodic}. By inspection the abstract groups are isomorphic, but it may be good to briefly look at the mathematical questions involved. Since every translation invariant walk $W$ is assigned an index $\sixR(W)\in\igS$, which is unchanged under homotopy, there must be a map $\iota:\igtiS\to\igS$. Since on both sides addition is defined in terms of direct sums of walks, $\iota$ must be a homomorphism for addition. Now $\iota$ might not be onto, because some classes in $\igS$ might fail to be realizable by translation invariant $W$. This possibility is easily dispelled by providing for each type $\symS$ a translation invariant walk $W$ so that $\sixR(W)$ generates $\igS$ (see next subsection). Thus all of our classes contain a translation invariant example. On the other hand, two translation invariant walks might be deformable into each other only by breaking translation invariance along the way. These would correspond to distinct elements in $\igtiS$, so $\iota$ might fail to be injective. In fact, this does happen, in a way, for symmetry type $\symD$. More precisely, we can construct \cite{UsOnTI} translation invariant walks $W,W'$ on the same fixed cell structure, which cannot be deformed into each other by keeping translation invariance, but can be if we group neighbouring cells, and keep only invariance by even translations. One can also find $W''$ such that $W\oplus W''$ can be deformed to $W'\oplus W''$. Thus the Grothendiek construction used in the definition of $\igtiS$ enforces a cancellation law, so that $W$ and $W'$ correspond to the same element after all. The conventional term for this is ``stable homotopy'', which is thus distinct from ordinary homotopy for type \symD, but apparently not for the other types \cite{UsOnTI}.

\subsection{The generating example}\label{sec:genex}

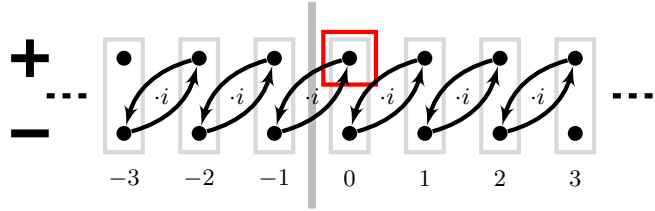
\begin{figure}
	\tikzstyle{block} = [minimum width = .15, minimum height = 0.15]
\tikzstyle{dot} = [circle,fill,inner sep = 2]
\tikzstyle{i} = [rectangle,fill=gray!30,fill opacity=0, text opacity=1, inner sep=.2, outer sep=0]
\tikzstyle{arr} = [->, ultra thick]

\begin{tikzpicture}
	[
	>=latex',auto,
		]
		\draw[line width=3pt] (-4.25,2) +(-.25,0) -- +(.25,0) +(0,-.25) -- +(0,.25);
		\draw[line width=3pt] (-4.25,1) +(-.25,0) -- +(.25,0);

		\draw[line width=2pt,dashed] (-3.5,1.5) -- (-4.1,1.5);
		\draw[line width=2pt,dashed] (3.5,1.5) -- (4.1,1.5);

		\foreach \i in {-3,-2,...,3}{
			\node[block] (a\i) at (\i,0.4) {$\i$};
			\node[dot] (u\i) at (\i,2) {};
			\node[dot] (d\i) at (\i,1) {};	
			
			\draw[gray!30,ultra thick] 
			({\i-.25},.75) -- ({\i+.25},.75) --
			({\i+.25},2.25) -- ({\i-.25},2.25) -- cycle;
		}

		\draw[line width=3pt, gray!50] (-.5,0) -- (-.5,2.75);
	
		\foreach \i in {-3,-2,...,2}{
			\node[i] at ({\i+.5},1.5) {$\cdot i$};	
		}
	
		\draw[red,ultra thick] (-.35,1.65) rectangle +(.7,.7);

		\path[arr] (u-2) edge[bend right] node [left] {} (d-3);
		\path[arr] (u-1) edge[bend right] node [left] {} (d-2);
		\path[arr] (u0) edge[bend right] node [left] {} (d-1);
		\path[arr] (u1) edge[bend right] node [left] {} (d0);	
		\path[arr] (u2) edge[bend right] node [left] {} (d1);
		\path[arr] (u3) edge[bend right] node [left] {} (d2);
		
		\path[arr] (d-3) edge[bend right] node [left] {} (u-2);
		\path[arr] (d-2) edge[bend right] node [left] {} (u-1);
		\path[arr] (d-1) edge[bend right] node [left] {} (u0);
		\path[arr] (d0) edge[bend right] node [left] {} (u1);	
		\path[arr] (d1) edge[bend right] node [left] {} (u2);
		\path[arr] (d2) edge[bend right] node [left] {} (u3);

\end{tikzpicture}
	\caption{Schematic representation of the generating example walk. It acts by swapping the cells indicated by arrows and multiplying by $i$.  The red box marks the kernel of $PWP$, where $P$ is the projection onto the cells to the right of the gray divide.}
	\label{fig:Wgenerate}
\end{figure}

As the simplest example, which actually generates all index groups, we consider the walk on $\ell^2(\Ir)\otimes\Cx^2$, shown in \Qref{fig:Wgenerate}
 (compare also \cite[Fig. 6]{Asbo1}). Its action in a basis $\ket{x,s}$ with $x\in\Ir$ and $s=\pm1$ is
\begin{equation}\label{gentwalk}
  W\ket{x,+1}=i\ket{x-1,-1}\qquad\mbox{and}\qquad W\ket{x,-1}=i\ket{x+1,+1}.
\end{equation}
That is, $W$ swaps pairs of basis elements belonging to neighbouring cells. The factor $i$ ensures that $W^2=-\idty$, so the eigenvalues of $W$ are $\pm i$. Clearly, $W$ is translation invariant. For involutive symmetries we have firstly $\rv$ as complex conjugation in the position space basis. Hence $\rv^2=\idty$, and $\rv W=-W\rv=W^*\rv$. The chiral symmetry acts like $\sigma_3$ in each cell, i.e., $\ch\ket{x,s}=s\ket{x,s}$. Then $\ch^2=\idty$, and $\ch W=-W\ch=W^*\ch$. Hence this is a walk of symmetry type \symBDI. For comparison with the formulas based on translation invariance, we will need the walk matrix in momentum space, which is
\begin{equation}\label{Wgen}
  \Wh(k)= \begin{pmatrix}
    0&i e^{ik}\\i e^{-ik}&0
  \end{pmatrix}.
\end{equation}

Let us compute the index $\sixR(W)$ according to \eqref{sixRL}, which requires first to determine the Hermitian operator
\begin{equation}\label{imgent}
  \IM W=\frac1{2i}\,{(W-W^*)}=\frac1{i}\,W.
\end{equation}
We are interested in the null space of $P(\IM W)P$. The pairs $(\ket{x,+1},\ket{x-1,-1})$ with $x>0$ are unaffected by $P$.  However, $\ket{0,+1}$ is mapped by $(\IM W)P$ to $\ket{-1,-1}$ and to zero by $P$. This spans the null space of $P(\IM W)P$. On this space $\ch=+1$, so for the symmetry type \symBDI, we get $\sixR(W)=1$. By \Qref{Tab:forget} we obtain the indices for symmetry type \symAIII{} (by keeping $\ch$) and \symD{} (by keeping $\ph$). In these cases, we also get  $\sixR(W)=1$ where ``$1$'' is the generator of the respective group $\Ir$ or $\Ir_2$.

This leaves two symmetry types with non-trivial index group, namely \symCII{} and \symDIII. Note that these have in some sense doubled indices, which suggests doubled versions of the above $W$ as generating examples. For symmetry type \symCII, we choose $W_D=W\oplus W$, which is chirally symmetric with respect to $\ch_D=\ch\oplus\ch$. In order to get a \symCII{} walk, we need a particle-hole symmetry with $\ph_D^2=-\idty$, which can be constructed by setting $\ph_D=\left( \begin{smallmatrix}0 & -\ph\\ \ph & 0\end{smallmatrix}\right)$, where $\ph$ is the symmetry used for $W$. Then the admissibility condition $\ph W=W\ph$ for $W$ readily implies the one for $W_D$. By construction, the index $\sixR(W_D)=2$ can then trivially be obtained by the previous example, and generates  the required index group $2\Ir$.

Similarly, for symmetry type \symDIII, we set  $W_D=W\oplus W^*$, which is $\ph_D$-symmetric with $\ph_D=\ph\oplus\ph$. We get a chiral symmetry with $\ch_D^2=-\idty$ by setting $\ch_D=\left( \begin{smallmatrix}0& -\idty\\ \idty & 0\end{smallmatrix}\right)$. Again, the index $\sixR(W_D)$ can easily be obtained from the considerations for $W$, and generates the index group $2\Ir_2$.

\subsection{Index formulas}
We now provide simple procedures for computing the index $\sixR(W)$ from the matrix valued function $k\mapsto\Wh(k)$. In each case we make sure that the procedure is homotopy invariant, behaves correctly under direct sums, and applies even if we only know $\Wh$ to be continuous. With differentiability one could provide also some interpretations by integrals over a Berry connection but we leave that to our more detailed work \cite{UsOnTI}.

\subsubsection{Types \symAIII, \symBDI, and \symCII}\label{sec:tixchiral}
In these cases we have a chiral symmetry $\ch$ with $\ch^2=\idty$. We can write $\Wh(k)$ in block matrix form with respect to the eigenspaces of $\ch$. Since each cell is to be balanced, $\tr\gamma=0$. Hence the eigenspaces have the same dimension, and the off-diagonal block $\Wh_{12}(k)$ is a square matrix. It turns out (see, e.g.\ \cite{UsOnTI}) that $W$ has gaps at $\pm1$ if and only if $\Wh_{12}(k)$ is invertible for all $k$. Therefore,
the curve $c:[-\pi,\pi]\to\Cx$ given by $c(k)=\det\Wh_{12}(k)$ is continuous (with periodic boundary condition) and does not pass through the origin. Hence the winding number $\wind(c)$ of this curve is defined. It does not depend on the bases chosen in the two eigenspaces of $\ch$ for extracting a matrix from the operator $\Wh(k)$. Indeed if $S_1,S_2$ are the unitary operators of such basis changes, $\Wh_{12}(k)$ is replaced by $S_1\Wh_{12}(k)S_2^*$, which modifies $c(k)$ by a constant factor and does not change the winding number.
Since $c(k)$ depends continuously on $W$ (taken in the norm topology), it is a homotopy invariant. For the direct sum of walks we get the direct sum also of the chirally off-diagonal blocks, hence the product of the determinants, and hence the sum of the winding numbers. And finally, the generating example \eqref{Wgen} is already written in the chiral basis, so $\wind(c)=1$. This shows that for the three symmetry types mentioned we have
\begin{equation}\label{tixChiral}
    \sixR(W)=\wind(c).
\end{equation}

\subsubsection{Type \symCII}\label{sec:tixCII}
The formula of the previous section holds also for \symCII. But according to \Qref{Tab:sym} the invariant should always be even. So we need to show that this is automatically true for the winding number \eqref{tixChiral} given the additional symmetry.

\begin{proof}
We can simplify the problem by deforming the bands to be flat, i.e., we can assume $\Wh_{12}(k)$ to be unitary.  Since $\ph$ commutes with $\ch$ it acts separately in each chiral block. These have the same even dimension, so with respect to suitable bases the $\ph$-operators in each block can be chosen to be the same. So with \eqref{tisymm} we are left to show (in streamlined notation): If $k\mapsto U(k)$ is a continuous function of unitaries, and $\ph$ antiunitary with $\ph^2=-1$ such that $\ph U(k)\ph^*=U(-k)$ then $c(k)=\det U(k)$ has even winding number.

Now the symmetry condition immediately implies $\overline{c(k)}=c(-k)$. When we write $c(k)=\exp(2\pi i a(k))$ with $a$ continuous on $[-\pi,\pi]$ the winding number is just $a(\pi)-a(-\pi)$, and the symmetry condition
requires that $a(k)+a(-k)$ is constant. Thus $\wind(c)=2(a(\pi)-a(0))$.  This is indeed even, provided we can show that $a(\pi)-a(0)$ is an integer or, equivalently, that $c(\pi)=c(0)$.

We will thus complete the proof by showing that all unitaries $U$ with $\ph U\ph^*=U$ have the same determinant, namely $1$. The non-real eigenvalues of such $U$ come in complex conjugate pairs, which do not contribute to the determinant. So we only need to look at the $-1$-eigenspace. Since this is left invariant by $\ph$ and $\ph^2=-\idty$ this eigenspace must be even dimensional, so again the contribution to the determinant is $1$.
Note that the only difference to the \symBDI\ case is this very last sentence.
\end{proof}

\subsubsection{Type \symD}\label{sec:tixD}
The invariant is described most easily in terms of the corresponding flat-band walk $\Wh_\flat(k)=2iB(k)-i\idty$, which satisfies $\Wh_\flat(k)^*=-\Wh_\flat(k)$. The symmetry constraint $\ph\Wh_\flat(k)\ph^*=\Wh_\flat(-k)$ then imposes that for $k=0$ and $k=\pi\equiv-\pi$, $\Wh_\flat(k)$ is real and antisymmetric in any basis in which $\ph$ is the complex conjugation.

Let us denote by $\AA$ the set of real antisymmetric unitaries. For such matrices the Pfaffian $\pfaff(U)$ is a well defined polynomial in the matrix elements such that $\det(U)=\pfaff(U)^2$. Since on $\AA$ $\pf(U)$ is obviously real valued, we find that $\pf(U)=\pm1$ for $U\in\AA$. Since $\pfaff$ is continuous this identifies two connected components of $\AA$. We claim that
\begin{equation}\label{tixD}
  (-1)^{\sixR(W)}=\frac{\pfaff(\Wh_\flat(\pi))}{\pfaff(\Wh_\flat(0))}.
\end{equation}
That is, if we write the index group $\ig(\symD)=(\Ir_2,+)$, the index is $0$ when $\Wh_\flat(0)$ and $\Wh_\flat(\pi)$ are in the same connected component of $\AA$, and $1$ otherwise. Obviously, this a homotopy invariant, due to the continuity of the Pfaffian. The addition formula for direct sums follows because the Pfaffian is multiplicative over direct sums. Moreover, for the generating walk \eqref{Wgen} we have $W=W_\flat$ and $\Wh(0)=-\Wh(\pi)$. When expressed in an $\ph$-real basis these matrices become unitary, real, and antisymmetric and equal to the only two such matrices in two dimensions, namely $\pm i\sigma_2$. That is, the index must be $1$.  This can also be seen by determining the index with respect to the symmetry type \symBDI, and forgetting the symmetries other than $\ph$.

Comparing with the chiral symmetry types it is remarkable that the classification in this case depends only on the values $\Wh$ for $0$ and $\pi$ and not on the connecting path. This can be seen directly as follows. For a flat-band walk we have $\det\Wh_\flat(k)\equiv1$, so the connecting path lies in $SU_d$. But then, by the simple connectedness of $SU_d$, any two such connections can be deformed into each other, and hence must give the same index.

\subsubsection{Type \symDIII}\label{sec:tixDIII}
The formula for this symmetry type combines the ideas of the previous two subsections. First of all, it suffices to treat the flat-band case. The block decompositions of the relevant operators are then
\begin{equation}\label{blockDIII}
  \ch=\begin{pmatrix}i&0\\0&-i\end{pmatrix} ,\quad
  \eta=\begin{pmatrix}0&K\\K&0\end{pmatrix} ,\quad
  \Wh(k)=\begin{pmatrix}0&Z(k)\\-Z(k)^*&0\end{pmatrix} ,\quad
\end{equation}
where $K$ is complex conjugation in a suitable basis, and $Z(k)$ is unitary for all $k$. The chiral symmetry is guaranteed by this form, and the particle-hole symmetry becomes $\overline{Z(k)}=-Z(-k)^*$. Hence  $Z(0)$ and $Z(\pi)$ are antisymmetric (although not necessarily real). Now consider the {\it antisymmetric closure} of $Z$, i.e., a continuous path $Z$ in the space of unitary matrices so that $Z(k)$ is the given function for $k\in[0,\pi]$, and $Z(k)$ is antisymmetric for $k\in[\pi,4]$ with  $Z(4)=Z(0)$. As in the chiral case we consider the curve $c(k)=\det Z(k)$ for $k\in[0,4]$ in the complex plane, which avoids the origin, because $Z(k)$ is everywhere unitary. We then set
\begin{equation}\label{tixDIII}
  \sixR(W)\equiv2\,\wind(c)\quad \mod4.
\end{equation}
We have to show that the antisymmetric closure exists, i.e., the set of antisymmetric unitaries is connected. This follows by using that the eigenvalues of such a matrix $U$ come in pairs $\pm\lambda$, and therefore $U=VU_0V^T$, where $V$ is an arbitrary unitary, and $U_0$ is a direct sum of copies of $i\sigma_2$ \cite{Zumino}. By contracting $V$ to the identity we find that all $U$ are connected to $U_0$. The above formula makes sense only if we can show that the result is independent of the choice of antisymmetric closure. Now the difference of the winding numbers of two closures is the winding number of a closed path running entirely in the antisymmetric manifold. On this path the Pfaffian $\pfaff(Z(k))$ is well defined and non-zero. If its winding number is $n$, $c(k)=\det(Z(k))=\pfaff(Z(k))^2$ has winding number $2n$, hence this number is well-defined modulo $2$. The factor $2$ is introduced for conformity with \Qref{Tab:sym}. Another way to put this is to say that the index of a walk is trivial if and only if it can be deformed to a walk with $Z(k)$ antisymmetric for all $k$.

Note that the direct sum property holds for this formula as well, and the forgetting relations are respected.

\subsection{Splitting the Split-Step Walk}\label{sec:splitss}
This example was introduced in \cite{Kita} and has become a testbed for many questions of this theory \cite{Kita2,Asbo1,Asbo2,Asbo4}. We use it here in a generalized (not translation invariant) form to demonstrate some effects of gentle vs.\ non-gentle decouplings.
The model has symmetry type \symBDI, and is of the form
\begin{equation}\label{eq:wss}
  W=BS_{\downarrow}AS_{\uparrow}B,
\end{equation}
where $A=\bigoplus_x A_x$ and $B=\bigoplus_x B_x$ are unitary operators acting sitewise with $\HH_c=\Cx^2$. The operations $S_{\uparrow}$ and $S_{\downarrow}$ represent the right shift of the spin-up vectors, and the left shift of the spin-down vectors, respectively. The chiral symmetry takes the form $\ch=\bigoplus_x \sigma_1$ and $\ph$, if applicable, is given by complex conjugation. Then one easily checks that both shifts commute with $\ph$, and
\begin{equation}
  \ch S_\downarrow=S_\uparrow^*\ch,\qquad\ch S_\uparrow=S_\downarrow^*\ch.
\end{equation}
With this relation one easily verifies that $W$ is admissible for the global \symBDI symmetry, provided each $A_x$ and $B_x$ is admissible for the symmetries acting in each cell, which we will assume from now on. For the standard translation invariant model of this kind we take $A_x=R(\theta_2)$ and $B_x=R(\theta_1/2)$, where $R(\theta)$ denotes the standard real rotation matrix by angle $\theta$. The parameter torus is shown in \Qref{fig:splitsse}. It is partitioned into regions with different index, separated by lines at which a spectral gap closes at $+1$ or at $-1$. An interactive version of this plane, where one can observe the changes in the dispersion relation, and the curve $c$ from \Qref{sec:tixchiral}, is provided at \cite{sse}.

It is suggestive to use also the unitarily equivalent form $B^*WB=(S_{\downarrow}A)(S_{\uparrow}B^2)=:W_{\downarrow}W_{\uparrow}$. This amounts to a site-dependent local basis change, which also makes the chiral symmetry site dependent. Naturally, this does not change the index. However, we can also exchange the two factors (change the ``time frame''), and consider $W_{\uparrow} W_{\downarrow}=W_{\downarrow}^*WW_{\downarrow}$, which is again unitarily equivalent. However, this operation does change the index, because the operator $W_{\downarrow}$ does not respect the cell structure. This turns the walk depicted in \Qref{fig:Wgenerate} into one, in which the swapping happens only within each cell, so the index becomes zero. More generally, the phase plane \Qref{fig:splitsse} is rotated by this operation by 90 degrees, so it exactly exchanges walks with trivial index and walks with non-trivial index.

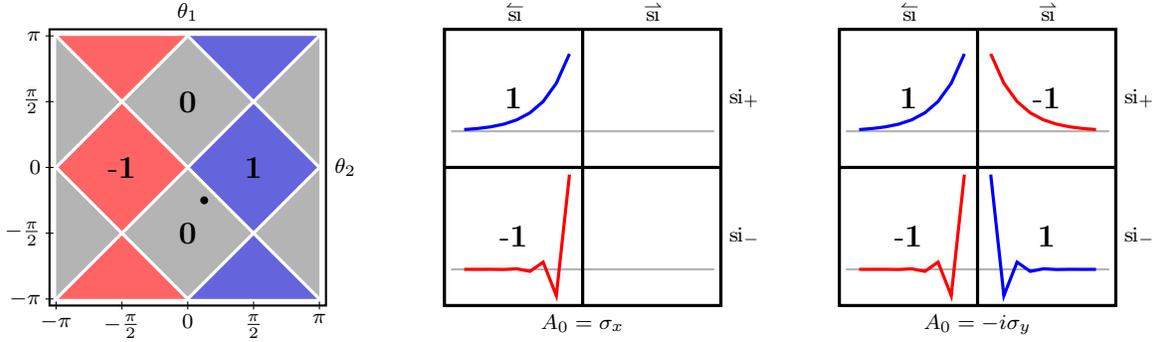
\begin{figure}
	\tikzset{
		>=stealth',
		big arrow/.style={
			very thick,
			postaction={decorate}
			},
		big arrow reverse/.style={
			very thick,
			decoration={markings, mark=at position 0.17 with {\arrow[thick, scale=2]{>}},
				mark=at position 0.5 with {\arrow[thick,scale=2]{>}},
				mark=at position 0.83 with {\arrow[thick,scale=2]{>}}},
			postaction={decorate}
			},
		mid graphic/.style={
			xshift=3cm
			},
		right graphic/.style={
			xshift=6cm
			}
	}
	
	\begin{tikzpicture}
	[
	scale=1.75,
	font=\footnotesize
		]
		
		\definecolor{pmcol}{RGB}{180,180,180}
		\definecolor{nncol}{RGB}{180,180,180}
		\definecolor{nmcol}{RGB}{100,100,220}
		\definecolor{pncol}{RGB}{255,100,100}

		\draw[thick] (-1.05,-1.05)  rectangle +(2.1,2.1);
		
		\fill[nncol] 	(-.5,-.5) -- +(-.5,-.5) -- +(-.5,.5)
								(-.5,-.5) -- ++(.5,.5) -- ++(.5,-.5) -- +(-.5,-.5)
								(.5,-.5) -- +(.5,.5) -- +(.5,-.5);
		\fill[pncol] 	(-.5,-.5) -- +(-.5,-.5) -- +(.5,-.5)
								(-.5,-.5) -- ++(.5,.5) -- ++(-.5,.5) -- +(-.5,-.5)
								(-.5,.5) -- +(.5,.5) -- +(-.5,.5);		
		\fill[pmcol] 	(-.5,.5) -- +(-.5,-.5) -- +(-.5,.5)
								(-.5,.5) -- ++(.5,.5) -- ++(.5,-.5) -- +(-.5,-.5)
								(.5,.5) -- +(.5,.5) -- +(.5,-.5);
		\fill[nmcol] 	(.5,-.5) -- +(-.5,-.5) -- +(.5,-.5)
								(.5,-.5) -- ++(.5,.5) -- ++(-.5,.5) -- +(-.5,-.5)
								(.5,.5) -- +(.5,.5) -- +(-.5,.5);
								
		\draw[white,very thick] (-1,-1) -- (1,1);
		\draw[white,very thick] (1,-1) -- (-1,1);
		\draw[white,very thick] (-1,0) -- (0,1) (0,-1) -- (1,0);
		\draw[white,very thick] (-1,0) -- (0,-1) (0,1) -- (1,0);
		
		\foreach \i in {-1,-.5,0,.5,1}{
			\draw[align=left] (-1.02,{\i}) -- (-1.08,{\i});
			\draw[align=left] ({\i},-1.02) -- ({\i},-1.08);
		}
		
		\draw (-1.05,-1)  node[left,align=left]{ $-\pi$};
		\draw (-1.05,-.5)  node[left,align=left]{ $-\frac{\pi}{2}$};
		\draw (-1.05,0)  node[left,align=left]{ $0$};
		\draw (-1.05,.5)  node[left,align=left]{ $\frac{\pi}{2}$};
		\draw (-1.05,1)  node[left,align=left]{ $\pi$};
		
		\draw (-1,-1.05)  node[below,align=center]{ $-\pi$};
		\draw (-.5,-1.05)  node[below,align=center]{ $-\frac{\pi}{2}$};
		\draw (0,-1.05)  node[below,align=center]{ $0$};
		\draw (.5,-1.05)  node[below,align=center]{ $\frac{\pi}{2}$};
		\draw (1,-1.05)  node[below,align=center]{ $\pi$};
		
		\fill (0.125, -.25)  circle (.03) ;

		\draw (0,1.05) node[above,align=center]{ $\theta_1$};
		\draw (1.05,0) node[right,align=left]{ $\theta_2$};
		
		
		\draw (.5,0) node[align=center]{\large$\bf{1}$};
		\draw (-.5,0) node[align=center]{\large$\bf{\text{-}1}$};
		\draw (0,-.5) node[align=center]{\large$\bf 0$};
		\draw (0,.5) node[align=center]{\large$\bf 0$};

		\draw[mid graphic, very thick] (-1.05,-1.05)  rectangle +(2.1,2.1);
		\draw[mid graphic, very thick] (-1.05,0) -- +(2.1,0) (0,-1.05) -- +(0,2.1);
		\draw[mid graphic, thick, opacity=0.3] (-1,0.275) -- (1,0.275);
		\draw[mid graphic, thick, opacity=0.3] (-1,-0.775) -- (1,-.775);
		\draw[mid graphic] (0,-1.05) node[below, align=center] {$A_0=\sigma_x$};
		\draw[mid graphic] (1.05,.525) node[right,align=center] {$\six_+$};
		\draw[mid graphic] (1.05,-.525) node[right,align=center] {$\six_-$};
		\draw[mid graphic] (0.525,1.05) node[above,align=center] {$\sixR$};
		\draw[mid graphic] (-0.525,1.05) node[above,align=center] {$\sixL$};
		
		\draw[right graphic, very thick] (-1.05,-1.05)  rectangle +(2.1,2.1);
		\draw[right graphic, very thick] (-1.05,0) -- +(2.1,0) (0,-1.05) -- +(0,2.1);
		\draw[right graphic, thick, opacity=.3] (-1,0.275) -- (1,0.275);
		\draw[right graphic, thick, opacity=0.3] (-1,-0.775) -- (1,-.775);
		\draw[right graphic] (0,-1.05) node[below, align=center] {$A_0=-i\sigma_y$};
		\draw[right graphic] (1.05,.525) node[right,align=center] {$\six_+$};
		\draw[right graphic] (1.05,-.525) node[right,align=center] {$\six_-$};
		\draw[right graphic] (0.525,1.05) node[above,align=center] {$\sixR$};
		\draw[right graphic] (-0.525,1.05) node[above,align=center] {$\sixL$};

		\draw[mid graphic,very thick,blue] plot coordinates {(-0.9, 0.287835) (-0.8, 0.295705) (-0.7, 0.308399) (-0.6, 0.328877) (-0.5, 0.361911) (-0.4, 0.415198) (-0.3, 0.501158) (-0.2, 0.639821) (-0.1, 0.863502)};
		\draw[mid graphic,very thick,red] plot coordinates {(-0.9, -0.774975) (-0.8, -0.77509) (-0.7, -0.774676) (-0.6, -0.77617) (-0.5, -0.770771) (-0.4, -0.790279) (-0.3, -0.719794) (-0.2, -0.974468) (-0.1, -0.0542976)};

		\draw[right graphic,very thick,red] plot coordinates {(0.1, 0.863502) (0.2, 0.639821) (0.3, 0.501158) (0.4, 0.415198) (0.5, 0.361911) (0.6, 0.328877) (0.7, 0.308399) (0.8, 0.295705) (0.9, 0.287835)};	
		\draw[right graphic,very thick,blue] plot coordinates {(0.1, -0.0542976) (0.2, -0.974468) (0.3, -0.719794) (0.4, -0.790279) (0.5, -0.770771) (0.6, -0.77617) (0.7, -0.774676) (0.8, -0.77509) (0.9, -0.774975)};
		\draw[right graphic,very thick,blue] plot coordinates {(-0.9, 0.287835) (-0.8, 0.295705) (-0.7, 0.308399) (-0.6, 0.328877) (-0.5, 0.361911) (-0.4, 0.415198) (-0.3, 0.501158) (-0.2, 0.639821) (-0.1, 0.863502)};
		\draw[right graphic,very thick,red] plot coordinates {(-0.9, -0.774975) (-0.8, -0.77509) (-0.7, -0.774676) (-0.6, -0.77617) (-0.5, -0.770771) (-0.4, -0.790279) (-0.3, -0.719794) (-0.2, -0.974468) (-0.1, -0.0542976)};

		\draw[mid graphic] (-.525,.525) node[align=center] {\large$\bf{1}$};
		\draw[mid graphic] (.525,.525) node[align=center] {};
		\draw[mid graphic] (-.525,-.525) node[align=center] {\large$\bf{\text{-}1}$};
		\draw[mid graphic] (.525,-.525) node[align=center] {};
		\draw[right graphic] (-.525,.525) node[align=center] {\large$\bf{1}$};
		\draw[right graphic] (.525,.525) node[align=center] {\large$\bf{\text{-}1}$};
		\draw[right graphic] (-.525,-.525) node[align=center] {\large$\bf{\text{-}1}$};
		\draw[right graphic] (.525,-.525) node[align=center] {\large$\bf{1}$};

	\end{tikzpicture}
	\caption{\label{fig:splitsse} \emph{Left:} parameter plane for the split-step walk \eqref{eq:wss}. The plane is split into regions of different symmetry index $\sixR$. The white lines represent the parameter configurations, for which the gap closes.\\
	\emph{Mid:} eigenfunctions for a decoupled split-step walk with parameter configuration $(\theta_1,\theta_2)=(\pi/8,-\pi/4)$ (black dot in left graphic). The decoupling coin at $x=0$ is chosen as $A_0=\sigma_x$. The eigenfunctions are positioned according to their corresponding contribution to the index (see index matrix \eqref{si22cross}). Top and bottom refer to walk-eigenvalues $+1$ and $-1$ respectively; left and right refer to the respective half-line. The label in each quadrant denotes the chirality of the respective eigenfunction and hence the value of the respective index contribution. Since $\sigma_x$ is a non-gentle perturbation, the pair $(\six_+,\six_-)=(+1,-1)$ (row-sums) is non-zero, whereas $(\sixR,\sixL)=(0,0)$ (column-sums), being insensitive to any compact perturbation, is zero.\\
	\emph{Right:} eigenfunctions for the same coin configuration, but now with the decoupling coin $A_0=R(\pi/2)=-i\sigma_y$. Since this is a gentle perturbation, both pairs $(\six_+,\six_-)=(\sixR,\sixL)=(0,0)$ are zero.}
\end{figure}

The family of split-step walks is also very suitable for discussing decouplings. In this case we just change one of the ``coin unitaries'' $A_x$, say $A_0$ to a so-called {\it splitting coin}, namely either
$A_0=\pm R(\pi/2)=\pm i\sigma_2$ or $A_0=\pm\sigma_1$. In either case one checks readily that $S_{\downarrow}AS_{\uparrow}$ leaves $\bigoplus_{x=0}^\infty\HH_x$ invariant. However, these two options are rather different. The family of \symBDI-admissible unitaries consists precisely of the rotations $R(\theta)$ and $\pm\sigma_1$, which are not connected to the rotations. Therefore, choosing a rotation to decouple the walk is a gentle perturbation, which leaves $\six_\pm(W)$ invariant. That is, the eigenvalues at $+1$ (resp., $-1$) must come in chirally opposite pairs. On the other hand, for the non-gentle decoupling by $A_0=\pm\sigma_1$ only the overall set of eigenvectors must be chirally balanced. This is shown in \Qref{fig:splitsse}, and can also be followed dynamically in \cite{sse}.

\subsection{Arbitrarily varying parameters}\label{sec:ssArbitrary}
What happens, if we let the angles $\theta_1(x)$ and  $\theta_2(x)$ vary arbitrarily, say in an interval? If the intervals are too large, so that they allow homogeneous phases of different index, then under such conditions we could produce arbitrarily many phase boundaries, and the essential gap would close. So let us consider two constants $\eps_1,\eps_2$ and suppose that one of the following cases applies:
\begin{equation}\label{varytable}\def\pito{\textstyle\frac\pi2}
  \begin{array}{rll}
    \mbox{Case 1: for all $x$,\ }&\theta_1(x)\in\left[-\pito-\eps_1,-\pito+\eps_1\right]& \mbox{and\ }\theta_2(x)\in\left[-\eps_2,+\eps_2\right]\cup\left[\pi-\eps_2,\pi+\eps_2\right]\\
    \mbox{Case 2: for all $x$,\ }&\theta_1(x)\in\ \left[\pito-\eps_1,\pito+\eps_1\right]& \mbox{and\ }\theta_2(x)\in\left[-\eps_2,+\eps_2\right]\cup\left[\pi-\eps_2,\pi+\eps_2\right]\\
    \mbox{Case 3: for all $x$,\ }&\theta_1(x)\in\left[-\eps_1,+\eps_1\right]\cup\left[\pi-\eps_1,\pi+\eps_1\right]& \mbox{and\ }\theta_2(x)\in\left[-\pito-\eps_2,-\pito+\eps_2\right]\\
    \mbox{Case 4: for all $x$,\ }&\theta_1(x)\in\left[-\eps_1,+\eps_1\right]\cup\left[\pi-\eps_1,\pi+\eps_1\right]& \mbox{and\ }\theta_2(x)\in\ \left[\pito-\eps_2,\pito+\eps_2\right]
  \end{array}
\end{equation}
We then claim that as long as
\begin{equation}\label{epsrect}
  \sin\frac{\eps_1}2+\sin\frac{\eps_2}2< \sqrt2,
\end{equation}
the corresponding walks are gapped, and have \symBDI-index $\sixR(W)=-1/+1/0/0$ in case 1/2/3/4, respectively. The allowed regions for the angles are shown in \Qref{fig:ssrectangles}.

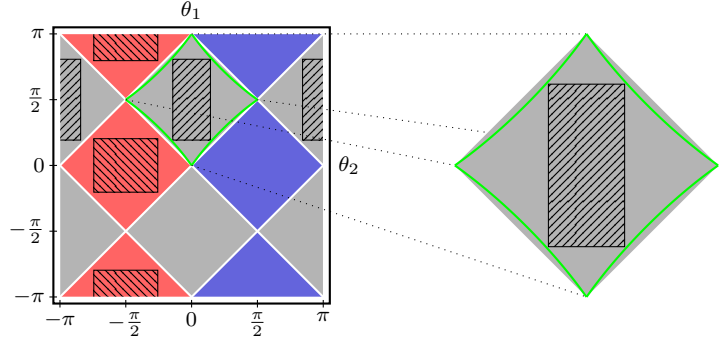
\begin{figure}
	\tikzset{
		>=stealth',
		big arrow/.style={
			very thick,
			postaction={decorate}
			},
		big arrow reverse/.style={
			very thick,
			decoration={markings, mark=at position 0.17 with {\arrow[thick, scale=2]{>}},
				mark=at position 0.5 with {\arrow[thick,scale=2]{>}},
				mark=at position 0.83 with {\arrow[thick,scale=2]{>}}},
			postaction={decorate}
			},
		right graphic/.style={
			xshift=3cm
		}
	}
	
	\begin{tikzpicture}
	[
	scale=1.75,
	font=\footnotesize
		]
		\definecolor{pmcol}{RGB}{180,180,180}
		\definecolor{nncol}{RGB}{180,180,180}
		\definecolor{mmcol}{RGB}{220,220,220}
		\definecolor{nmcol}{RGB}{100,100,220}
		\definecolor{pncol}{RGB}{255,100,100}
		\definecolor{linecol}{RGB}{0,255,0}
		
		\draw[thick] (-1.05,-1.05)  rectangle +(2.1,2.1);
		
	\fill[nncol] 	(-.5,-.5) -- +(-.5,-.5) -- +(-.5,.5)
		(-.5,-.5) -- ++(.5,.5) -- ++(.5,-.5) -- +(-.5,-.5)
		(.5,-.5) -- +(.5,.5) -- +(.5,-.5);
	\fill[pncol] 	(-.5,-.5) -- +(-.5,-.5) -- +(.5,-.5)
		(-.5,-.5) -- ++(.5,.5) -- ++(-.5,.5) -- +(-.5,-.5)
		(-.5,.5) -- +(.5,.5) -- +(-.5,.5);		
	\fill[pmcol] 	(-.5,.5) -- +(-.5,-.5) -- +(-.5,.5)
		(-.5,.5) -- ++(.5,.5) -- ++(.5,-.5) -- +(-.5,-.5)
		(.5,.5) -- +(.5,.5) -- +(.5,-.5);
	\fill[nmcol] 	(.5,-.5) -- +(-.5,-.5) -- +(.5,-.5)
		(.5,-.5) -- ++(.5,.5) -- ++(-.5,.5) -- +(-.5,-.5)
		(.5,.5) -- +(.5,.5) -- +(-.5,.5);
	
	\draw[white,thick] (-1,-1) -- (1,1);
	\draw[white,thick] (1,-1) -- (-1,1);
	\draw[white,thick] (-1,0) -- (0,1) (0,-1) -- (1,0);
	\draw[white,thick] (-1,0) -- (0,-1) (0,1) -- (1,0);
	
	\foreach \i in {-1,-.5,0,.5,1}{
		\draw[align=left] (-1.02,{\i}) -- (-1.08,{\i});
		\draw[align=left] ({\i},-1.02) -- ({\i},-1.08);
	}
	
	\draw (-1.05,-1)  node[left,align=left]{ $-\pi$};
	\draw (-1.05,-.5)  node[left,align=left]{ $-\frac{\pi}{2}$};
	\draw (-1.05,0)  node[left,align=left]{ $0$};
	\draw (-1.05,.5)  node[left,align=left]{ $\frac{\pi}{2}$};
	\draw (-1.05,1)  node[left,align=left]{ $\pi$};
	
	\draw (-1,-1.05)  node[below,align=center]{ $-\pi$};
	\draw (-.5,-1.05)  node[below,align=center]{ $-\frac{\pi}{2}$};
	\draw (0,-1.05)  node[below,align=center]{ $0$};
	\draw (.5,-1.05)  node[below,align=center]{ $\frac{\pi}{2}$};
	\draw (1,-1.05)  node[below,align=center]{ $\pi$};
	
	\draw (0,1.05) node[above,align=center]{ $\theta_1$};
	\draw (1.05,0) node[right,align=left]{ $\theta_2$};

\def\dummydist{0.007}
	\draw[pattern=north east lines, pattern color=black] (-0.15+\dummydist, 0.814136-\dummydist) rectangle (0.15-\dummydist, 0.185864+\dummydist);
	\draw[pattern=north east lines, pattern color=black] (-0.85+\dummydist, 0.814136-\dummydist) rectangle (-1.01, 0.185864+\dummydist);
	\draw[pattern=north east lines, pattern color=black] (1.01, 0.814136-\dummydist) rectangle (0.85-\dummydist, 0.185864+\dummydist);
	\draw[pattern=north west lines, pattern color=black] (-0.25-\dummydist, 0.21034-\dummydist) rectangle (-0.75+\dummydist,-0.21034+\dummydist);
	\draw[pattern=north west lines, pattern color=black] (-0.25-\dummydist, -1+0.21034-\dummydist) rectangle (-0.75+\dummydist,-1.01);
	\draw[pattern=north west lines, pattern color=black] (-0.75+\dummydist,1-0.21034+\dummydist) rectangle (-0.25-\dummydist,1.01);

	\draw[thick,color=linecol] plot [smooth] coordinates {(0., 1.) (0.05, 0.932782) (0.1, 0.871257) (0.15, 0.814136) (0.2, 0.760654) (0.25, 0.71034) (0.3, 0.662911) (0.35, 0.618203) (0.4, 0.576144) (0.45, 0.536727) (0.5, 0.5)};
	\draw[thick,color=linecol] plot [smooth] coordinates {(0., 0) (0.05, 0.0672175) (0.1, 0.128743) (0.15, 0.185864) (0.2, 0.239346) (0.25, 0.28966) (0.3, 0.337089) (0.35, 0.381797) (0.4, 0.423856) (0.45, 0.463273) (0.5, 0.5)};
	\draw[thick,color=linecol] plot [smooth] coordinates {(-0.5, 0.5) (-0.45, 0.536727) (-0.4, 0.576144) (-0.35, 0.618203) (-0.3, 0.662911) (-0.25, 0.71034) (-0.2, 0.760654) (-0.15, 0.814136) (-0.1, 0.871257) (-0.05, 0.932782) (0., 1.)};
	\draw[thick,color=linecol] plot [smooth] coordinates {(-0.5, 0.5) (-0.45, 0.463273) (-0.4, 0.423856) (-0.35, 0.381797) (-0.3, 0.337089) (-0.25, 0.28966) (-0.2, 0.239346) (-0.15, 0.185864) (-0.1, 0.128743) (-0.05, 0.0672175) (0., 0)};
		
	\draw[white, very thick] (-1.02,-1.01) -- (1.02,-1.01);
	\draw[white, very thick] (-1.02,1.01) -- (1.02,1.01);
	\draw[white, very thick] (-1.01,-1.02) -- (-1.01,1.02);
	\draw[white, very thick] (1.01,1.02) -- (1.01,-1.02);

	\draw[dotted] (0,1) -- (3,1) (0,0) -- (3,-1) (-.5,.5) -- (2,0) (.5,.5) -- (4,0);
	\fill[nncol,right graphic] (0,-1) -- (1,0) -- (0,1) -- (-1,0) --  (-1,0);
	\draw[right graphic,pattern=north east lines, pattern color=black] (-0.3+1.5*\dummydist, 0.628272-1.5*\dummydist) rectangle (0.3-1.5*\dummydist, -0.628272+1.5*\dummydist);
	\draw[right graphic,thick,color=linecol] plot [smooth] coordinates {(0., 1.) (0.05, 0.931136) (0.1, 0.865565) (0.15, 0.80281) (0.2, 0.742514) (0.25, 0.684405) (0.3, 0.628272) (0.35, 0.573949) (0.4, 0.521307) (0.45, 0.470244) (0.5, 0.420681) (0.55, 0.372556) (0.6, 0.325822) (0.65, 0.280447) (0.7, 0.236406) (0.75, 0.193688) (0.8, 0.152288) (0.85, 0.112207) (0.9, 0.0734541) (0.95, 0.0360452) (1., 0)};
	\draw[right graphic,thick,color=linecol] plot [smooth] coordinates {(0., -1.) (0.05, -0.931136) (0.1, -0.865565) (0.15, -0.80281) (0.2, -0.742514) (0.25, -0.684405) (0.3, -0.628272) (0.35, -0.573949) (0.4, -0.521307) (0.45, -0.470244) (0.5, -0.420681) (0.55, -0.372556) (0.6, -0.325822) (0.65, -0.280447) (0.7, -0.236406) (0.75, -0.193688) (0.8, -0.152288) (0.85, -0.112207) (0.9, -0.0734541) (0.95, -0.0360452) (1., 0)};
	\draw[right graphic,thick,color=linecol] plot [smooth] coordinates {(0., 1.) (-0.05, 0.931136) (-0.1, 0.865565) (-0.15, 0.80281) (-0.2, 0.742514) (-0.25, 0.684405) (-0.3, 0.628272) (-0.35, 0.573949) (-0.4, 0.521307) (-0.45, 0.470244) (-0.5, 0.420681) (-0.55, 0.372556) (-0.6, 0.325822) (-0.65, 0.280447) (-0.7, 0.236406) (-0.75, 0.193688) (-0.8, 0.152288) (-0.85, 0.112207) (-0.9, 0.0734541) (-0.95, 0.0360452) (-1., 0)};
	\draw[right graphic,thick,color=linecol] plot [smooth] coordinates {(0., -1.) (-0.05, -0.931136) (-0.1, -0.865565) (-0.15, -0.80281) (-0.2, -0.742514) (-0.25, -0.684405) (-0.3, -0.628272) (-0.35, -0.573949) (-0.4, -0.521307) (-0.45, -0.470244) (-0.5, -0.420681) (-0.55, -0.372556) (-0.6, -0.325822) (-0.65, -0.280447) (-0.7, -0.236406) (-0.75, -0.193688) (-0.8, -0.152288) (-0.85, -0.112207) (-0.9, -0.0734541) (-0.95, -0.0360452) (-1., 0)};

	\end{tikzpicture}
	\caption{\label{fig:ssrectangles}\emph{Left:} Parameter regions \eqref{varytable} in which the angles $\theta_i(x)$ are allowed to vary arbitrarily, while keeping the gap open. (Case 1 and case 4 distinguished by hatching). \\ \emph{Right:} Magnified part to make visible the boundary line (green) on which the corners of all hatched rectangles lie according to \eqref{epsrect}.}
\end{figure}

The idea behind these choices of intervals is that for $\eps_1=\eps_2=0$ either every $\theta_1(x)$ (cases 1,2) or every $\theta_2(x)$ (cases 3,4) is set at a decoupling value $\pm\pi/2$ (see \Qref{sec:splitss}. Therefore, the walk decomposes into a direct sum of $2\times2$-matrices. For the cases 1,2 these blocks do not follow the cell decomposition (compare \Qref{fig:Wgenerate}), for cases 2,3 they are just single cells. The other angle is then left a choice of two values, for which these $2\times2$-blocks differ by a sign, but either way have spectrum $\{i,-i\}$. Hence the whole walk $W_0$ has this spectrum, and hence $W_0^2=-\idty$.

So for every walk $W$ satisfying \eqref{varytable}, there is another one $W_0$ with $W_0^2=-\idty$ whose angles $\theta_i(x)$ differ by at most $\eps_i$. Thus in the decomposition \eqref{eq:wss} we have
\begin{equation}\label{ssAest}
  \norm{A-A_0}=\sup_x\norm{R(\theta_1(x))-R(\theta_{1,0}(x))}\leq \norm{R(\eps_1)-\idty}=\abs{e^{i\eps_1}-1}=2\sin\frac{\eps_1}2.
\end{equation}
Of course, the analogous estimate holds for $B$, and using \eqref{epsrect} we conclude that
\begin{equation}\label{ssWest}
  \norm{W-W_0}\leq \norm{A-A_0}+\norm{B^2-B_0^2}\leq 2\sin\frac{\eps_1}2+2\sin\frac{\eps_2}2<\sqrt2.
\end{equation}
But then a standard resolvent estimate ensures that the gap is open: Proceeding as in \eqref{eq:R1-R0}, we get for the resolvents $R=(\idty-W)\inv$ and $R_0=(\idty-W_0)\inv$
the relation $R(\idty-(W-W_0)R_0)=R_0$, which can be solved for $R$ by a geometric series as soon as $\norm{(W-W_0)R_0}<1$. Since $\norm{R_0}=\abs{(1-i)\inv}=1/\sqrt2$, this is guaranteed by \eqref{ssWest}. Hence $R$ is a bounded operator, meaning that $1$ is not in the spectrum. Of course, the same applies to $-1$, so both gaps of $W$ are open.

We remark that it is crucial for this argument that $A$ can be estimated as a direct sum, just leading to a maximum of norm differences. This feature fails if we do not allow the two angles to vary completely independently
but just constrain all pairs $(\theta_1(x),\theta_2(x))$, and perhaps $(\theta_1(x+1),\theta_2(x))$ in one particular phase region. Numerical evidence suggests that if these pairs belong to the same value of the phase diagram (not mixing cases 3 and 4), and also keep a finite distance from the boundaries the resulting walk is also gapped. But we have at the moment no method to show this.

Now assume a walk $W$, with parameters chosen according to one of the four cases \eqref{varytable}. Then $W$ can continuously be deformed into an appropriate $W_0$, with $W_0^2=-\idty$, without closing the gap. For such walks however, it is straightforward to compute the respective index $\sixR(W)=-1/+1/0/0$ in case $1/2/3/4$.

\section{Indices sensitive to non-gentle perturbations}\label{sec:homotopy-indices}

In this section we return to non-translation invariant walks, but will stick to the strictly unitary setting as described in \Qref{sec:firstsetting}. In this scenario the symmetry indices $\six_\pm$ provide an invariant to classify the gentleness of perturbations. We will first describe the basic results, and then give an index criterion to decide the gentleness of any compact perturbation.

\subsection{The relative index of a compact perturbation}\label{sec:locpert}

It is often convenient to write perturbations of unitary operators in multiplicative form, so that $W'=VW$.
In this case we will also refer to the unitary $V$ as the ``perturbation''.
For the perturbation properties listed in \Qref{def:pertsorts} it is easy to find the corresponding properties of $V$. For example $W'$ is a local perturbation iff $V-\idty$ is only non-zero on finitely many cells, and it is a compact perturbation iff $V-\idty$ is a compact operator. Because some of the admissibility conditions with respect to the symmetry involve the adjoint walk it is {\it not} true that the product of admissible operators is admissible. Consequently, the condition for $W'$ to be admissible for the same symmetry representation as $W$ is not directly the admissibility of $V$, but a modified form described in the following Lemma.

\begin{lem}\label{lem:twiddlesym} Let $W$ be an admissible walk, and introduce the symmetry operators
	\begin{equation}\label{twiddlesym}
	\tph=\ph, \quad \trv=W\rv, \quad\mbox{and}\quad\tch=W\ch,
	\end{equation}
	whenever these operators are part of the symmetry type. Then these operators are a symmetry representation of the same type. Moreover, for any unitary $V$ the operator $VW$ satisfies the symmetry admissibility conditions for the appropriate subset of $(\ph,\rv,\ch)$ if and only if $V$ satisfies the commutation relations for $\trh=(\tph,\trv,\tch)$.
	The subspace
	\begin{equation}\label{HV}
	\HH_V:=(V-\idty)\HH=(VW-W)\HH=(V^*-\idty)\HH
	\end{equation}
	is invariant under $\tph,\trv,\tch,V$ and $V^*$.
	When $V-\idty$ is compact, $V$ has an essential gap at $-1$ and, in particular, the $-1$-eigenspace $\HH_V^-\subset\HH_V$ is finite dimensional.
\end{lem}

\begin{defi}\label{def:relInd}
	In the setting of the Lemma, the symmetry index $\tsix_-(V)$, i.e., the symmetry index of $\trh$ restricted to $\HH_V^-$,
	is called the {\dff relative index} of a perturbation $VW$ of $W$, and will be denoted by $\sixrel VW:W$.
\end{defi}

\begin{proof}[Proof of \Qref{lem:twiddlesym}] We have to verify that the representation $\trh$ satisfies the sign convention of \Qref{lem:phases}, and has the same squares as $(\ph,\rv,\ch)$, and that the admissibility conditions of $V$ and $VW$ are equivalent. All of this is straightforward algebra.
	
	The equality of the three spaces follows because $V-\idty=(VW-W)W^*=(V^*-\idty)(-V)$. That $\HH_V$ is invariant under $V,V^*$ is clear from this expression. Invariance under $\tch$ follows because $\tch(V-\idty)\HH=(V^*-\idty)\tch\HH$, and similarly for the other symmetries.
	
	The essential gap at $-1$ and the finite dimensionality of $\HH_V^-$ follow directly from the compactness of $V-\idty$. The symmetry index $\tsix_-(V)$ is then well defined, since $\HH_V^-$ is finite-dimensional.
\end{proof}

By homotopy invariance of $\six_-$ (\Qref{pro:homoto}), the relative index vanishes for gentle perturbations, so it classifies compact perturbations modulo gentle compact ones. As the following Lemma shows this is actually sharp, i.e., compact perturbations with vanishing relative index must be gentle.

\begin{lem}\label{lem:contract}
	Let $W'=VW$ be a compact perturbation of $W$. Then the following are equivalent:
	\begin{itemize}
		\item[(1)] $\sixrel VW:W=0$
		\item[(2)] There is a continuous path $t\mapsto V_t$ of $\trh\,$-admissible unitary operators on $\HH_V$ connecting $V$ with $\idty$,
		so that $W_t=V_tW$ will be a homotopy connecting $VW$ and $W$.
	\end{itemize}
\end{lem}

\begin{proof}
	(2) implies (1), because $\tsix_-$ is a homotopy invariant by \Qref{pro:homoto}. Conversely, suppose $\tsix_-(V)=0$. On the complement of $\HH_V^-$ in $\HH_V$, the operator $V$ is gapped so we can deform it to the identity by contracting the eigenvalues along the unit circle to $1$. On $\HH_V^-$ the index condition means that there is a gapped unitary, whose eigenvalues we can either deform to $-1$, so that it coincides with $V$ on that subspace, or to $+1$, as required.
\end{proof}

The relative index theory gives a concise description of the phenomenon of non-gentle perturbations. Indeed it has been argued \cite{Thiang} that the indices of a topological phase classification should be such relative indices. However, the weakness of this point of view is that it is confined to just the compact perturbations of a fixed walk. Each such ``island'' is charted well, but no connection between the different islands is made. In particular, this approach leaves open the question whether a compact perturbation with $\sixrel W':W\neq0$ can after all be contracted to $W$ by a continuous path involving also global (non-compact) deformations. This issue will be resolved in the following section in a surprisingly simple way: the relative index turns out to be the difference of some absolute indices that we have already introduced. Since absolute indices are invariant with respect to arbitrary gentle perturbations (\Qref{pro:homoto}), a compact perturbation with non-vanishing relative index cannot be contracted.

\subsection{Relative index from absolute}

\begin{thm} \label{thm:locpert}
	Let $W'=VW$ be a compact perturbation of $W$, and let $\sixrel W':W$ denote the relative symmetry index defined in \Qref{def:relInd}. Then
	\begin{equation}\label{locpert}
	\sixrel W':W=\six_-(W')-\six_-(W)=-\bigl(\six_+(W')-\six_+(W)\bigr).
	\end{equation}
\end{thm}

A good way to think about this identity is as an analogue of a determinant product. In fact, for a finite dimensional overall space with symmetry type \symD, it is just that: according to \eqref{siDet}, we can directly translate the relation $\det(VW)=\det(V)\det(W)$ to the index setting to get the result. For infinite dimensional spaces, however, $\six_-(W)$ is not directly a determinant, understood as the infinite product of all eigenvalues, unless one makes the convention that every non-real eigenvalue is immediately combined with its conjugate, and thus omitted from the product, and that this rule is also extended to the continuous spectrum of $W$. That is certainly a reasonable convention and leads to an equivalent definition of $\six_-(W)$. But then it is a non-trivial question whether the product formula still holds. The Theorem answers this in the affirmative.
Similar remarks apply to the symmetry types \symAIII, \symBDI, and \symCII, for which in finite dimension we get from \eqref{sitrch}:
\begin{equation}
\tsix_-(V)=\tfrac12\tr(\widetilde{\gamma}(\idty-V))
=\tfrac12\tr(\gamma(W-W'))=\six_-(W')-\six_-(W).
\end{equation}

\begin{proof}
The proof will be given for the first identity in \eqref{locpert}. The second one follows from $\six_+(W)=\six(W)-\six_-(W)$ and $\six(W')=\six(W)$, since $\six(W)$ is invariant under compact perturbations.
	
	We first use the homotopy invariance of $\six_-$ to simplify the situation. Indeed, if we change $V$ in a norm continuous way (keeping the required symmetry) all terms in \eqref{locpert} will remain constant. Since $V-\idty$ is compact it has discrete spectrum on $\HH_V$. We deform $V$ by continuously moving every complex conjugate eigenvalue pair $\lambda,\overline{\lambda}$ to $1$ keeping the eigenvectors. This ensures that the symmetries are respected along the path. The resulting operator $V'$ has additional eigenvalues $1$, and $\HH_{V'}$ is smaller than $\HH_V$, namely only the $-1$-eigenspace of $V$. Therefore, we can take $V(=V')=(\idty-2P)$, with $P$ a finite dimensional projection onto a subspace which is invariant under the symmetry representation $(\tph,\trv,\tch)$ from \Qref{lem:twiddlesym}.
	
	The left hand side of \eqref{locpert} is then given by the symmetry index of the finite dimensional representation $\trh_P$ on $P\HH$. The two summands on the right hand side are given by the symmetry indices of the representations $\rho_N$ and $\rho_{N'}$, where $N$ and $N'$ are the $-1$-eigenprojections of $W$ and $W'$, respectively.
	To prove the index formula, we need to compute $\six(\rho_{N'})$ in terms of $\six(\rho_{N})$ and $\six(\trh_P)$. By definition, $N'\HH$ is the solution space of $W'\psi=-\psi$, which we rewrite as
	\begin{equation}\label{eq:eigenval}
	(W+\idty)\psi=2PW\psi.
	\end{equation}
	We first split the spaces $N\HH, N'\HH$ and $P\HH$ into symmetry invariant subspaces to identify the correct summands on both sides of \eqref{locpert} (see \Qref{fig:relativeindex}).
	
	$N$ and $N'$ can be written as $N=N_0\oplus N_1$, with
	\begin{eqnarray}
	N_0\HH&=&\{\phi\in N\HH\mid P\phi=0\}=N\HH\cap P^\perp\HH,\\
	N_1\HH&=&(N-N_0)\HH=NP\HH,\nonumber
	\end{eqnarray}
	and likewise for $N'$, where we use the notation $P^\perp=\idty-P$ for projections. The subspace $N_0\HH$, and consequently $N_1\HH$, is $\rho$-invariant, since $P\gamma\phi=PWW^*\gamma\phi=P\tch W\phi=-\tch P\phi=0$ and $P\ph\phi=\ph P\phi=0$, 	for $\phi\in N_0\HH$. Here and in the remainder of this proof we check only symmetry conditions for $\ph$ and $\ch$, since the computation for $\rv$ is analogous, and follows anyhow if all three symmetries are present.

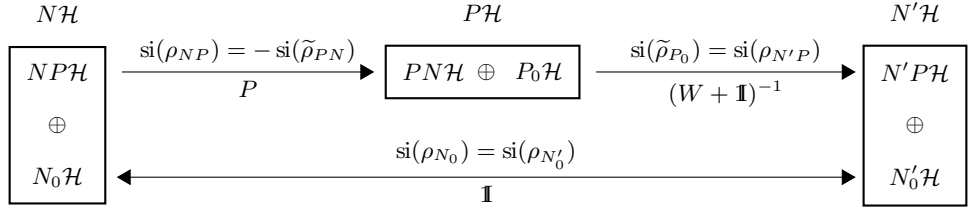
\begin{figure}[ht]
		\begin{tikzpicture}[>=triangle 60]
	\def\hdist{5}
	\def\vdist{0.7}
	\def\pos{above}
	
	\node (n0)	at (0,0) {$N_0\HH$};
	\node		at (0,\vdist) {$\oplus$};
	\node (np)	at (0,2*\vdist) {$NP\HH$};
	\node (pn)	at (\hdist,2*\vdist) {$PN\HH$};
	\node		at (\hdist+\vdist,2*\vdist) {$\oplus$};
	\node (p0)	at (\hdist+2*\vdist,2*\vdist) {$P_0\HH$};
	\node (ns0)	at (2*\hdist+2*\vdist,0) {$N'_0\HH$};
	\node		at (2*\hdist+2*\vdist,\vdist) {$\oplus$};
	\node (nsp)	at (2*\hdist+2*\vdist,2*\vdist) {$N'P\HH$};
	
	\def\adist{1mm}
	\node[right=\adist of n0] (n0a) {};
	\node[right=\adist of np] (npa) {};
	\node[right=\adist of p0] (p0a) {};
	\node[left=\adist of ns0] (ns0a) {};
	\node (phantomnsp) at (2*\hdist+2*\vdist,2*\vdist) {$\phantom{N'_0\HH}$};
	\node (phantompn) at (\hdist,2*\vdist) {$\phantom{N'_0\HH}$};
	\node[left=\adist of phantomnsp] (nspa) {};
	\node[left=\adist of phantompn] (pna) {};

	\node[draw,fit=(n0) (np),thick] (nbox) {}; 
	\node[draw,fit=(p0) (pn),thick] (pbox) {}; 
	\node[draw,fit=(ns0) (nsp),thick] (nsbox) {}; 
	
	\node[above=2mm of nbox] {$N\HH$};
	\node[above=2mm of pbox] {$P\HH$};
	\node[above=2mm of nsbox] {$N'\HH$};

	\draw[<->] (n0a)--(ns0a) node [midway,\pos]
		{\small $\six(\rho_{N_0})=\six(\rho_{N'_0})$} node [midway,below] {$\idty$};
	\draw[->] (npa) -- (pna) node [midway,\pos]
		{$\six(\rho_{NP})=-\six(\trh_{PN})$} node [midway,below] {$P$};
	\draw[->] (p0a) -- (nspa) node [midway,\pos]
		{\small $\six(\trh_{P_0})=\six(\rho_{N'P})$} node [midway,below] {$(W+\idty)^{-1}$};
\end{tikzpicture}
		\caption{\label{fig:relativeindex}
			Mappings between the subspaces that are involved in a finite rank perturbation $W'=(\idty-2P)W$ (see proof of \Qref{thm:locpert}). The respective map between the spaces is denoted below the arrows and the bidirectional arrow stands for equal spaces. The resulting consequence for the index of the symmetry representation given on each of the subspaces  is denoted above the arrows.}
	\end{figure}	
Analogously to $N\HH$, $P\HH$ is split into
	\begin{equation}
	P\HH=P_0\HH\oplus (P-P_0)\HH:=\bigl(P\HH\cap N^\perp\HH\bigr)\oplus PN\HH.
	\end{equation}
	Each of these spaces is invariant for $\trh$, since $P$ and $N$ commute with the symmetries of $\trh$ ($N$ commutes with $\rho$ and with $W$, and therefore with $\trh$). The symmetry invariance of all these subspaces allows us to decompose the symmetry indices as $\six(\rho_N)=\six(\rho_{N_0})+\six(\rho_{N_1})$, the analogous relation for  $\rho_{N'}$, and $\six(\trh_P)=\six(\trh_{PN})+\six(\trh_{P_0})$. We will now show that some of these summands are equal, due to the existence of appropriate intertwining relations. The relevant spaces and maps are shown in \Qref{fig:relativeindex}.
	
	Consider the eigenvalue equation \eqref{eq:eigenval} for $W'$. For one simple kind of solutions this equation can be satisfied with both sides equal to $0$. That is, the solution $\psi\in N'\HH$ satisfies two further conditions: on one hand, $\psi\in N\HH$, because the left hand side vanishes. On the other, $P\psi=-PW\psi=0$, so $\psi\in P^\perp\HH$, that is $\psi\in N_0\HH$ and $\psi\in N_0'\HH$. Therefore $N_0=N'_0$.
	
	The next step in the analysis of the eigenvalue equation \eqref{eq:eigenval} is the analysis of $N'_1\HH$, i.e., the solutions orthogonal to the ones already found in $N_0'\HH$. Let $\phi\neq0$ be the vector equal to both sides in \eqref{eq:eigenval}. Then the right hand side implies $\phi\in P\HH$, and the left $\phi\in N^\perp\HH$, so
	$\phi\in P_0\HH$. Since in the case considered here $\psi\in N^\perp\HH$, we can reconstruct $\psi$ from $\phi$. Restricted to $N^\perp\HH$, the inverse of $(W+\idty)$ exists and we get $\psi=(W+\idty)^{-1}\phi$. This restricted inverse is also known as the pseudo-inverse of $(W+\idty)$, and is bounded due to the essential gap condition. For this $\psi$ to satisfy \eqref{eq:eigenval}, $\phi$ must fulfill the consistency condition
	\begin{equation}
	0= P(\idty-W)(W+\idty)^{-1}P\phi =\vcentcolon iH\phi.
	\end{equation}
	As one easily checks, $H$ is a $\trh$-admissible Hamiltonian on the finite dimensional space $P_0\HH$.
	Therefore, the symmetry representation on $P_0\HH\ominus\ker(H)$ is balanced and we get $\six(\trh_{P_0})=\six(\trh\vert_{\ker(H)})=\tsix(H)$.
	To identify the symmetry representations $\trh\vert_{\ker(H)}$ and $\rho_{N'_1}$ with each other, first note that $(W+\idty)^{-1}$ bijectively maps $\ker(H)$ to $N'_1\HH$. Moreover, it intertwines the symmetry representations $\rho$ and $\trh$:
	let $\phi\in\ker(H)$,  then
	\begin{align}
	(W+\idty)^{-1}\tch\phi &= (W+\idty)^{-1}W\ch\phi = \ch W^*(W^*+\idty)^{-1}\phi = \ch (W+\idty)^{-1}\phi \\
	(W+\idty)^{-1}\tph\phi &= (W+\idty)^{-1}\ph\phi= \ph (W+\idty)^{-1}\phi .
	\end{align}
	Thus the symmetry representation $\trh$ on $\ker(H)$ is unitarily equivalent (a possible unitary being the polar isometry of $(W+\idty)^{-1}$) to $\rho$ on $N'_1\HH$ and we get $\six(\rho_{N'_1})=\six(\trh\vert_{\ker(H)})=\six(\trh_{P_0})$.
	
	Finally, we have to connect the spaces $N_1\HH=NP\HH$ and $PN\HH$. We consider $P$ as a map
	$P:N\HH\to PN\HH$. Its kernel is precisely $N_0\HH$, so restricted to $N_1\HH$ it is a bijection. Moreover, $P$ intertwines the symmetries: for  $\phi\in N_1\HH$, using $W\phi=-\phi$,
	\begin{align}\label{Pchisign}
	P\ch\phi &= PWW^*\ch\phi=PW\ch W\phi=-P\tch\phi=-\tch P\phi \\
	P\ph\phi &= \ph P\phi.
	\end{align}
	The same relations hold for the polar isometry of $P$, which is a unitary from  $NP\HH$ onto $PN\HH$.
	Hence the respective symmetry representations are the same up to an additional minus sign for $\ch$ and $\rv$. However, for any finite dimensional representation $\six(\ph,-\rv,-\ch)=-\six(\ph,\rv,\ch)$ (see the proof of \Qref{prop:six}). Hence $\six(\trh_{PN})=-\six(\rho_{N_1})$.	
	To summarize,
	\begin{align*}
	\six(\trh_{P})&=\six(\trh_{P_0})+\six(\trh_{PN})
    =\six(\rho_{N'P})-\six(\rho_{NP})\\
	&=\six(\rho_{N'P})+\six(\rho_{N'_0})-\six(\rho_{N_0})-\six(\rho_{NP})
	=\six_-(W')-\six_-(W).
	\end{align*}
\end{proof}

\noindent A direct consequence of \Qref{thm:locpert} is the following chain rule for the relative index:

\begin{cor}\label{cor:chainrule}
	Let $W'$ and $W''$ be compact perturbations of $W$. Then $W''$ is a compact perturbation of $W'$ and
	\begin{equation}\label{chainrule}
		\sixrel W'':W=\sixrel{W''}:{W'}+\sixrel W':W.
	\end{equation}
\end{cor}

Note that there is a subtlety to keep in mind: the multiplicative perturbations  $V_1=W'(W)^*$ and $V_2=W''(W')^*$, corresponding to the two relative indices on the right hand side of \eqref{chainrule}, are admissible for different symmetry representations of the same type: $V_1$ is admissible for $\trh=(\ph,W\rv,W\ch)$, whereas $V_2$ is admissible for $\trh'=(\eta,W'\rv,W'\ch)$.

However, when we perturb a walk in two separate regions, we can ignore this fact. Consider a strictly local walk $W$, and two local perturbations $V_1,V_2$ in sufficiently distant regions of the system (as compared to the interaction length). Then we get for the overall perturbation $V_1V_2$: $\HH_{V_1V_2}=\HH_{V_1}\oplus\HH_{V_2}$. Therefore, since $\HH_{V_1}\perp\HH_{V_2}$, the representation $\trh'$, restricted to $\HH_{V_1}$ coincides with $\trh$, restricted to the same subspace and we get $\tsix_{-}(V_1)=\tsix'_-(V_1)$. Hence
\begin{equation}\label{locindexAdd}
\sixrel V_2V_1W:W=\sixrel V_2W:W+\sixrel V_1W:W.
\end{equation}
Note that the sum may be zero, even if the summands are not. In that case we know that the overall perturbation can be contracted locally (on a finite region containing $\HH_{V_1}$ and $\HH_{V_2}$), even though the individual perturbations cannot (see \Qref{fig:tubes}).

\begin{figure}
	\input{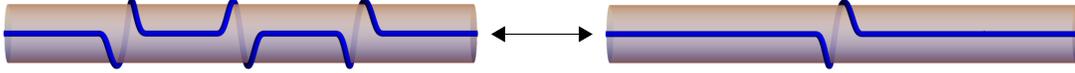}
	\caption{\label{fig:tubes}
		An example of locally contracting three perturbations to one. The tube visualizes the infinite chain, while the winding of the blue ribbon characterizes the relative index of the perturbation at the corresponding position.
	}
\end{figure}

\section{Decoupling Theory}\label{sec:decoup}

\subsection{Decoupling-Index of a one-dimensional walk}\label{sec:old-index}
The decoupling construction we use is built on the construction used in \cite{OldIndex}, where decoupling is studied in a setting which is on one hand more general (no symmetry, no gap) and, on the other, more special (strict locality and strict unitarity).
Strict locality was needed in \cite{OldIndex} because walks were just used as an analogy to understand the decoupling of cellular automata (aka.\ interacting systems). In the present context this constraint is unnatural.  Therefore we begin by recapitulating and at the same time generalizing the index and decoupling construction as the basis of the construction under symmetries.

Let $W$ be an essentially unitary operator. Then, by definition, it is a Fredholm operator, whose index we denote by $\indF(W)=\dim\ker W-\dim\ker W^*$. Then, if $W$ is also essentially local, i.e.,  $[P,W]$ is compact for the half-space projection $P$, and hence $W'=PWP+(\idty-P)W(\idty-P)$ is a compact, essentially unitary perturbation of $W$. In particular, $PWP$ is essentially unitary on the half-space, and we define
\begin{equation}\label{oldIxF}
  \ind(W):=\indF(PWP).
\end{equation}
We will see presently that this extends the definition given in \cite{OldIndex} for the case of unitary, strictly local walks. The above definition is in the spirit of \Qref{sec:indices}, and allows us to quickly establish the basic properties, by referring to Fredholm theory. Clearly, $\ind(W)$ is stable under homotopy and compact perturbations, independent of $a$ when $P=\Pg a$, and satisfies the product formula $\ind(W_1W_2)=\ind(W_1)+\ind(W_2)$. Moreover, for the shift $S$ on $\ell^2(\Ir)$, $PSP$, considered as an operator on $\ell^2(\Nl)\cong P\ell^2(\Ir)$ is the unilateral shift, one of the standard examples of Fredholm theory, and we get $\ind(S)=1$. More generally, for translation invariant $W$, $\ind(W)$ is the total winding number of the quasi-energy spectrum around the Brillouin circle \cite{OldIndex}. Note also that in analogy to the symmetry index, we really get a right index
$\ind(W)=\mbox{``}\indFR(W)\mbox{''}=\indF(PWP)$ and a left index $\mbox{``}\indFL(W)\mbox{''}=\indF((\idty-P)W(\idty-P))$ with $\indFL(W)+\indFR(W)=\indF(W)$. But since we are almost exclusively interested in the unitary case with $\indF(W)=0$, and hence $\indFL(W)=-\indFR(W)$, we will not use this notation.

For a decoupled unitary walk, $PWP$ is unitary, so such a walk has a vanishing index. By stability with respect to compact perturbations we find that a necessary condition for the existence of a compact decoupling is a vanishing index. That this condition is also sufficient will be shown by explicit construction in the next subsection. When we study decouplings in the setting of the present paper this raises the question whether we need $\ind(W)=0$ as an additional assumption. The answer is no, from two independent lines of reasoning: firstly, for most symmetry types $\ind(W)=0$ actually follows just from symmetry (see the remark following \Qref{prop:swapH01}). A short version of this argument is that the $\ind(W)$ changes sign under $W\mapsto W^*$, so it has to vanish for any symmetry type containing $\ch$ or $\rv$. But our theory is also to apply to the types \symA, \symC, and \symD. In this case (and in all others) we can invoke a second general feature of our setting, the essential gap assumption. This is shown in the following proposition, which is a variant of the Brown-Douglas-Fillmore Theorem \cite[Thm.~3.1]{BDF}. We will sketch a proof for completeness.

\begin{prop}\label{prop:bdf}
Let $W$ be an essentially local unitary operator with $\ind(W)\neq0$. Then the essential spectrum of $W$ is the full unit circle.
\end{prop}

\begin{proof}
We will show that there is a compact perturbation $W'$, which is of the form $S\oplus U$, where $U$ is a decoupled unitary and $S$ is a bilateral shift with $\ind(S)=\ind(W)$. Note that the direct sum $S\oplus U$ is not along the given cell structure. Then since the spectrum of $S$ is absolutely continuous and equal to the full unit circle, the essential spectrum of $W'$ is the full unit circle, a property which transfers to $W$ by the stability of the essential spectrum under compact perturbations.

For the construction we may assume $n=\ind(W)>0$, since otherwise we can consider $W^*$. Since we assumed $W$ to be unitary we have $\indF(W)=0=\ind_F(PWP)+\indF((\idty-P)W(\idty-P))$, so it suffices to do the construction on each side: for the almost unitary $PWP$ with Fredholm index $n$ we need to find a compact perturbation $W'=S_n\oplus U$, where $S_n$ is isomorphic to the unilateral shift of multiplicity $n$. With the analogous construction on $(\idty-P)\HH$ we can then join the two unilateral shifts to a bilateral one. The basic idea is to reduce the pair of ``defect indices'' $(n_+,n_-)=(\dim\ker W,\dim\ker W^*)$ by the same amount $n_-$ so that they become $(n_+-n_-,0)=(\indF(PWP),0)$. This is done by setting $W_1=W+W_\delta$ so that $W_\delta$ vanishes on $(\ker W)^\perp$, and  $W_\delta:\ker W\subset P\HH\to\ker W^*$ is any linear map of maximal rank (i.e., $n_-$). Clearly, $W_1$ is a compact, more precisely a rank $n_-$ perturbation of $W$. By construction, $W_1W_1^*$ has trivial kernel, and because $W_1W_1^*-\idty$ is compact, $0$ lies in the region where $W_1W_1^*$ has purely discrete spectrum.  Therefore, it is boundedly invertible. We can then form the polar co-isometry $W'=(W_1W_1^*)^{-1/2}W_1$, which is a compact perturbation of $W_1$, hence of $W$.

What is left to show is that a co-isometry $W'$ of index $n$, i.e., an operator with $W'(W')^*=\idty$ and $\dim\ker W'=n$ can be decomposed into a unitary part and a shift of multiplicity $n$ (Wold- von Neumann decomposition \cite[Sect.~{X}]{WoldvN},\cite{nagy}). The idea is to inductively build a ``cell structure'' with $\HH_0=\ker W'$ and $\HH_j=(W'^*)^j\HH_0$, and to show that $W'$ is unitary on the orthogonal complement of $\bigoplus_{j\geq0}\HH_j$.
\end{proof}

With the existence of decouplings established we now move to the explicit constructions.

\subsection{Canonical decoupling and the algebra generated by two projections}\label{sec:2proj}

In this section we will explicitly construct a gentle decoupling for any essentially local unitary with $\ind(W)=0$. By definition, this is a unitary $V$ with $V-\idty$ compact so that $W'=VW$ is decoupled, i.e., commutes with $P=\Pg0$. The idea is to do this in a ``canonical'' way, which is to say that any symmetries of $W$ and $P$ are automatically transferred to the decoupling. Since we want to construct a gentle decoupling, $-1$ should not be an eigenvalue of the constructed $V$, because then $\tsix_-(V)=0$, and by \Qref{lem:contract} it can be contracted to the identity.

The construction is made canonical via the use of a canonical object, the C*-algebra generated by two projections \cite{Sinclair}, of which index aspects have been studied in \cite{AvronSeiler}.
In our case the projections will be the half-space projection $P$ and its one-step translate $Q=WPW^*$, where $P-Q=[P,W]W^*$ is a compact operator. The decoupling condition $[P,VW]=0$ then translates to the operator $V$ as the intertwining condition
\begin{equation}\label{VQPV}
VQ=PV.
\end{equation}
A ``canonical'' $V$ would be just a noncommutative function of $P$ and $Q$. We will be only partially successful in this, but will still get a major simplification of the symmetry considerations in the next subsection.

The ``C*-algebra generated by two projections'', call it $\AA$, is the set of ``noncommutative functions'' of abstract projections $P$ and $Q$, so that ``plugging in'' a particular choice of projections is the evaluation of the function. In other words, every choice of two concrete projections in some Hilbert space gives a representation of $\AA$. There is a Weierstra\ss\ Theorem saying that $\AA$ is generated by the noncommutative polynomials of $P$ and $Q$ \cite[Lem.\;A.2]{RW89}.
The reason why this $\AA$ is a very useful object (in contrast to, e.g., the analogously defined algebra generated by three projections) is that each monomial is just an alternating string of $P$s and $Q$s determined by the first and last element, and the power of $(PQ)$ in between. These four classes of monomials multiply like $2\times2$-matrices. More formally, the irreps of $\AA$ are one-dimensional in the simplest case. These are relevant for all matters concerning commuting projections. They are just parameterized by the respective eigenvalues of $P$ and $Q$. All other irreps are two-dimensional, with $P,Q$ given by two one-dimensional projections. These are parameterized by the angle $\alpha$ between the two subspaces.

We hence introduce the following subspaces \cite{Halmos,Andruchow}:
\begin{equation}\label{Hpm}
\begin{aligned}
\HH_{11}&=\{\phi\in\HH|\, P\phi=Q\phi=\phi\} \\
\HH_{00}&=\{\phi\in\HH|\, P\phi=Q\phi=0\} \\
\HH_{10}&=\{\phi\in\HH|\, P\phi=\phi,\ Q\phi=0\} \\
\HH_{01}&=\{\phi\in\HH|\, P\phi=0,\ Q\phi=\phi\} \\
\HH_\perp&= \bigcap_{\eps_1,\eps_2}\HH_{\eps_1,\eps_2}^\perp.
\end{aligned}
\end{equation}

Let us assume now that $(P,Q)$ arises from a walk as above, i.e., $Q=WPW^*$. Then the spaces $\HH_{11}$ and $\HH_{00}$ are associated to the regions ``far to the right'' and ``far to the left'', respectively. When $W$ is even strictly local, they are clearly infinite dimensional, but in general they might be empty. Every vector in $\HH_{10}$ or $\HH_{01}$ is an eigenvector for $P-Q$ with absolute value $1$. Since this operator is compact, both these spaces are finite dimensional. Moreover,
\begin{equation}\label{H10}
\HH_{10}=\{\phi\in P\HH|\, PW^*P\phi=0\}=\ker(PW^*P)
\end{equation}
when $PW^*P$ is considered as an operator on $P\HH$. Similarly, we get $W^*\HH_{01}=\ker(PWP)$. Hence
\begin{equation}\label{oldind1}
\ind (W)=\dim\HH_{01}-\dim\HH_{10}.
\end{equation}

It will be useful to collect some straightforward algebraic facts. We follow \cite{AvronSeiler} introducing the operators
\begin{eqnarray}\label{decoupAB}
A=P-Q &\quad\mbox{and}\quad& B=\idty-P-Q \\
AB+BA=0 && A^2+B^2=\idty.\label{decoupAB2}
\end{eqnarray}
These are important because $[A^2,P]=[A^2,Q]=0$, so $A^2$ and $B^2$ lie in the center of $\AA$. In an irrep with $P=\kettbra\phi$ and $Q=\kettbra\psi$ with an angle $\alpha$ between $\phi$ and $\psi$  we have $A^2=(1-\abs{\braket\phi\psi}^2)\idty$, and the eigenvalues of $A$ are $\pm\sin\alpha$. More generally speaking, by \eqref{decoupAB2}, the eigenvalues of $A$ on $\HH_\perp$ come in $\pm$ pairs so $A$ formally has trace $0$ on this subspace. Hence, taking into account also the kernel $\ker A=\HH_{00}\oplus\HH_{11}$ and the $\pm1$ eigenspaces $\HH_{01}$ and $\HH_{01}$ of $A$ with \eqref{oldind1}
\begin{equation}\label{oldind2}
\ind (W)=-\tr A=\tr (WPW^*-P),
\end{equation}
whenever $A$ is trace class. This is the index definition used in \cite{OldIndex}, where $A=P-Q$ was even finite rank. The definition \eqref{oldind1} is hence more broadly applicable than \eqref{oldind2}, and is used in \cite{AvronSeiler} even for $A$ with some continuous spectrum.

In order to satisfy \eqref{VQPV}, we look for operators $X\in\AA$ satisfying the intertwining relation $XQ=PX$. Once we have such an operator we can produce further ones by multiplying with central elements. Moreover, $X^*X$ will be central, so we can form the polar isometry of any solution $X$ to get an intertwining isometry. The first choice that comes to mind \cite{AvronSeiler} is $X=B$, because $PB=BQ=-PQ$. However, $B=-\idty$ on $\HH_{11}$, which could be infinite dimensional. In particular, it is infinite dimensional for strictly local walks but may be finite dimensional or even empty for essentially local ones.
This will then also be true for the polar isometry, and in conflict with the requirement that $V-\idty$ should be compact. For the purpose of contracting to the identity we would like to avoid eigenvalues $-1$ altogether. One option is to flip the sign on $\HH_{11}$ by an additional factor $(\idty-2P)$, so the next natural candidate, apparently first proposed by  T.\;Kato \cite{Kato} in another context (see also the discussion in \cite{AvronSeiler}), is
\begin{equation}\label{V=X}
X=(\idty-2P)B=B(\idty-2Q)=\idty-P-Q+2PQ.
\end{equation}
$X$ intertwines (i.e., $XQ=PX=PQ$) and is $+\idty$ on the subspace $\HH_{00}\oplus\HH_{11}$. In fact, $X-\idty=AQ-PA$ is compact, when $A$ is.
Moreover,
$X^*X=XX^*=B^2=(X+X^*)/2$, which implies that $X$ is normal with spectrum on the circle $(\im z)^2+(\re z-1/2)^2=1/4$ (see \Qref{fig:Xspec}).
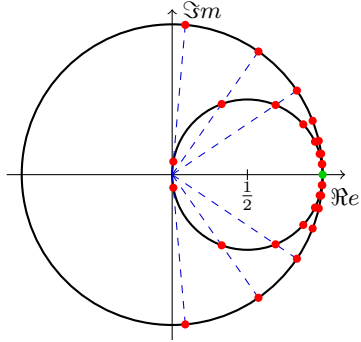
\begin{figure}

	\begin{tikzpicture}[scale=2]
	
	\colorlet{colorp}{red}		
	\colorlet{colorgap}{blue!80!black}
	\colorlet{colorcont}{green!80!black}
	\def\tph{0.02}
	\newcommand{\points}[1]{\draw[colorp,fill,thick] ({cos(#1)},{sin(#1)}) circle (\tph);
		\draw[colorp,fill,thick] ({cos(#1)*cos(#1)},{cos(#1)*sin(#1)}) circle (\tph);\draw[colorp,fill,thick] ({cos(-#1)},{sin(-#1)}) circle (\tph);
		\draw[colorp,fill,thick] ({cos(-#1)*cos(-#1)},{cos(-#1)*sin(-#1)}) circle (\tph)}
	\newcommand{\lines}[1]{\draw[colorgap,dashed] (0,0) -- ({cos(#1)-\tph*cos(#1)},{sin(#1)-\tph*sin(#1)});
		\draw[colorgap,dashed] (0,0) -- ({cos(#1)-\tph*cos(#1)},{-sin(#1)+\tph*sin(#1)})}

	\draw[->] (-1.1,0) -- (1.15,0);
	\draw[->] (0,-1.1) -- (0,1.1);
	\draw (0.5,1 pt) -- (0.5,.-1pt);
	\node at (0.5,-4pt) {$\tfrac{1}{2}$};
	\node at (1.15,-0.15) {$\re$};
	\node at (0.2,1.1) {$\im$};
	\draw[thick] (0,0) circle (1);
	\draw[thick] (0.5,0) circle (0.5);
	\foreach \ang in {85, 55, 34}
	\lines{\ang};
	\foreach \ang in {85, 55, 34, 21, 13, 8, 4}
	\points{\ang};

	\draw[colorcont,fill,thick] (1,0) circle (\tph);

	\end{tikzpicture}
	
	\caption{\label{fig:Xspec}
		Spectrum of $X=(\idty-2P)B$ (inner circle) and $\Vcan=(X^*X)^{-1/2}X$ (outer circle). Red points denote discrete spectrum and the green point at $1$ denotes the only accumulation point in both spectra. It is obvious, that, by construction, $\Vcan$ has no eigenvalues with negative real part.
	}
\end{figure}

This somewhat strange property  can be understood in terms of another symmetry, which will be needed below (\Qref{prop:swapH01}), so we describe it here. We introduce the operator
\begin{equation}\label{UX}
U=(\idty-2P)(\idty-2Q)=2X-\idty.
\end{equation}
Clearly, $U$ is unitary, so the second identity clarifies the spectral statement made above. $U$ implements the compact perturbation
\begin{equation}\label{Uperturb}
W_U=UW=(\idty-2P)\bigl(W(\idty-2P)W^*\bigr)W=(\idty-2P)W(\idty-2P).
\end{equation}
This is clearly not a decoupling, and $U$ is not intertwining.

This is not a coincidence. In fact, there is no \emph{unitary} element $V$ in the canonical algebra $\AA$ with the intertwining property: by definition, every operator in $\AA$ is reduced by the subspaces $\HH_{ij}$, but the intertwining condition requires $V\HH_{01}\subset\HH_{10}\oplus\HH_{11}$. Thus for every intertwiner $V\in\AA$ we must have $V\HH_{01}=0$, as we verified for the above $X$. The best we can do with a canonical approach is therefore to choose as our canonical decoupling operator $\Vcan$ the polar isometry of  $X$ from \eqref{V=X}, i.e.,
\begin{equation}\label{Vcanonical}
\Vcan=(X^*X)^{-1/2}X.
\end{equation}
Since $X$ is normal the order of these factors is irrelevant. Moreover, since the spectrum of $X$ is discrete with only $1$ as an accumulation point the eigenvalue $0$ is isolated, so on the complement of the null space $\HH_{01}\oplus\HH_{10}$ the operator $(X^*X)^{-1/2}$ is bounded, so $\Vcan-\idty$ is compact. Furthermore, $\Vcan$ has no eigenvalues with negative real part and in particular no $-1$-eigenvalues. This is best understood with the help of \Qref{fig:Xspec}. The eigenvalues of $X$ lie on a circle in the half-plane with positive real part. Multiplying $X$ with $(X^*X)^{-1/2}$ projects them onto the unit circle, keeping their real part positive.

We remark that $\Vcan$ is not actually in the universal C*-algebra $\AA$, because we used spectral information in the construction which is not true for arbitrary projection pairs $(P,Q)$. Given the additional information that $A$ is compact, however, we see that it is indeed an element of the C*-algebra generated by the concrete operators $P$ and $Q$.

Of course, for a proper decoupling we need a unitary operator satisfying \eqref{VQPV}. This will be of the form
\begin{eqnarray}\label{totalV}
V=\Vcan\oplus V_{01}\quad &\mbox{on}&\ \bigl(\HH_{00}\oplus\HH_{11}\oplus\HH_\perp\bigr)\oplus\bigl(\HH_{01}\oplus\HH_{10}\bigr)\\
&\mbox{where}&\strut\quad  V_{01}\HH_{01}=\HH_{10}\ \mbox{and\ }V_{01}\HH_{10}=\HH_{01}.
\end{eqnarray}
Here $\Vcan$ and $V_{01}$ are unitary on their respective supports.
The swapping condition for $V_{01}$ expresses the decoupling condition, because on $\HH_{01}\oplus\HH_{10}$ the projection onto the first summand is just $Q$ and $P$ projects onto the second.

\subsection{Decoupling under symmetry}
A decoupling $V$ is a special kind of compact perturbation, which therefore has to be admissible with respect to the symmetry representation $(\tph,\trv,\tch)$ on the space $\HH_V=\HH_\perp\oplus\HH_{01}\oplus\HH_{10}$ defined in \eqref{twiddlesym}. The projections $P$ and $Q=WPW^*$ then satisfy the relations
\begin{equation}\label{QPsymm}
\begin{array}{rlrl}
\tph Q&=Q\tph \quad, \   &\tph P&=P\tph\\
\trv Q&=P\trv,\    &\trv P&=Q\trv\\
\tch Q&=P\tch,\ &\tch P&=Q\tch,
\end{array}
\end{equation}
provided the respective symmetries are part of the type under consideration. This translates into the following symmetry properties of the subspaces:

\begin{prop}\label{prop:swapH01}\strut
	\begin{itemize}
		\item[(1)] The subspaces $\HH_{00},\HH_{11},\HH_\perp$, and $\HH_{01}\oplus\HH_{10}$ are each invariant under the symmetries $(\tph,\trv,\tch)$.
		\item[(2)] For $i,j=0,1$ and $i\neq j$, and any symmetry which is part of the type under consideration,
		\begin{equation}\label{swapH01}
		\tph\HH_{ij}=\HH_{ij},\quad  \trv\HH_{ij}=\HH_{ji}, \quad\mbox{and}\,\quad \tch\HH_{ij}=\HH_{ji},
		\end{equation}
		\item[(3)] On $\HH_{00}\oplus\HH_{11}\oplus\HH_\perp$ the canonical decoupling operator $\Vcan$ given by \eqref{Vcanonical} is admissible without eigenvalue $-1$.
		\item[(4)] The representation on the subspace  $\HH_{01}\oplus\HH_{10}$ is balanced.
	\end{itemize}
\end{prop}

\begin{proof}
	(1) and (2) follow immediately from \eqref{QPsymm}.
	
	For (3) combine \eqref{V=X} and \eqref{QPsymm} to find, for example, that $\trv X=(1-Q-P+2QP)\trv=X^*\trv$. The other admissibility conditions follow analogously.
	Then $X^*X$ commutes with all symmetries, and hence $\Vcan$ is admissible. It was already noted that $V+V^*$ is positive, so all eigenvalues have positive real part (compare \Qref{fig:Xspec}).
	
	(4)	Consider the perturbation $W_U=(\idty-2P)W(\idty-2P)=UW$ from \eqref{Uperturb}. Since the symmetries act sitewise, $P$ commutes with all symmetry operators. So clearly $W_U$ is again admissible. Hence, by
	\Qref{lem:twiddlesym}, $U$ is admissible for $(\tph,\trv,\tch)$. The $-1$-eigenspace of $U$ is precisely $\HH_{01}\oplus\HH_{10}$. Then by \Qref{thm:locpert} we obtain
	the index of the representation on this subspace as
	\begin{equation}\label{nullindexV}
	\tsix_-(U)=\six_-(W_U)-\six_-(W).
	\end{equation}
	Now the $-1$-eigenspaces of $W_U$ and $W$ are mapped to each other by $(\idty-2P)$, and thus carry unitarily equivalent symmetry representations. So the indices on the right hand side are equal and
	the index of $(\tph,\trv,\tch)$ on $\HH_{01}\oplus\HH_{10}$ is zero.
\end{proof}

It remains to establish the conditions for a decoupling on $\HH_{01}\oplus\HH_{10}$. Since we generally assumed an essential gap, \Qref{prop:bdf} implies that $\ind(W)=\dim\HH_{01}-\dim\HH_{10}=0$, so {\it some} decoupling exists. The same conclusion is arrived at even more simply via \eqref{swapH01} for all symmetry types containing either $\ch$ or $\rv$, which only excludes types \symA, \symD, and \symC. This leaves only one type which appears to require an additional condition for decoupling.

\begin{lem}\label{lem:H0110}
	Let $\HH_{01}$ and $\HH_{10}$ be Hilbert spaces of the same finite dimension, and let $(\tph,\trv,\tch)$ be a
	balanced symmetry representation on $\HH=\HH_{01}\oplus\HH_{10}$ of one of the types in \Qref{Tab:sym} such that the relations \eqref{swapH01} are satisfied.
	In the case of type \symAII{}, assume in addition that $\dim\HH_{01}$ is even.
	Then there is an admissible unitary operator $V$ with eigenvalues $\pm i$ such that $V\HH_{01}=\HH_{10}$ and $V\HH_{10}=\HH_{01}$.
\end{lem}

\begin{proof}
	In order to reduce the number of case distinctions, we may introduce a chiral symmetry $\tch$ if there is none (types \symA, \symD,\symC, \symAI, and \symAII). Thereby we merely introduce an additional admissibility constraint on $V$. Such a construction does {\it not} make any claims about a higher symmetry of the underlying walk;  the additional symmetry is confined to the abstract context of this Lemma.
	The new $\tch$ has to be unitary and swap the subspaces $\HH_{01}$ and $\HH_{10}$, and we will choose $\tch^2=\idty$ when possible. In addition, the symmetry representation has to remain balanced to fit the assumptions of the Lemma. This is trivially fulfilled, if we can choose $\tch^2=\idty$, since this, together with the swapping relation for $\HH_{01}$ and $\HH_{10}$, implies $\tr\tch=0$. For \symA, this is the only constraint, for \symD, we just need to choose $\tch$ real with respect to $\tph$. For \symC, choose a basis of the form $\chi^{ij}_1, \tph\chi^{ij}_1, \chi^{ij}_2, \tph\chi^{ij}_2,\dots.\tph\chi^{ij}_n\in\HH_{ij}$ for $i\neq j=0,1$, where $\dim\HH_{ij}=2n$.
	Then set $\tch\chi^{ij}_\mu=\chi^{ji}_\mu$, and extend $\tch$  by the convention that the symmetries commute. Obviously, this gives $\tch^2=\idty$. Finally, for type \symAI{} and \symAII{} choose a basis $\chi^{01}_\mu$ of $\HH_{01}$, and set $\tch\chi^{01}_\mu=\trv\chi^{01}_\mu$. This leads to a symmetry with $\tch^2=\trv^2$, i.e., of type \symBDI{} or \symDIII{}. In the case of \symAII{}, the new symmetry representation of type \symDIII{} is balanced, when we assume $\dim\HH_{01}$ to be even. This is a crucial assumption, as we will see below.
	
	Now consider the cases with $\tch^2=\idty$, and let $A=P-Q$ be the operator from \Qref{sec:2proj}. That is, $A$ is the operator with $1$-eigenspace $\HH_{10}$ and $-1$-eigenspace $\HH_{01}$. It satisfies $\tph A=A\tph$, $\trv A=-A\trv$, and $\tch A=-A\tch$.
	Then we set $V=A\tch$. Because $\tch$ is Hermitian, we have $V^*=\tch A=-V$, so $V^2=-\idty$. Moreover, one easily checks the admissibility conditions.
	The only cases left are the ones with three symmetries and $\tch^2=-\idty$, i.e., \symCI{} and \symDIII, for which we will choose bases to construct $V$ by hand. In both cases $\dim \HH_{01}$ is even, however, this happens for different reasons: for \symCI{} it follows from $\tph^2=-\idty$, since $\tph$ leaves $\HH_{01}$ invariant. In the case of \symDIII{} it is guaranteed by the fact that the symmetry representation is balanced by \Qref{prop:swapH01}, since this implies the dimension of $\HH_{01}\oplus\HH_{10}$ to be an integer multiple of four (see \Qref{sec:indexgroups}).

	For \symCI{}, choose a basis $\{\chi^{01}_\mu\}$ for $\HH_{01}$ in the same way as we did for \symC{}. This also gives a basis for $\HH_{10}$, by setting $\chi^{10}_\mu=\trv\chi^{01}_\mu$.
	An admissible $V$ with $V^2=-\idty$ is then given by $V\chi^{01}_\mu=\chi^{10}_\mu$ and $V\chi^{10}_\mu=-\chi^{01}_\mu$.

	For \symDIII{}, choose an $\tph$ invariant basis of $\HH_{01}$ and set $\chi^{10}=\trv\chi^{01}$. We then define $V\chi^{01}_{2\mu-1}=\chi^{10}_{2\mu}$ and $V\chi^{01}_{2\mu}=-\chi^{10}_{2\mu-1}$, and extend $V$ to $\chi^{10}_{\nu}$ to satisfy $V^2=-\idty$.
	
	Let us finally comment on the problem for \symAII{} if we waive the additional assumption for the dimension of $\HH_{01}$. The problem occurs whenever $\trv$ with $\trv^2=-\idty$ is present:
	$V$, as we construct it here, is a purely off-diagonal matrix with respect to the decomposition $\HH_{01}\oplus\HH_{10}$, thus, $\rank V_i=\dim\HH_{01}$ for the non-zero matrix blocks $V_i$. Now pick a basis $\{\chi^{01}_\mu\}$ for $\HH_{01}$ and set $\chi^{10}_\mu=\trv\chi^{01}_\mu$. The admissibility condition for $V$ then implies $V_i=-V_i^T$ with respect to this basis, which is only possible in even dimensions.
\end{proof}

We can now assemble the decoupling constructions.

\begin{thm}[Gentle Decoupling Theorem]\label{thm:gentdecoup}
	Let $W$ be a walk satisfying the conditions in \Qref{sec:firstsetting}. For symmetry type \symAII{} assume, in addition that the $+1$-eigenspace of $P-WPW^*$ is even dimensional.
	Then there is a continuous path  $t\mapsto W_t$ of admissible unitaries such that $W_0=W$, and $W_1$ commutes with $P$.
\end{thm}

\begin{proof}
	We define  $V=\Vcan\oplus V_{01}$ as in \eqref{totalV} with $\Vcan$ after \eqref{Vcanonical}, and $V_{01}$ from \Qref{lem:H0110}.
	On one hand, it satisfies the decoupling condition \eqref{VQPV}. On the other, it has no eigenvalue $-1$, and $V-\idty$ is compact. Hence, by \Qref{lem:contract} we can contract it to the identity, while keeping the admissibility condition.
\end{proof}

\section{Completeness of invariants}\label{sec:complete}
For any topological classification in terms of invariants the question of completeness arises. It depends crucially on the respective setting, both the set of objects to be classified and the transformations that generate the equivalence relation. For example, when we find that two translation invariant walks have the same indices we can ask whether one can be deformed into the other while keeping locality, symmetry and gap. But it also makes sense to ask whether the connecting path can even be chosen to be translation invariant. This question will be  addressed in
\cite{UsOnTI}.
One completeness result was already shown for the relative index (\Qref{lem:contract}). In that case only  the compact perturbations of some fixed walk $W_0$ were considered, with norm continuous perturbations respecting symmetry.

In this section we will settle the main completeness issue for the symmetry indices. There are three natural scenarios for this:
\begin{itemize}
\item[(\rom1)] All walks in the sense of \Qref{sec:firstsetting}, with respect to gentle perturbations (\Qref{def:pertsorts}).
   In this case we know the independent invariants $\sixR,\sixL,\six_-$, and the dependent invariant $\six_+=\sixL+\sixR-\six_-$.
\item[(\rom2)] All walks, with respect to both gentle and compact perturbations. In this case we know the invariants $\sixR,\sixL$.
\item[(\rom3)] All unitary operators, which have an essential gap and are admissible for the symmetry, but do not necessarily satisfy any locality condition, with respect to gentle perturbations.
           This is basically the setting of \Qref{sec:groups}.  In this case we know the invariants $\six_+,\six_-$.
\end{itemize}

\noindent The aim of this section is to prove the following:

\begin{thm}\label{thm:complete}
In each of the three scenarios described above, and for all symmetry types from \Qref{Tab:sym} other than \symA, \symAI, and \symAII, the indices indicated are complete.
Moreover, all index combinations can be realized by joining two translation invariant, strictly local walks with a finite crossover region.
\end{thm}

Let us first comment on the types \symA, \symAI, and \symAII, which we exclude here. These have the characteristic property that they do not contain a symmetry swapping the sign of the Hamiltonian or the imaginary part of the spectral values of a walk. Eigenvalues at $\pm1$ are therefore not protected, and the index group $\igS$ is zero. While in the other cases with trivial index group (\symC, \symCI) the Theorem makes the fairly trivial but correct statement that all such walks can be deformed into each other, for the \symA-types this is actually false. This is shown by the example $W=i\idty$, $W'=-i\idty$, which cannot be deformed into each other while respecting the essential gap. Let $n_+$ be the dimension of the eigenspace of $W$ for $\im z>0$, and $n_-$ for $\im z\leq0$. Then by either continuous or compact perturbations we can make any finite change in these cardinals. Then since $n_++n_-=\infty$, we have exactly three classes, characterized by $n_-<\infty$, $n_+<\infty$ and $n_-=n_+=\infty$. This is the description of all classes in the given three scenarios.

\begin{proof}[Proof of \Qref{thm:complete}:]\ The rest of this section is devoted to proving the Theorem.
We have to show that if $W'$ and $W$ have the same indices they can be linked by a path satisfying the respective conditions of each scenario. We will use this freedom to first simplify both walks.

\noindent{\it Step 1:}\/  In the settings  (\rom1) and (\rom2) a decoupling is helpful. Recall from \Qref{sec:twosettings} that for any decoupling $W_L\oplus W_R$ of a given walk $W$ we can consider the matrix of indices $\six_\pm(W_{L,R})$. What matrices can occur? Equivalently, when the walk is already given as $W=W_L\oplus W_R$, what changes can we make? To this end we introduce compact, but usually non-gentle perturbations
$W'_R$ and $W'_L$ with given relative indices $n_L=\sixrel W'_L:{W_L}$ and $n_R=\sixrel W'_R:{W_R}$, such that
\begin{equation}\label{inxmatplus}
   \sixmat{W'}=\sixmat W+ \left(\begin{array}{cc} -n_L&-n_R\\n_L&n_R\end{array}\right).
\end{equation}
In order to achieve any desired value of $n_L$ and $n_R$, for symmetry types with non-zero index groups, we need to identify a subrepresentation in each half-chain, on which to do the modification. This typically needs a high multiplicity for an irrep with a particular index. Fortunately, there is always an infinite supply of such irreps, because each half-chain is an infinite sum of balanced representations. Therefore, any choice $(n_L,n_R)$ is possible.

Now in scenario (\rom2) this perturbation is itself allowed, so $W'$ is as good a decoupling as $W$. In scenario (\rom1) $W'=W'_L\oplus W'_R$ is equivalent to $W=W_L\oplus W_R$ only if the overall perturbation is gentle, and then  \Qref{lem:contract} has the necessary and sufficient condition
$\sixrel W':W=n_L+n_R=0$. In either case we can prescribe any index matrix consistent with the indices of a given walk and find a decoupling with these indices.

For our task of connecting some walks $W',W$ with equal indices this means that we can choose both in decoupled form so that the indices of the respective left and right half-space walks still coincide.
This reduces the task to connecting half-space walks. Any connection in the sense of scenario (\rom3) will do, because decoupled walks are automatically essentially local, and thus the paths constructed in this way will also satisfy the locality condition. In this sense the locality constraint has disappeared from the problem, leaving only the symmetry and the essential gap condition, i.e., scenario (\rom3).

\noindent{\it Step 2:}\/  Assuming now setting (\rom3) we deform both walks further to simplify the spectrum.  We would need this procedure only in scenario (\rom3), but we state it a bit more generally.

\begin{lem}\label{lem:hcstep2}
In each of the settings {\rm(\rom1), (\rom2), and (\rom3)}, as well as in the purely translation invariant setting, every $W$ can be transformed within that setting to another one, for which
\begin{itemize}
\item[(1)] the spectrum is contained in $\{\pm1,\pm i\}$
\item[(2)] the eigenspaces at $\pm1$ contain no balanced subrepresentations, and are finite dimensional
\item[(3)] the combined eigenspace at $\pm i$ is infinite dimensional.
\end{itemize}
\end{lem}

\begin{proof}
(1) The main idea is to use the continuous functional calculus. The admissibility conditions for each symmetry (\Qref{sec:abstypes}) remain true when $W$ is replaced by a Laurent polynomial $p(W)$ in $W$ with real coefficients. By the Weierstra\ss\ Theorem, we conclude that the conditions are also preserved when $p$ is any continuous function on the unit circle such that $\overline{p(z)}=p(\overline{z})$. $p(W)$ is then defined in the functional calculus. Moreover, when $p$ maps the unit circle to itself, leaves $\pm1$ fixed, and satisfies $\abs{\im\,p(z)}\geq\abs{\im\,z}$ we get an admissible unitary $p(W)$ with gap at least as large as that of $W$. Clearly, we can find a norm continuous path $p_t$ connecting $p_0(z)=z$, and some function $p_1$ with $p_1(z)={\rm sign}(\im\,z)i$, whenever $\abs{\im\,z}>\veps$. Here $\veps>0$ is chosen sufficiently small so that no part of the spectrum other than the isolated eigenvalues $\pm1$ lies in the set with $\abs{\im\,z}\leq\veps$. Hence $t\mapsto p_t(W)$ continuously connects $W$ with a unitary $p_1(W)$ with the spectrum indicated.
It is an elementary property of the functional calculus that an operator commuting with $W$ will also commute with $p_t(W)$. Hence in the purely translational invariant setting the deformation will preserve translation invariance. Now suppose $W$ is essentially local, so that the commutator $[P,W]$ with the half-space projector $P$ is compact. This is the same as saying that in the Calkin algebra their images $\calk P$ and $\calk W$ commute. But then
$\calk{p_t(W)}=p_t\bigl(\calk W\bigr)$ so this commutation property is also conserved.

(2) If there is a balanced subrepresentation in one of the (finite dimensional) eigenspaces at $-1$ (resp.~$+1$), we can find a gapped unitary in that representation. All these are continuously connected to $-\idty$ (resp.~$+\idty$) on that eigenspace, by shifting the respective eigenvalues without touching the corresponding eigenspaces.

(3) This is a direct consequence from (1), (2) and the fact, that the essential spectrum of $W$ is non-empty.
\end{proof}

Note that \Qref{lem:hcstep2} (3) is true for each eigenspace at $\pm i$ individually, if we restrict our consideration to symmetry types that either contain $\ph$ or $\ch$, since these symmetries map the eigenspaces onto each other. So for the purpose of proving \Qref{thm:complete}, we have this slightly stronger statement.\\

\noindent{\it Step 3:}\/ Consider now a walk $W$  satisfying the conclusion of \Qref{lem:hcstep2}. Together with the symmetry operators it satisfies a rather simple set of algebraic rules: the relations among the symmetries, the admissibility conditions and $W^4=W^*W=WW^*=\idty$.  It is easy to determine all irreducible representations of these rules. First, there are the irreps with $W=\pm 1$, combined with an irrep of the symmetry, as listed in \Qref{Tab:sym}. Then there is a single irrep for the imaginary eigenspaces of $W$. Indeed, fix any eigenvector $W\phi=i\phi$, and choose it $\rv$-real, if $\rv^2=\idty$. Then act with all available symmetry operators on $\phi$, noting that this gives further eigenvectors of $W$. Thus one gets an invariant subspace, on which the action of all operators is fixed. Note that here we used the exclusion of the \symA-types, since otherwise we would have had to consider $W\phi=-i\phi$ separately.

It follows that the system of symmetries together with a $W$ is completely characterized up to unitary equivalence by the indices $\six_\pm(W)$. Indeed, as is evident from \Qref{Tab:sym}, for a completely unbalanced representation the index determines the representation up to unitary equivalence, and for the unique unbalanced representation we have infinite multiplicity. So considering two walks $W$ and $W'$ with the same index data we can find
a unitary operator $Z$ such that $ZWZ^*=W'$, and $Z\sigma =\sigma Z$ for any of the symmetry operators $\sigma$.

Suppose now that we can find a continuous function $[0,1]\ni t\mapsto Z_t$ with $Z_0=\idty$, $Z_1=Z$ so that $Z_t$ commutes with the symmetries for all $t$. Then  $t\mapsto W_t=Z_tWZ_t^*$ is a continuous path connecting $W$ to $W'$. The admissibility condition is satisfied for all $t$, and so is the essential gap condition, because all $W_t$ have the same spectrum.

\noindent{\it Step 4:}\/ The final step is to show that any symmetry-commuting unitary can be contracted to the identity. This is actually not true in finite dimension, and one can set up a kind of index theory along the lines of \Qref{sec:groups} to describe the connected components. The following Lemma shows that the infinite dimensional case is simpler. This is the only case we need, and hence the following concludes the proof of \Qref{thm:complete}.

\begin{lem}\label{lem:scconnect}
Let $\HH$ be a Hilbert space with a symmetry representation, which contains a direct sum of infinitely many balanced subrepresentations. Then the set of symmetry-commuting unitary operators is path connected in the norm topology.
\end{lem}

\begin{proof}
(1) Consider the spectral resolution of a symmetry-commuting $Z$, which we can write as
\begin{equation}\label{S-spectral}
  Z=-P_- + \int_{-\pi}^\pi\!\! E(d\alpha)\ e^{i\alpha},
\end{equation}
where $P_-$ is the projection onto the $-1$-eigenspace of $Z$ with a spectral measure $E$ on the complement of $P_-\HH$, which is normalized to $\idty-P_-$. The commutation with the symmetries means that $\ch$ commutes with each spectral projection, $\ph E(A)\ph^*=E(-A)$ for every measurable subset $A\subset(-\pi,\pi)$, and a similar condition for $\rv$. The family $Z_t=\int E(d\alpha)\exp(i(1-t)\alpha)$ for $t\in[0,1]$ is then a continuous symmetry-commuting contraction of $Z$ to $\idty-2P_-$. Hence we can assume $Z$ to be of this form

(2) Consider a Hilbert subspace $\HH_2\subset P_-\HH$, which can be decomposed into two subspaces on which the symmetries act in the same way, i.e., $\HH_2\cong\HH_1\oplus\HH_1$ with symmetry representation $\rho_2\cong\rho_1\oplus\rho_1$. Then with respect to this direct sum we define
\begin{equation}\label{Scontract}
  Z_t=\left(\begin{array}{cc} \cos(\pi t)\idty&-\sin(\pi t)\idty\\\sin(\pi t)\idty&\cos(\pi t)\idty\end{array}\right),
\end{equation}
where $\idty$ denotes the identity in $\HH_1$. This commutes with the symmetries, and agrees with $Z$ for $t=1$ (since $\HH_2\subset P_-\HH$) and with $\idty_\HH$ for $t=0$.

(3) Now consider some irrep of the symmetry. If it is contained in $P_-\HH$ with either infinite multiplicity or finite even multiplicity, we can use a process as in (2) to transfer the respective subrepresentation to the $+1$-eigenspace of $Z$. In the case of a finite odd multiplicity, observe that the $+1$-eigenspace must contain the representation with infinite multiplicity. Therefore, we can reverse the process (2) to make the multiplicity in $P_-\HH$ infinite. Then we apply the argument in the previous sentence.
\end{proof}

To conclude, let us give a short summary of the rather lengthy proof: in the first step, we reduced the scenarios ($\rom1$) and ($\rom2$) to ($\rom3$) via decoupling. Since any decoupled operator is essentially local, this allows us to ignore the locality condition.
In the next step we transformed the spectrum of the operators under consideration to $\{\pm1,\pm i\}$ without leaving scenario ($\rom3$).
We then finished the proof by showing that the class of operators obtained in this way is simply connected.
\end{proof}

\section{Finite systems}\label{sec:finite}

In this section, we discuss implications of our theory for spatially finite systems. Of course, real physical systems are finite, so a theory requiring the system to be infinite is strictly speaking empirically vacuous. On the other hand, the infinite system may be a convenient idealization without which the physical property under consideration could not be sharply defined. Prominent examples are the theory of phase transitions in Statistical Mechanics, and the characterization of propagation behaviour in terms of spectral properties by the RAGE Theorem \cite{Last}. In these cases the infinite system can be approached through a sequence of increasing finite ones, along which ``more of the same'' is added, which is an implicit appeal to translation invariance. Now we have emphasized that our theory does not require any translation invariance, so one may well wonder about the connection to a possible finite version.

To begin with, it is clear that without further structure our theory says very little about finite systems: the index $\six(W)$ on a finite set of cells vanishes identically, and hence does not even depend on $W$. The indices $\sixR(W)$ and $\sixL(W)$, which are defined also for essentially unitaries, are zero, because every $W$ can be contracted to the essentially unitary operator $0$. Only $\six_-(W)$ give some homotopy information. Even the translation invariant theory trivializes: when quasi-momentum becomes discretized, the winding number of a curve parametrized by quasi-momentum, or a Berry phase make no sense.

On the other hand, predictions of the theory like the topologically protected eigenvalues between distinct phases can be easily tested numerically, that is, in a finite system. There is a characteristic modification, however: the eigenvalues will now not appear exactly at $\pm1$, but very close to these values. One still sees some residual topological stability, as the phenomenon is independent on how the crossover between the different phases is constructed. Of course, it cannot be completely independent, since the ``crossover region'' could grow to be comparable to the size of the system. What counts in the end is that the crossover region is well padded with bulk regions. The {\dff bulk} systems must be described in a way that it makes sense to add more and more of the same. In the limit, the infinite system can then also be assigned indices in our theory. There are probably many ways to set up a suitable notion of bulk systems for which this vague description makes sense, i.e., conclusions about finite systems can be made. We will therefore introduce now the basic tool for such conclusions, and will come back to the suitable notions of bulk afterwards.

The mechanism for rigorous conclusions about boundary eigenvalues is described in the following Lemma, which is an adapted version of a Lemma due to Temple and Kato \cite{Simon1985}. While the original result gives a lower bound on the number of eigenvalues in an interval around an approximate eigenvalue and corresponding orthonormal approximate eigenvectors, we need a similar statement for only approximately orthogonal approximate eigenvectors.

\begin{lem}[Temple-Kato]\label{lem:TempleKato}
Let $U$ be a normal operator, $\theta\in\Cx$, and $\{\phi_\ell\}_{\ell=1}^K$ a set of vectors satisfying, for $\ell,k=1,\ldots,K$,
\begin{enumerate}
  \item $\abs{\braket{\phi_\ell}{\phi_k}-\delta_{\ell k}}\leq \veps_1<\frac1K$ and
  \item $\norm{\left(U-\theta\right)\phi_\ell}\leq \veps_2$.
\end{enumerate}
Then the spectral projection of $U$ for the disk around $\theta$ with radius $r$ has dimension at least $K$, provided
\begin{equation}\label{templeradius}
   r> \frac{K\veps_2}{\sqrt{1-K\veps_1}}.
\end{equation}
\end{lem}

\begin{proof}\def\dbas#1{\widetilde{\phi}_{#1}}
We first show that the set $\{\phi_\ell\}$ is linearly independent which is equivalent to the Gram matrix $G_{k\ell}= \braket{\phi_\ell}{\phi_k}$ being non-singular. To get a lower bound on this positive semidefinite matrix, let $\eta\in\Cx^K$ be a unit vector. Then using the estimates $G_{\ell\ell}\geq 1-\veps_1$ and $\abs{G_{k_\ell}}\leq\veps_1$ for $k\neq\ell$ we find
\begin{eqnarray}
  \braket\eta{G\eta}&=& \sum_\ell G_{\ell\ell}\abs{\eta_\ell}^2+ 2\sum_{k<\ell}\re\bigl( \overline{\eta_k}G_{k\ell}\eta_\ell \bigr)\nonumber\\
      &\geq&(1-\veps_1)-2\veps_1\sum_{k<\ell}\abs{\eta_k}\,\abs{\eta_\ell}
          =(1-\veps_1)-\veps_1\Bigl(\bigl(\sum_{k}\abs{\eta_k}\bigl)^2-\sum_k\abs{\eta_k}^2\Bigr)\nonumber\\
      &\geq& (1-K\veps_1),
\end{eqnarray}
where in the last step we used that the unit vector $\eta$ minimizing this expression has $\abs{\eta_\ell}=1/\sqrt K$. Hence $G$ is invertible with $\norm{G^{-1}}\leq (1-K\veps_1)^{-1}$.

Now consider the span of the $\phi_\ell$, $\KK$, and the same dual basis $\dbas\ell\in\KK$, which is defined by $\braket{\dbas\ell}{\phi_k}=\delta_{\ell k}$. One needs to compute the coefficients for expanding an arbitrary vector
$\psi$ in the $\phi_\ell$, summarized as $P_\KK=\sum_\ell\ketbra{\phi_\ell}{\dbas\ell}$. One readily verifies that $\dbas\ell=\sum_k(G^{-1})_{\ell k}\phi_k$ such that the Gram matrix of the dual basis is the inverse $\braket{\dbas\ell}{\dbas k}= (G^{-1})_{\ell k}$. In particular, $\norm{\dbas\ell}^2\leq\norm{G^{-1}}\leq (1-K\veps_1)^{-1}$. By this we can show that all $\psi\in\KK$ are nearly $\theta$-eigenvalues:
\begin{equation}\label{neartemple}
 \norm{(U-\theta)\psi}=\bigl\Vert\sum_\ell(U-\theta)\ket{\phi_\ell}\brAket{\dbas\ell}\psi \bigr\Vert
                     \leq \sum_\ell\norm{(U-\theta)\phi_\ell} \norm{\dbas\ell}\norm\psi
                     \leq K \veps_2 (1-K\veps_1)^{-1/2}\norm\psi
\end{equation}

Finally, suppose the eigenspace of $U$ for the disk of radius $r$ had dimension $<K$. Then we could find a vector $\psi\in\KK$ orthogonal to it. For such a vector we find with the spectral resolution $E$ of $U$:
\begin{equation}\label{notintemple}
  \norm{(U-\theta)\psi}^2=\int\braket\psi{E(dz)\psi}\abs{z-\theta}^2\geq r^2\int\braket\psi{E(dz)\psi}=r^2\norm{\psi}^2,
\end{equation}
where the inequality follows, because by assumption the spectral measure $\braket\psi{E(dz)\psi}$ vanishes on the disk $\abs{z-\theta}\leq r$. Combining \eqref{notintemple} with \eqref{neartemple} we get a contradiction, when $r$ is chosen as stated in the Lemma.
\end{proof}

In order to see how to apply this Lemma, suppose that we have two distinct bulk systems A and B with some crossover C, embedded as ACB inside a large but finite system. We compare this with a larger system with more of the bulk A and more of the bulk B added, and in the limit with a system with infinite bulks of type A$'$ and B$'$, and still the same crossover region A$'$CB$'$. Suppose that in the infinite system our theory predicts $K$ protected eigenvalues at $\theta=+1$. Then the corresponding orthonormal eigenfunctions $\psi_\ell$ will decay away from the crossover region. We truncate these to the original region ACB, adjust normalization and call these functions $\phi_\ell=P_{\rm ACB}\psi_\ell$, $\ell=1,\ldots k$. Now everything depends on the decay of the infinite volume eigenfunctions: if this is strong enough, the truncated functions remain almost orthonormal so the $\veps_1$ of the Lemma is small. By the same token (and because $U$ of the Lemma, in our case the walk $W$, is bounded), they will still be almost eigenfunctions, so $\veps_2$ is small. By the Lemma we can therefore conclude that there will be at least $K$ eigenvalues in a disk of a small radius $r$ around $1$. This will be true of {\it any} system containing the $ACB$ piece, and the eigenfunctions will be localized near $C$, just as they are in the infinite system.

\begin{figure}[t]
\begin{tikzpicture}
		\node[name=a] {\includegraphics{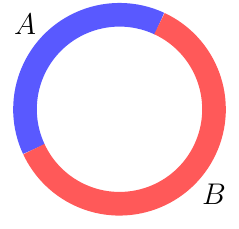}};
		\node[right of=a,node distance=6.5cm] (b) {\includegraphics{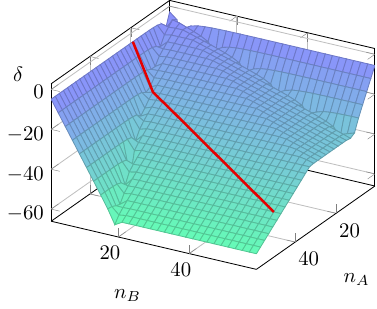}};
\end{tikzpicture}
	\caption{\label{fig:phase-circle}		
		Deviation of boundary eigenvalues from $\pm 1$ for a walk with two different bulk-phases of finite length on the circle.
		\emph{Left:} geometry of the setting. We consider the \emph{split-step-walk} $W=S_\downarrow R(\theta_2)S_\uparrow R(\theta_1)$ from \Qref{sec:splitss} on a circle with different parameters for regions $A$ and $B$. The parameters are $A:(\theta_1,\theta_2)=(-\pi/4,2\pi/16)$ with $\sixR(W_A)=-1$, $B:(\theta_1,\theta_2)=(\pi/4,3\pi/16)$ with $\sixR(W_B)=1$. The lengths of the regions $A$ and $B$ are denoted by $n_A$ and $n_B$, respectively.
		\emph{Right:} Plot of the deviation, measured by $\delta=\log\bigl(1-\max_i|\re(\lambda_i)|\bigr)$, where $\lambda_i$ are the eigenvalues of the combined walk $W_{AB}$, for different $n_A,n_B$. Red: A line parameterized as $(n_A,n_B)=(19+r,1+2r)$, $r=1,2,\ldots$. 
		}
\end{figure}

Note that we have left the notion of ``bulk'' rather vague here. Let us take this to mean ``translation invariant system'' for the moment. For the argument via \Qref{lem:TempleKato} to work we then need
\begin{enumerate}
\item A gapped system, so the infinite volume theory applies
\item Strict locality, or at least sufficient decay of the jump amplitudes in $W$ with distance, so that truncated eigenfunctions still give good approximate eigenfunctions.
\item Decay of the solutions of $W\psi=\pm\psi$, when read as a difference equation for $\psi:\Ir\to\HH_1$. Note that this is a linear difference equation, which is generally solved by exponential functions
      $\psi(x)=\exp(\pm\lambda x)\psi_0$. The smallest $\abs\lambda$ appearing here (for either eigenvalue) is called the inverse bulk correlation length $\lambda_0$.
\end{enumerate}
In a half-infinite bulk region only exponentially decaying eigenfunctions can occur, because they are the solution of a finite order difference equation. Hence the truncation to a finite length $L$ will introduce an error $\sim\exp(-\lambda_0L)$.  Hence if the size of all bulk pieces in a finite system is several correlation lengths of the respective infinite bulk system, the predictions of the infinite volume theory apply to every crossover region in the system. An example of this phenomenon is given in \Qref{fig:phase-circle}. It shows that the predicted eigenvalues are exponentially close to $\pm1$ over 60 orders of magnitude. Moreover, the decay rates depend on the adjoining phases, so that for a linear growth of both phases one gets a crossover behaviour.

It is clear that similar arguments apply to some less restrictive interpretations of ``bulk''. For example, we can allow disorder in the form of random local perturbations of a translation invariant system, which are norm small compared to the gap of the homogeneous system. The exponential decay of solutions is enhanced in this case by Anderson localization \cite{dynloc}. The critical part of the argument is the essential gap condition. In \Qref{sec:ssArbitrary} we considered arbitrarily varying angles. The boundaries for the allowed intervals (\Qref{fig:ssrectangles}) need to be in one and the same phase, for otherwise infinitely many approximate phase boundaries would be allowed, thus closing the essential gap. The same happens (with probability 1) for a disordered system, whenever the support of the probability measure for the angles $\theta_i(x)$ intersects two different phase regions.

\section{Conclusions}

We gave a complete homotopy characterization of essentially gapped one-dimensional quantum walks with discrete symmetries. In doing this we were careful to develop the most general setting in which these results live naturally. In particular, we made no assumptions of translation invariance, and found a locality condition, which is arguably the weakest sensible one, also in the translation invariant case.  The classifying index groups of the symmetry types were determined by an elementary group theoretical construction not using K-theory. The main new feature of the walk case compared to the Hamiltonian case is the distinction between gentle and non-gentle local (and, generally, compact) perturbations. This leads to the appearance of one additional invariant.

\section{Outlook}\label{sec:OverAndOut}
The following are possible directions for continuing the research reported here.
\begin{itemize}
\item {\it Higher lattice dimension}\\
  In the Hamiltonian case the heuristic literature also suggests a fairly detailed picture. The translation invariant case again leads to a classification of vector bundles over the Brillouin zone, and a now famous periodic table of index groups \cite{kitaevPeriodic}. At the interface of distinct bulk phases one now expects edge modes. However, a rigorous setting which handles arbitrary translation invariant walks of general symmetry types does not seem to exist, certainly not in the unitary case. Even the analog of \Qref{sec:old-index} seems to be open. A good direction might be the K-theory for non-commutative C*-algebras as in \cite{SchulzNonTechnicalOverview,Roe}.
\item {\it Interacting systems/QCAs}\\
  Walks correspond to a one-particle theory or, upon second quantization, to a theory of many non-interacting particles. So what happens when we turn on an interaction? Similarly, we may look at a situation where every cell may hold a particle, so the overall Hilbert space is a tensor product rather than a direct sum of the individual cells, i.e., we have quantum cellular automaton (QCA). The analog of the index theory described here in  \Qref{sec:old-index} has been established also for the QCA case \cite{OldIndex}, with positive rational rather than integer values of the index.
\item {\it Schur-function approach}\\
  In the course of this work we initially relied  strongly on the possibility to reduce the determination of eigenvalues in terms of a finite dimensional problem provided by the theory of (matrix valued) Schur functions \cite{QDapproach}. This method still has some detailed statements to offer, which we will explore in a future publication.
\end{itemize}

\section*{Acknowledgements}
C. Cedzich, T. Geib, C. Stahl and R. F. Werner acknowledge support from the ERC grant DQSIM, the DFG SFB 1227 DQmat, and the European project SIQS.

The work of L. Vel\'azquez is partially supported by the research project MTM2014-53963-P from the Ministry of Science and Innovation of Spain and the European Regional Development Fund (ERDF), and by Project E-64 of Diputaci\'on General de Arag\'on (Spain).

A. H. Werner thanks the Humboldt Foundation for its support with a Feodor Lynen Fellowship and the VILLUM FONDEN via the QMATH Centre of Excellence (Grant No. 10059).

\bibliography{tphbib}

\begin{thebibliography}{10}

\bibitem{dynloc}
A.~Ahlbrecht, V.~B. Scholz, and A.~H. Werner.
\newblock Disordered quantum walks in one lattice dimension.
\newblock {\em J. Math. Phys.}, 52(10):102201, 2011.
\newblock  \href{https://arxiv.org/abs/1101.2298}{{\ttfamily arXiv:1101.2298}}.

\bibitem{TRcoin}
A.~Ahlbrecht, H.~Vogts, A.~H. Werner, and R.~F. Werner.
\newblock Asymptotic evolution of quantum walks with random coin.
\newblock {\em J. Math. Phys.}, 52(4):042201, 2011.
\newblock  \href{https://arxiv.org/abs/1009.2019}{{\ttfamily arXiv:1009.2019}}.

\bibitem{Altland-Zirnbauer}
A.~Altland and M.~R. Zirnbauer.
\newblock Nonstandard symmetry classes in mesoscopic normal-superconducting
  hybrid structures.
\newblock {\em Phys. Rev. B}, 55(2):1142--1161, 1997.
\newblock  \href{https://arxiv.org/abs/cond-mat/9602137}{{\ttfamily
  arXiv:cond-mat/9602137}}.

\bibitem{Andruchow}
E.~Andruchow.
\newblock Pairs of projections: geodesics, {F}redholm and compact pairs.
\newblock {\em Complex Anal. Oper. Th.}, 8(7):1435--1453, 2014.

\bibitem{Asbo1}
J.~K. Asb\'oth.
\newblock Symmetries, topological phases, and bound states in the
  one-dimensional quantum walk.
\newblock {\em Phys. Rev. B}, 86(19):195414, 2012.
\newblock  \href{https://arxiv.org/abs/1208.2143}{{\ttfamily arXiv:1208.2143}}.

\bibitem{Asbo2}
J.~K. Asb\'oth and H.~Obuse.
\newblock Bulk-boundary correspondence for chiral symmetric quantum walks.
\newblock {\em Phys. Rev. B}, 88(12):121406, 2013.
\newblock  \href{https://arxiv.org/abs/1303.1199}{{\ttfamily arXiv:1303.1199}}.

\bibitem{AvronSeiler}
J.~Avron, R.~Seiler, and B.~Simon.
\newblock The index of a pair of projections.
\newblock {\em J. Funct. Anal.}, 120(1):220 -- 237, 1994.

\bibitem{TopSilberhorn}
S.~Barkhofen, T.~Nitsche, F.~Elster, L.~Lorz, A.~Gabris, I.~Jex, and
  C.~Silberhorn.
\newblock Measuring topological invariants and protected bound states in
  disordered discrete time quantum walks, 2016.
\newblock  \href{https://arxiv.org/abs/1606.00299}{{\ttfamily
  arXiv:1606.00299}}.

\bibitem{BDF}
L.~G. Brown, R.~G. Douglas, and P.~A. Fillmore.
\newblock Unitary equivalence modulo the compact operators and extensions of
  {C}*-algebras.
\newblock In {\em Proceedings: Dalhousie University, Halifax}, pages 58--128.
  Springer, 1973.

\bibitem{GawedzkIndex}
D.~Carpentier, P.~Delplace, M.~Fruchart, and K.~Gaw\k{e}dzki.
\newblock Topological index for periodically driven time-reversal invariant
  2{D} systems.
\newblock {\em Phys. Rev. Lett.}, 114(10):106806, 2015.
\newblock  \href{https://arxiv.org/abs/1407.7747}{{\ttfamily arXiv:1407.7747}}.

\bibitem{GawedzkIIndex}
D.~Carpentier, P.~Delplace, M.~Fruchart, K.~Gaw\k{e}dzki, and C.~Tauber.
\newblock Construction and properties of a topological index for periodically
  driven time-reversal invariant 2{D} crystals.
\newblock {\em Nucl. Phys. B}, 896:779 -- 834, 2015.
\newblock  \href{https://arxiv.org/abs/1503.04157}{{\ttfamily
  arXiv:1503.04157}}.

\bibitem{UsOnTI}
C.~Cedzich, T.~Geib, C.~Stahl, L.~Vel\'azquez, A.~H. Werner, and R.~F. Werner.
\newblock Complete homotopy invariants for translation invariant symmetric
  quantum walks on a chain.
\newblock in preparation.

\bibitem{letter}
C.~Cedzich, F.~A. Gr{\"u}nbaum, C.~Stahl, A.~H. Werner, and R.~F. Werner.
\newblock Bulk-edge correspondence of one-dimensional quantum walks.
\newblock {\em J. Phys. A: Math. Theor.}, 49(21):21LT01, 2016.
\newblock  \href{https://arxiv.org/abs/1502.02592}{{\ttfamily
  arXiv:1502.02592}}.

\bibitem{QDapproach}
C.~Cedzich, F.~A. Gr{\"u}nbaum, L.~Vel{\'a}zquez, A.~H. Werner, and R.~F.
  Werner.
\newblock A quantum dynamical approach to matrix {K}hrushchev's formulas.
\newblock {\em Comm. Pure Appl. Math.}, 69(5):909--957, 2016.
\newblock  \href{https://arxiv.org/abs/1405.0985}{{\ttfamily arXiv:1405.0985}}.

\bibitem{electric}
C.~Cedzich, T.~Ryb\'ar, A.~H. Werner, A.~Alberti, M.~Genske, and R.~F. Werner.
\newblock Propagation of quantum walks in electric fields.
\newblock {\em Phys. Rev. Lett.}, 111(16):160601, 2013.
\newblock  \href{https://arxiv.org/abs/1302.2081}{{\ttfamily arXiv:1302.2081}}.

\bibitem{MPSphaseII}
X.~Chen, Z.-C. Gu, and X.-G. Wen.
\newblock Classification of gapped symmetric phases in one-dimensional spin
  systems.
\newblock {\em Phys. Rev. B}, 83(3):035107, 2011.
\newblock  \href{https://arxiv.org/abs/1008.3745}{{\ttfamily arXiv:1008.3745}}.

\bibitem{sse}
C.Stahl.
\newblock Interactive tool at
  \url{https://qig.itp.uni-hannover.de/bulkedge/sse}.

\bibitem{Gensketal}
M.~Genske, W.~Alt, A.~Steffen, A.~H. Werner, R.~F. Werner, D.~Meschede, and
  A.~Alberti.
\newblock Electric quantum walks with individual atoms.
\newblock {\em Phys. Rev. Lett.}, 110(19):190601, 2013.
\newblock  \href{https://arxiv.org/abs/1302.2094}{{\ttfamily arXiv:1302.2094}}.

\bibitem{Graf}
G.~M. Graf and M.~Porta.
\newblock Bulk-edge correspondence for two-dimensional topological insulators.
\newblock {\em Commun. Math. Phys.}, 324(3):851--895, 2013.
\newblock  \href{https://arxiv.org/abs/1207.5989}{{\ttfamily arXiv:1207.5989}}.

\bibitem{Grimmet}
G.~Grimmett, S.~Janson, and P.~F. Scudo.
\newblock Weak limits for quantum random walks.
\newblock {\em Phys. Rev. E}, 69(2):026119, 2004.
\newblock  \href{https://arxiv.org/abs/quant-ph/0309135}{{\ttfamily
  arXiv:quant-ph/0309135}}.

\bibitem{OldIndex}
D.~Gross, V.~Nesme, H.~Vogts, and R.~F. Werner.
\newblock Index theory of one dimensional quantum walks and cellular automata.
\newblock {\em Commun. Math. Phys.}, 310(2):419--454, 2012.
\newblock  \href{https://arxiv.org/abs/0910.3675}{{\ttfamily arXiv:0910.3675}}.

\bibitem{Schulz2016index}
J.~Gro{\ss}mann and H.~Schulz-Baldes.
\newblock Index pairings in presence of symmetries with applications to
  topological insulators.
\newblock {\em Commun. Math. Phys.}, 343(2):477--513, 2016.
\newblock  \href{https://arxiv.org/abs/1503.04834}{{\ttfamily
  arXiv:1503.04834}}.

\bibitem{Halmos}
P.~R. Halmos.
\newblock Two subspaces.
\newblock {\em Trans. Amer. Math. Soc.}, 144:381--389, 1969.

\bibitem{HasanKaneReview}
M.~Hasan and C.~L. Kane.
\newblock \textit{Colloquium}: Topological insulators.
\newblock {\em Rev. Mod. Phys.}, 82(4):3045--3067, 2010.
\newblock  \href{https://arxiv.org/abs/1002.3895}{{\ttfamily arXiv:1002.3895}}.

\bibitem{dynlocalain}
A.~Joye.
\newblock Dynamical localization for d-dimensional random quantum walks.
\newblock {\em Quant. Inf. Process.}, 11(5):1251--1269, 2012.
\newblock  \href{https://arxiv.org/abs/1201.4759}{{\ttfamily arXiv:1201.4759}}.

\bibitem{KaneMeleTopOrder}
C.~L. Kane and E.~J. Mele.
\newblock {$\mathbb{Z}_{2}$} topological order and the quantum spin {H}all
  effect.
\newblock {\em Phys. Rev. Lett.}, 95(14):146802, 2005.
\newblock  \href{https://arxiv.org/abs/cond-mat/0506581}{{\ttfamily
  arXiv:cond-mat/0506581}}.

\bibitem{KaneMeleQSH}
C.~L. Kane and E.~J. Mele.
\newblock Quantum spin {H}all effect in graphene.
\newblock {\em Phys. Rev. Lett.}, 95(22):226801, 2005.
\newblock  \href{https://arxiv.org/abs/cond-mat/0411737}{{\ttfamily
  arXiv:cond-mat/0411737}}.

\bibitem{Karski:2009}
M.~Karski, L.~F{\"o}rster, J.~M. Choi, W.~Alt, A.~Widera, and D.~Meschede.
\newblock Nearest-neighbor detection of atoms in a {1D} optical lattice by
  fluorescence imaging.
\newblock {\em Phys. Rev. Lett.}, 102(5):053001, 2009.
\newblock  \href{https://arxiv.org/abs/0807.3894}{{\ttfamily arXiv:0807.3894}}.

\bibitem{Kato}
T.~Kato.
\newblock {\em Perturbation theory of linear operators}.
\newblock Springer, 1966/1984.

\bibitem{kitaevPeriodic}
A.~Kitaev.
\newblock Periodic table for topological insulators and superconductors.
\newblock {\em AIP Conf. Proc.}, 1134:22--30, 2009.
\newblock  \href{https://arxiv.org/abs/0901.2686}{{\ttfamily arXiv:0901.2686}}.

\bibitem{KitaevLectureNotes}
A.~Kitaev and C.~Laumann.
\newblock Topological phases and quantum computation.
\newblock In {\em Les Houches Summer School ``Exact methods in low-dimensional
  physics and quantum computing''}. Oxford University Press, 2010.
\newblock  \href{https://arxiv.org/abs/0904.2771}{{\ttfamily arXiv:0904.2771}}.

\bibitem{Kita2}
T.~Kitagawa.
\newblock Topological phenomena in quantum walks: elementary introduction to
  the physics of topological phases.
\newblock {\em Quant. Inf. Process.}, 11(5):1107--1148, 2012.
\newblock  \href{https://arxiv.org/abs/1112.1882}{{\ttfamily arXiv:1112.1882}}.

\bibitem{Kita}
T.~Kitagawa, M.~S. Rudner, E.~Berg, and E.~Demler.
\newblock Exploring topological phases with quantum walks.
\newblock {\em Phys. Rev. A}, 82(3):033429, 2010.
\newblock  \href{https://arxiv.org/abs/1003.1729}{{\ttfamily arXiv:1003.1729}}.

\bibitem{Last}
Y.~Last.
\newblock Quantum dynamics and decompositions of singular continuous spectra.
\newblock {\em J. Funct. Anal.}, 142(2):406--445, 1996.

\bibitem{nagy}
B.~Nagy, C.~Foias, H.~Bercovici, and L.~K{\'e}rchy.
\newblock {\em Harmonic analysis of operators on {H}ilbert space}.
\newblock Springer, 2010.

\bibitem{SchulzBaldesBook}
E.~Prodan and H.~Schulz-Baldes.
\newblock {\em Bulk and boundary invariants for complex topological insulators:
  from {K}-theory to physics}.
\newblock Mathematical Physics Studies. Springer, 2016.
\newblock  \href{https://arxiv.org/abs/1510.08744}{{\ttfamily
  arXiv:1510.08744}}.

\bibitem{ZhangTopologicalReview}
X.-L. Qi and S.-C. Zhang.
\newblock Topological insulators and superconductors.
\newblock {\em Rev. Mod. Phys.}, 83(4):1057, 2011.
\newblock  \href{https://arxiv.org/abs/1008.2026}{{\ttfamily arXiv:1008.2026}}.

\bibitem{Sinclair}
I.~Raeburn and A.~M. Sinclair.
\newblock The {C}*-algebra generated by two projections.
\newblock {\em Math. Scand.}, 65(2):278--290, 1989.

\bibitem{RW89}
G.~A. Raggio and R.~F. Werner.
\newblock Quantum statistical mechanics of general mean field systems.
\newblock {\em Helv. Phys. Acta}, 62(8):980--1003, 1989.

\bibitem{ReedSi4}
M.~Reed and B.~Simon.
\newblock {\em Methods of modern mathematical physics, vol.~\rom4}.
\newblock Academic Press, 1978.

\bibitem{Roe}
J.~Roe.
\newblock {\em Lectures on coarse geometry}.
\newblock AMS, 2008.

\bibitem{Schnyder2}
S.~Ryu, A.~P. Schnyder, A.~Furusaki, and A.~W. Ludwig.
\newblock Topological insulators and superconductors: tenfold way and
  dimensional hierarchy.
\newblock {\em New J. Phys.}, 12(6):065010, 2010.
\newblock  \href{https://arxiv.org/abs/0912.2157}{{\ttfamily arXiv:0912.2157}}.

\bibitem{Schnyder1}
A.~Schnyder, S.~Ryu, A.~Furusaki, and A.~Ludwig.
\newblock Classification of topological insulators and superconductors.
\newblock {\em AIP Conf. Proc.}, 1134:10--21, 2009.
\newblock  \href{https://arxiv.org/abs/0905.2029}{{\ttfamily arXiv:0905.2029}}.

\bibitem{Schreiber:2010cl}
A.~Schreiber, K.~N. Cassemiro, V.~Poto{\v c}ek, A.~G{\'a}bris, P.~J. Mosley,
  E.~Andersson, I.~Jex, and C.~Silberhorn.
\newblock Photons walking the line: a quantum walk with adjustable coin
  operations.
\newblock {\em Phys. Rev. Lett.}, 104(5):050502, 2010.
\newblock  \href{https://arxiv.org/abs/0910.2197}{{\ttfamily arXiv:0910.2197}}.

\bibitem{MPSphaseI}
N.~Schuch, D.~P\'erez-Garc\'ia, and I.~Cirac.
\newblock Classifying quantum phases using matrix product states and projected
  entangled pair states.
\newblock {\em Phys. Rev. B}, 84(16):165139, 2011.
\newblock  \href{https://arxiv.org/abs/1010.3732}{{\ttfamily arXiv:1010.3732}}.

\bibitem{SchulzZ2}
H.~Schulz-Baldes.
\newblock {$\mathbb{Z}_2$}-indices and factorization properties of odd
  symmetric {F}redholm operators.
\newblock {\em Doc. Math.}, 20:1481--1500, 2015.
\newblock  \href{https://arxiv.org/abs/1311.0379}{{\ttfamily arXiv:1311.0379}}.

\bibitem{SchulzNonTechnicalOverview}
H.~Schulz-Baldes.
\newblock Topological insulators from the perspective of non-commutative
  geometry and index theory.
\newblock {\em Jahresber. Deutsch. Math.-Verein.}, 118(4):247--273, 2016.
\newblock  \href{https://arxiv.org/abs/1607.04013}{{\ttfamily
  arXiv:1607.04013}}.

\bibitem{Simon1985}
B.~Simon and M.~Taylor.
\newblock Harmonic analysis on {SL(2,{$\mathbb R$})} and smoothness of the
  density of states in the one-dimensional {A}nderson model.
\newblock {\em Commun. Math. Phys.}, 101(1):1--19, 1985.

\bibitem{explorerA}
C.~Stahl.
\newblock Interactive Mathematica notebook at
  \url{http://qig.itp.uni-hannover.de/bulkedge}.

\bibitem{Asbo4}
B.~Tarasinski, J.~K. Asb\'oth, and J.~P. Dahlhaus.
\newblock Scattering theory of topological phases in discrete-time quantum
  walks.
\newblock {\em Phys. Rev. A}, 89(4):042327, 2014.
\newblock  \href{https://arxiv.org/abs/1401.2673}{{\ttfamily arXiv:1401.2673}}.

\bibitem{Thiang}
G.~C. Thiang.
\newblock On the {K}-theoretic classification of topological phases of matter.
\newblock {\em Ann. Inst. H. Poincar\'e Phys. Th\'eor.}, 17(4):757--794, 2016.
\newblock  \href{https://arxiv.org/abs/1406.7366}{{\ttfamily arXiv:1406.7366}}.

\bibitem{WoldvN}
J.~von Neumann.
\newblock Allgemeine {E}igenwerttheorie {H}ermitischer {F}unktionaloperatoren.
\newblock {\em Math. Ann.}, 102(1):49--131, 1929.

\bibitem{Wignerbook}
E.~P. Wigner.
\newblock {\em Group theory and its application to the quantum mechanics of
  atomic spectra}.
\newblock Academic Press, 1959.

\bibitem{Zumino}
B.~Zumino.
\newblock Normal forms of complex matrices.
\newblock {\em J. Math. Phys.}, 3(5):1055--1057, 1962.

\end{thebibliography}


\end{document}